\documentclass[draftclsnofoot,onecolumn]{IEEEtran}

\usepackage{cite}
\usepackage{amsmath,amssymb,amsfonts, pgffor}
\usepackage[ruled,vlined]{algorithm2e}
\usepackage[hidelinks]{hyperref}
\usepackage{graphicx}
\usepackage{amsthm}
\usepackage[table]{xcolor}
\usepackage{nicematrix}
\usepackage{multirow}
\usepackage{tikz}
\usetikzlibrary{matrix,positioning,arrows,shapes,chains,fit,scopes}

 \newtheorem{theorem}{Theorem}
 \newtheorem{corollary}[theorem]{Corollary}

 \newtheorem{lemma}[theorem]{Lemma}

\newtheorem{example}{Example}

\theoremstyle{definition}
 \newtheorem{definition}{Definition}
\theoremstyle{remark}
\newtheorem{remark}{Remark}

\newcommand{\mc}[1]{\mathcal #1}
\newcommand{\vv}[1]{\mathbf #1}

\newcommand{\ls}{\lambda}

\newcommand{\bspan}[1]{\langle #1 \rangle}
\newcommand{\numevnum}{300}
\newcommand{\cv}[1]{\mathbf{#1}}

\usetikzlibrary{matrix,decorations.pathreplacing, calc, positioning,fit}
\tikzset{arrow style mul/.style={circle,yshift=4pt,}}

\DeclareMathOperator*{\argmin}{arg\,min}

\newcounter{relctr} %
\everydisplay\expandafter{\the\everydisplay\setcounter{relctr}{0}} %

\AtBeginDocument{} %
\usepackage {xcolor}

\begin{document}

\title{On Optimal Finite-length Block Codes of Size Four for Binary Symmetric Channels}

\author{Yanyan Dong and Shenghao Yang%
  \thanks{This paper was presented in part at ISITA 2020 and ISIT 2023.}
  \thanks{Y.~Dong is with the Department of Electrical and Computer Engineering, National University of Singapore, Singapore 117597, Singapore; yan.dong@nus.edu.sg}
  \thanks{S.~Yang is with the School of Science and Engineering, The Chinese University of Hong Kong, Shenzhen, Shenzhen 518172, China; shyang@cuhk.edu.cn}
}

\maketitle

\begin{abstract}
A binary code of blocklength $n$ and codebook size $M$ is called an $(n,M)$ code, which is studied for memoryless binary symmetric channels (BSCs) with the maximum likelihood (ML) decoding.
  For any $n \geq 2$, some optimal codes among the \emph{linear} $(n,4)$ codes have been explicitly characterized in the previous study, but whether the optimal codes among the linear codes are better than all the nonlinear codes or not is unknown.
  In this paper, we first show that for any $n\geq 2$, there exists an optimal code (among all the $(n,4)$ codes) that is either linear or in a subset of nonlinear codes, called \emph{Class-I} codes. 
  We identified \emph{all} the optimal codes among the linear $(n,4)$ codes for each blocklength $n\geq 2$, and found ones that were not given in literature. For any $n$ from $2$ to $\numevnum$, all the optimal $(n,4)$ codes are identified, where  except for $n=3$, all the optimal $(n,4)$ codes are equivalent to linear codes. There exist optimal $(3,4)$ codes that are not equivalent to linear codes. 
Furthermore, we derive a subset of nonlinear codes called Class-II codes and justify that for any $n >\numevnum$, the set
composed of linear, Class-I and Class-II codes and their equivalent codes contains all the optimal $(n,4)$ codes.
Both Class-I and Class-II codes are close to linear codes in the sense that they involve only one type of columns that are not included in linear codes. 
Our results are obtained using a new technique to compare the ML decoding performance of two codes, featured by a partition of the entire range of the channel output.
\end{abstract}

\section{Introduction}
Shannon's channel capacity~\cite{shannon1948mathematical} is the
 maximum achievable rate in the sense that there exists
a code with an arbitrarily small error probability when the
blocklength is sufficiently large. Asymptotically capacity achieving channel codes
have been designed, e.g., polar codes \cite{Arikan09} and LDPC codes
\cite{gallager1962low,mitchell2015spatially}.  For practical
applications, codes of short blocklengths are preferred due
to lower latency and lower computation cost.
There have been analyses of the finite-length performance of practical
codes such as LDPC codes, polar codes, convolutional codes and BCH codes (see,
e.g.,
\cite{di2002finite,eslami2013finite,mondelli2014polar,olmos2015scaling,gaudio2017performance,cheng2021bch}).
In addition, bounds on the maximum channel coding rate achievable at a
given blocklength and error probability have been
investigated~\cite{valembois2004sphere,wiechman2008improved,polyanskiy2010channel}.
A classical question related to finite-length codes is the optimal
codes at a given blocklength and codebook size for memoryless binary
symmetric channels (BSCs) with respect to the maximum likelihood (ML)
decoding~\cite{Slepian1956A,fontaine1959group,wagner1966search,tokura1967search,cordaro1967optimum,peterson72}.

A binary code of blocklength $n$ and codebook size $M$ is called an
$(n,M)$ code, where $M\leq 2^n$. A $(n,2^k)$ code for a
certain integer $ k\leq n $ is said to be \emph{linear} if it is a
subspace of $\{0,1\}^n$. For given $n$ and $M$, it is a difficult
problem to find the optimal codes among all $(n,M)$ codes for
BSCs in terms of the ML
decoding. %
For many blocklengths $n$, all equidistant
codes that maximize the minimum Hamming distance   are strictly
suboptimal on a BSC~\cite{chen2013equidistant}. Though the ML decoding is equivalent to the minimum Hamming distance decoding, a code with the largest minimum Hamming distance among all $(n, M)$ codes is not necessarily optimal for ML decoding~\cite{chen2013optimal}.
 For given $ n $ and $ M $, it is hard to search the optimal
$ (n,M) $ codes by computers when $n$ is slightly
large~\cite{peterson72}.  In particular, the optimal codes among the
linear codes when the blocklength is small have been searched by
computer in
\cite{fontaine1959group,wagner1966search,tokura1967search}. If an optimal
code among the linear codes is perfect or quasi-perfect, it is optimal
among all the codes~\cite{Slepian1956A,peterson72}.  In general, it is
unknown whether an optimal code among the linear codes is optimal
among all the codes.  Except for codes that are perfect or
quasi-perfect, little is known about optimal codes for BSCs.

For BSCs, when the channel crossover
probability is small, the optimal code among the linear $(n,4)$ codes have been claimed for each block length $n$~\cite{cordaro1967optimum}.  Kl{\o}ve
\cite{klove2006binary} gave his conjecture of the generator matrices
for the optimal code among the $ (n,2^2) $ and $ (n,2^3) $ linear
codes.  Chen, Lin and Moser \cite{chen2013optimal} proved the optimality of a sequence of $(n,3)$  codes for $n=2,3,\ldots$, and they showed that a sequence of $ (n,4) $  linear codes for $n=2,3,\ldots$
formed by the conjectured optimal linear codes in~\cite{klove2006binary} are optimal among all the $ (n,4) $ 
linear codes. %
Vazquez-Vilar et al.~\cite{vazquez2016bayesian} compares
the optimal codes among all linear $(n,4)$ codes over a BSC proved
in~\cite{chen2013optimal} to the meta-converse lower bound.  For the binary erasure channels, the optimal codes were
found among all the $(n,M)$ codes satisfying $M\leq 4$
in~\cite{lin2018weak}.

In this paper, we study the optimal $(n,4)$ codes for BSCs with
respect to the ML decoding, considering both linear and nonlinear
codes. We say a property is \emph{universal} if it holds for any
crossover probability in the open set $(0,1/2)$. Similar as
in~\cite{cordaro1967optimum,chen2013optimal}, we use a matrix to
represent $(n,4)$ codes, where each codeword is a row of the
matrix, and we use the column types of the code matrix to present our
results.  We say two $(n,4)$ codes are \emph{equivalent} if one can be
obtained from the other by row or column interchanging and column
flipping. Two equivalent codes have the same ML decoding performance
for BSCs universally.  We obtain the following main results:
\begin{itemize}
\item For any blocklength $n\geq 2$, there exists an optimal $(n,4)$
  code that is either linear or in a subset of nonlinear codes, called
  the \emph{Class-I} codes. Class-I codes are close to linear codes
  in the sense that they involve only one type of columns that are not included in linear codes. 
  
\item All the optimal codes among the linear $(n,4)$ codes are
  identified for each given blocklength $n\geq 2$, and are
  universal. We obtain more optimal codes among the linear $ (n,4) $ codes
  than those shown in \cite{chen2013optimal} when $n=3k$ or $3k+1$ for a positive integer $k$. When
  $n\neq 3$, our results are consistent with the optimal codes among
  the linear codes claimed by Cordaro and Wagner for small values of
  the crossover probability in \cite{cordaro1967optimum}. When $n=3$,
  we found more optimal codes, one of which has an all-zero column.
  
\item For any blocklength $n$ from $2$ to $\numevnum$, all the optimal $(n,4)$ codes are characterized, and are universal. Except for $n=3$, all the optimal $(n,4)$ codes are equivalent to linear codes. There exist optimal $(3,4)$ codes that are not equivalent to linear codes. 
  
\item For any blocklength $n > 300$, the set composed of
  linear, Class-I and Class-II codes and their equivalent codes
  contains all the optimal codes. Class-II codes are close to linear codes
  in the same sense as Class-I codes.
\end{itemize}

This paper also moves forward the analytical techniques of binary block codes.
In~\cite{chen2013optimal}, two linear 
$(n,4)$ codes with one column different are compared. 
We derive a new technique to compare the ML decoding performance of two $(n,4)$ codes $C$ and $C'$ with differences in some columns, where $C$ and $C'$ are not necessarily to be linear.
Our technique can induce a strict partial order relation among $(n,4)$ codes, which is crucial for characterizing all the optimal codes. Therefore, even for linear $(n,4)$ codes, our technique can induce stronger results than those in~\cite{chen2013optimal}.
Our technique starts with a specific partition of the entire range of the channel output, i.e., $\{0,1\}^n$, and a permutation $g$ on the range, such that for each subset $\mc Y$ in the partition, one of the following three conditions holds:
\begin{enumerate}
\item For all $\vv y \in \mc Y$, the decoding performance of $\vv y$ for $C$ is the same as the decoding performance of $g(\vv y)$ for $C'$;
\item For all $\vv y \in \mc Y$, the decoding performance of $\vv y$ for $C$ is better than the decoding performance of $g(\vv y)$ for $C'$;
\item For all $\vv y \in \mc Y$, the decoding performance of $\vv y$ for $C$ is worse than the decoding performance of $g(\vv y)$ for $C'$.
\end{enumerate}
See Fig.~\ref{fig:cc} for an illustration of a partition with three subsets.
With such a partition, we only need to focus on the subsets satisfying conditions 2) and 3) for the decoding performance comparison. To make the problem simpler, we hope that the partition has a small number of subsets.
We find such partitions of $5$ subsets for the two cases we use: i) two codes with differences in one column, and ii) two codes with differences in two columns. Based on the code comparison results of these two cases, our main results about the optimal codes can be proved without further involving the technical details of comparing two codes. 

\begin{figure}
  \centering
  \begin{tikzpicture}[scale=0.5]
    \draw[thick] (0,0)--(0,4)--(4,4) node[above,midway] {$\{0,1\}^n$}--(4,0)--cycle;
    \draw[->] (4.5,2) -- (6.5,2) node[above,midway] {$g$};

    \draw (2,2) -- (0,3);
    \draw (2,2) -- (1,0);
    \draw (2,2) -- (4,2.5);
    \node at (2,3) {$\mc Y_1$};
    \node at (0.8,1.5) {$\mc Y_2$};
    \node at (3,1.3) {$\mc Y_3$};

    \begin{scope}[xshift=7cm]
    \draw[thick] (0,0)--(0,4)--(4,4) node[above,midway] {$\{0,1\}^n$}--(4,0)--cycle;
    \draw (2,2) -- (0,1.6);
    \draw (2,2) -- (2,4);
    \draw (2,2) -- (4,0);
    \node at (1.7,1) {$\mc Y_1'$};
    \node at (1.2,3) {$\mc Y_2'$};
    \node at (3,2.3) {$\mc Y_3'$};
  \end{scope}
  \end{tikzpicture}
  \caption{Illustration of our technique to compare two codes $C$ and $C'$. $\{\mc Y_1, \mc Y_2, \mc Y_3\}$ forms a partition of $\{0,1\}^n$. In the figure, $\mc Y_i' = g(\mc Y_i)$ for $i=1,2,3$. For all $\vv y \in \mc Y_1$, the decoding performance of $\vv y$ for $C$ is the same as the decoding performance of $g(\vv y)$ for $C'$. For all $\vv y \in \mc Y_2$, the decoding performance of $\vv y$ for $C$ is better than the decoding performance of $g(\vv y)$ for $C'$. For all $\vv y \in \mc Y_3$, the decoding performance of $\vv y$ for $C$ is worse than the decoding performance of $g(\vv y)$ for $C'$. To compare $C$ and $C'$, we only need to study $\mc Y_2$ and $\mc Y_3$.}
  \label{fig:cc}
\end{figure}

In the remainder of this paper, we first formulate the problem and introduce our main results in \S\ref{sec:formulationresults}. Then \S\ref{sec:approach} presents the major code comparison results and elaborates how they imply the main results in  \S\ref{sec:formulationresults}. The proofs of the major code comparison results are in the following sections.  
In \S\ref{sec:formulation}, we formally introduce the general approach for comparing the ML decoding performance of two codes and discuss a special case in detail, where two codes only differ in one column (see \S\ref{sec:1col}). The comparison of codes with two columns  different is provided in \S\ref{sec:fliptwobits}.  \S\ref{sec:linearcode} compares the performance between two linear codes with only one column different, and  \S\ref{sec:classI} is dedicated to the analysis of Class-I codes, both based on the results in \S\ref{sec:1col}. In \S\ref{sec:remark}, we discuss the open problems and the directions for future research. 

\section{Problem Formulation and Main Results}\label{sec:formulationresults}

\subsection{Formulation of $(n,M)$ Binary Codes}

For positive integers $M$ and $n$ with $M\leq 2^n$, 
an $(n,M)$ binary node $\mc C$ is a subset of $\{0,1\}^n$ of size $M$, and an $(n,2^k)$ code with integer $k\leq n$ is said to be \emph{linear} if it is a subspace of $\{0,1\}^n$. Using the codewords of $\mc C$ as rows, we can form an $M\times n$ binary matrix $C$, which is used interchangeably with $\mc C$. 
For $i=1,\ldots, M$, let $\vv c_i$ be the $i$th row of $C$, i.e., the $i$th codeword of $\mc C$.

For $\vv x, \vv y\in \{0,1\}^n$, let $w(\vv x)$ be the \emph{Hamming weight} of $\vv x$ and let
$\vv x\oplus \vv y$ be the bit-wise exclusive OR of $\vv x$ and $\vv y$. Hence, $w(\vv x\oplus \vv y)$ is the Hamming distance between $\vv x$ and $\vv y$.
The minimum distance of $y\in \{0,1\}^n$ with the code $C$ is denoted as
\begin{equation}\label{eq:dcy}
  d_{C}(\vv y) = \min_{\vv c \in \mc C} w(\vv c \oplus \vv y).
\end{equation}

We use BSC$(\epsilon) $ to denote the symmetric channel (BSC) with crossover probability $\epsilon$ ($0<\epsilon< \frac{1}{2}$). 
Suppose the code $C$ is used for BSC$(\epsilon)$. For a channel input $\vv{x} \in \{0,1\}^n$, the channel output is $\vv y \in \{0,1\}^n$ with probability
\begin{equation*}%
  p(\vv y|\vv x) = (1-\epsilon)^{n-w(\vv x \oplus \vv y)}\epsilon^{w(\vv x\oplus \vv y)}.
\end{equation*}
The maximum-likelihood (ML) decoding rule decodes an output $\vv y$ to a codeword $\vv c$ if
$w(\vv c\oplus \vv y) = d_{C}(\vv y)$,
where a tie is resolved arbitrarily. 
Define
\begin{equation*}
  \alpha_{C}(d)=|\{\vv y\in \{0,1\}^n: d_C(\vv y)=d \}|,
\end{equation*}
which is the number of outputs $\vv y$ that is decoded to a codeword of distance $d$.
Note that the value $\alpha_{C}(d)$ does not depend on $\epsilon$.
The ({average}) correct decoding probability of $C$ is
\begin{IEEEeqnarray}{rCl}
  \ls_{C} (\epsilon) & = &
  \frac{1}{|C|}\sum_{\vv y \in \{0,1\}^n}  (1-\epsilon)^{n-d_{\cv C}(\vv y)}\epsilon^{d_{\cv C}(\vv y)} \label{eq:lambda:1}\\
  & = & \frac{1}{|C|} \sum_{d=0}^n \alpha_{C}(d) (1-\epsilon)^{n-d}\epsilon^{d}. \label{eq:lambda}
\end{IEEEeqnarray}
\begin{definition}
  For BSC$(\epsilon)$, $0<\epsilon<\frac{1}{2}$, the following relations are defined between two $(n,M)$ codes $C$ and $C'$:
  \begin{enumerate}
  \item Code $C$ is \emph{better (resp. strictly better)} than $C'$ for the given crossover probability $ \epsilon $ if $\ls_{C}(\epsilon)\geq \ls_{C'}(\epsilon)$  (resp. $\ls_{C}(\epsilon)> \ls_{C'}(\epsilon)$);
\item Code $C$ is \emph{an optimal code  for the given crossover probability $ \epsilon $} if it is better than any other $(n,M)$ codes;
\item When $M=2^k$, code $C$ is \emph{optimal among linear codes for the given  crossover probability $ \epsilon $} if it is better than any other $(n,2^k)$ \emph{linear} codes;
\item If valid for all crossover probability $\epsilon\in (0,\frac{1}{2})$, a property of a code is said to be \emph{universal}. In particular, we write $ \ls_{C'} \geq \ls_C $ (resp. $ \ls_{C'} > \ls_C $) if $ C' $ is universally better (resp. universally strictly better) than $ C; $
\end{enumerate}
\end{definition}
\begin{remark}
	 In this paper, if the crossover probability $ \epsilon $ is not mentioned for a property, the property discussed are implied to be universal.
\end{remark}
\subsection{Formulation of $(n,4)$ codes}\label{sec:n4f}

In this paper, we focus on $(n,4)$ codes with $n\geq 2$,  which have four codewords. The columns of an $(n,4)$ code $C$ are of vectors in $\{0,1\}^{4}$. 
We use $\bspan{i}$ to denote the binary vector of length $4$ associated with an integer $i= 0,1,\ldots,15$. For example,  %
\begin{equation*}
  \bspan{1} =
  \begin{pmatrix}
    0 \\ 0 \\ 0 \\ 1
  \end{pmatrix}, \quad 
  \bspan{2} =
  \begin{pmatrix}
    0 \\ 0 \\ 1 \\ 0
  \end{pmatrix}.
\end{equation*}
W use $|i|_C$ to denote the number of columns of $C$ equal to $\bspan{i}$.
We may write $|i|_C$ as $|i|$ when the code $C$ is implied in the context.
For example, the $(8,4)$ code
\begin{equation*}
  C = 
  \begin{pNiceMatrix}[last-row]
  0&  0 & 0 & 0 & 0 & 0 & 0 & 0\\
    0&0 & 0 & 0 & 1 & 1 & 1 & 1\\
    0&0 & 1 & 1 & 0 & 0 & 1 & 1\\
    0&1 & 0 & 1 & 0 & 1 & 0 & 1\\
  \bspan{0}&  \bspan{1} & \bspan{2} & \bspan{3} & \bspan{4} & \bspan{5} & \bspan{6} & \bspan{7}
  \end{pNiceMatrix}
\end{equation*}
has the $i$th column of type $\bspan{i}$ and $|i|=1$ for $i=0,1,\ldots,7$.

The column types of $C$ has been used in
literature~\cite{cordaro1967optimum,chen2013optimal}. For example,
Chen, Lin and Moser~\cite{chen2013optimal} compared different codes by
induction in $n$, i.e., increasing one column a time, for studying
$(n,3)$ codes and linear $(n,4)$ codes. In this paper, we develop a
new technique to compare two $(n,4)$ codes with difference in some rows of one or
two columns. The following facts about $(n,4)$
codes are
straightforward~\cite{cordaro1967optimum,chen2013optimal}. First, flipping all the bits in a column does not change the decoding
performance.
 Second, codes with only row interchanging have the same ML decoding performance.
 Third,  column permutations of $C$ do not affect
the decoding performance.  Due to these facts, we define the following
equivalence relation to simplify our discussion.

\begin{definition}
 We say two $ (n,4) $ codes $ C $ and $C' $ are \emph{equivalent} if $ C' $ can be obtained by interchanging rows, interchanging columns and flipping all the bits in some  columns of $C $.
\end{definition}

Two equivalent codes have the same decoding performance universally.  A linear
code is equivalent to its coset codes. By column flipping, a code with a column $\bspan{i}$, $8\leq i \leq 15$, is equivalent to the code obtained by changing the column $\bspan{i}$ to $\bspan{15-i}$. Therefore, a code is optimal if it is optimal among all the codes with columns only from $\bspan{0},\bspan{1},\ldots,\bspan{7}$.
Among codes with only these eight types $\bspan{0},\bspan{1},\ldots,\bspan{7}$,
the linear codes have only the column types $\bspan{0}$, $\bspan{3}$, $\bspan{5}$ or $\bspan{6}$:
\begin{equation*}
\bspan{0}  = 
\begin{pmatrix}
  0 \\ 0 \\ 0 \\ 0
\end{pmatrix},\quad
\bspan{3} =
\begin{pmatrix}
0 \\0\\1\\1
\end{pmatrix}, \quad 
\bspan{5} =
\begin{pmatrix}
0 \\1\\0\\1
\end{pmatrix}, \quad 
\bspan{6} =
\begin{pmatrix}
0\\1\\1\\0
\end{pmatrix}.
\end{equation*}
We call $\bspan{3}$, $\bspan{5}$ and $\bspan{6}$ the \emph{linear types}. Codes with only the linear type columns are also called the weak flip codes in \cite{chen2013optimal}. 
\begin{definition}\label{def:linear}
  We use $ C(n_3,n_5,n_6)$ to represent an $(n,4)$ linear code $ C$ with $|3|_ C = n_3$, $|5|_C=n_5$, $|6|_ C=n_6$ and $n_3+n_5+n_6=n$.
\end{definition}
An $(n,4)$ linear code must have at least two distinct linear types   since otherwise, the four codewords cannot form a subspace of dimension four and thus is not linear.
We define the following three classes of nonlinear codes to better present our results.
\begin{definition}\label{def:class1}
 An $(n,4)$ code $C$ is said to be \emph{Class-I} if $|1|$ is odd, $|3|,|5|,|6|$ are of the same parity, and $|1|+|3|+|5|+|6|=n$.
\end{definition}
\begin{definition}\label{def:class2}
	An $ (n,4) $ code is said to be \emph{Class-II} if $|1|>0$, $|1|+|3|+|5|+|6|=n$ and satisfies one of the following conditions:
	\begin{enumerate}
		\item[a)] $ |1| $ and $ |3| $ are even and $ |5|$ and $|6| $ are odd;
		\item[b)] $ |1|,\ |5|$ and $ |6| $ are even and $ |3| $ is odd.
	\end{enumerate}
\end{definition}

\begin{table}[t]
	\centering
	\caption{Comparison of linear, Class-I and Class-II  codes. Note that for a linear code, at least two of $ |3| $, $ |5| $ and $ |6| $ should be positive.}\label{tab:class123}

        \begin{tabular}{ccccccccccc}
          \hline
          \rowcolor[HTML]{EFEFEF} 
          class & subclass  & blocklength &$|0|$& $|1|$  & $|2|$  & $|3|$    & $|4|$  & $|5|$     & $|6|$     & $|7|$                        \\ \hline
          linear	&          &  &$\geq 0$    & $ 0 $     & $0$    & $ \geq 0 $ & $0$  & $ \geq 0 $    & $ \geq 0 $    & $0$  \\\hline
               &\cellcolor[HTML]{EFEFEF}a         & \cellcolor[HTML]{EFEFEF}odd               &\cellcolor[HTML]{EFEFEF}$0$      & \cellcolor[HTML]{EFEFEF}odd     & \cellcolor[HTML]{EFEFEF}$0$    & \cellcolor[HTML]{EFEFEF}even   & \cellcolor[HTML]{EFEFEF}$0      $     & \cellcolor[HTML]{EFEFEF}even    & \cellcolor[HTML]{EFEFEF}even    &\cellcolor[HTML]{EFEFEF}$0$  \\ \cline{2-11} 
          
          \multirow{-2}{*}{I}  & b              & even     &$0$  & odd                         & $0$                        & odd                        & $0$                        & odd                         & odd                         & $0$         \\ \hline
               &\cellcolor[HTML]{EFEFEF}a        & \cellcolor[HTML]{EFEFEF}even        &\cellcolor[HTML]{EFEFEF}$0$   &\cellcolor[HTML]{EFEFEF}even, $>0$    & \cellcolor[HTML]{EFEFEF}$0$    & \cellcolor[HTML]{EFEFEF}even   & \cellcolor[HTML]{EFEFEF}$0$    & \cellcolor[HTML]{EFEFEF}odd     & \cellcolor[HTML]{EFEFEF}odd     &\cellcolor[HTML]{EFEFEF}$ 0$           
          \\ \cline{2-11} 	\multirow{-2}{*}{II} 
               &b & odd    &$0$   &  even, $>0$ &  $0$ &  odd &  $0$ & even &  even &  $0$        \\ \hline
        \end{tabular}
	\end{table}

We compare linear, Class-I and Class-II $(n,4)$ codes in Table \ref{tab:class123}, where we name  each case of Class-I and Class-II codes as a subclass. Among the operations that preserve the code equivalence, column interchanging does not change $|i|$, but both row interchanging and column flipping affect $|i|$. For two equivalent $(n,4)$ codes $C$ and $C'$  with only columns in $\{\bspan{0},\bspan{1},\ldots,\bspan{7}\}$, we always have
\begin{IEEEeqnarray*}{rCl}
  |0|_C & = & |0|_{C'},  \\
  \{|1|_C,|2|_C,|4|_C,|7|_C\} & = & \{|1|_{C'},|2|_{C'},|4|_{C'},|7|_{C'}\},\\
  \{|3|_C,|5|_C,|6|_C\} & = & \{|3|_{C'},|5|_{C'},|6|_{C'}\}.
\end{IEEEeqnarray*}
In other words, we cannot change a linear type to a nonlinear type while preserving the equivalence.
Hence, there are no equivalent codes belonging to two different classes (or subclasses) among linear, Class-I and Class-II codes. 

\subsection{Main Results about Optimal $(n,4)$ Codes}\label{sec:mainresults}
We give the main results about the optimal $(n,4)$ codes.  Firstly, we derive a relatively small set that contains an optimal code in the following theorem.
\begin{theorem}[Optimal Code Inclusive Set]\label{thm:two}
	For any BSC$(\epsilon)$, $0<\epsilon<\frac{1}{2}$, there exists an optimal $(n,4)$ code in the set formed by all the linear codes and Class-I codes.
\end{theorem}
Note that in Theorem~\ref{thm:two}, it is unknown whether the optimal codes for different crossover probability $ \epsilon $ are the same or not.

By comparing two linear codes with only one column different, %
 the following theorem presents  all the  linear codes  which are universally strictly better than any other non-equivalent linear codes.

 \begin{theorem}[All Optimal Linear Codes]\label{thm:linearopt}
   Considering linear $(n,4)$ codes for BSCs, the following properties are universally satisfied: %
   \begin{enumerate} 
   \item if  $n =  3k-1$ for some $k\in \mathbb N^+$,  $C(k,k,k-1)$ is  optimal among all linear codes;
   \item if  $n =  3k$ for some $k\in \mathbb N^+$. 
     When $k\geq 2$,  $C(k+1,k+1,k-2)$ and $C(k+1,k,k-1)$ are  optimal among all linear codes.
     When $ k=1 $, i.e. $ n=3 $, $C_A$,
     $C(1,1,1)$ and $C(1,2,0)$  are  optimal among all linear codes, where 
     \begin{equation}\label{eq:0columnlinear}
     	C_A \triangleq \begin{pmatrix}
     		 0 & 0 & 0\\
     		0 & 1 & 0\\
     		1 & 0 & 0\\
     		1 & 1 & 0
     	\end{pmatrix};
     \end{equation}
   \item if  $n =  3k+1$ for some $k\in \mathbb N^+$,  $C(k+1,k,k)$ and $C(k+2,k,k-1)$ are    optimal among all linear codes.

   \end{enumerate}
   Moreover, these $(n,4)$ linear codes are universally strictly better than any other $(n,4)$ linear codes that are not equivalent to them.
 \end{theorem}

 \begin{table}[t]
   \centering
   \caption{All the Optimal Linear $(n,4)$ Codes, where $C(i,j,k)$ is defined in Definition~\ref{def:linear} and $C_A$ is given in \eqref{eq:0columnlinear}.}\label{table:optimal-linear}
   \begin{tabular}{c|c}
     \hline
     \rowcolor[HTML]{EFEFEF} 
     $n$      & optimal linear $(n,4)$ codes                               \\ \hline
     $3k-1$, $k \geq 1$ & $C(k,k,k-1)$                        \\ \hline
     \rowcolor[HTML]{EFEFEF} 
     $3$   & $C_A$, $C(1,1,1)$ and $C(1,2,0)$           \\ \hline
     $3k$, $k\geq 2$    & $C(k+1,k+1,k-2)$ and $C(k+1,k,k-1)$ \\ \hline
     \rowcolor[HTML]{EFEFEF} 
		$3k+1$, $k \geq 1$ & $C(k+1,k,k)$ and $C(k+2,k,k-1)$     \\ \hline
	\end{tabular}
      \end{table}

In Table \ref{table:optimal-linear}, the representative optimal linear codes are given with the  equivalent codes omitted. Note that these optimal linear codes are optimal for any crossover probability. 
When the crossover probability is small, 
the optimal code among all the linear $(n,4)$ binary codes
have been claimed for each block length $n$ in~\cite{cordaro1967optimum}.
When $n\neq 3$, the optimal codes given in Theorem~\ref{thm:linearopt} are the same as the ones claimed in \cite{cordaro1967optimum}.
When $n=3$, we obtain two more optimal codes $C_A$ and $C(1,1,1)$ than \cite{cordaro1967optimum}. In \cite{chen2013optimal}, an $(n,4)$ code that is optimal among all the linear codes is derived inductively for each blocklength $n\geq 2$. We get more optimal codes than them when $n=3k$ or $3k+1$ for a positive integer $k$. 

The following theorem characterizes a set containing all optimal codes for any given blocklength.
\begin{theorem}[All Optimal Codes]\label{thm:nbig3}
  For BSCs, the following properties about $(n,4)$ codes are universally satisfied:
  \begin{enumerate}
  \item When $n=2$, all the optimal codes are equivalent to $C(1,1,0)$;
  \item When $n=3$, %
    the set composed of
     $\begin{pmatrix}
			0&0&0\\
			0&1&1\\
			0&0&1\\
			1&1&1
                      \end{pmatrix}$,
                      $\begin{pmatrix}
                         0&0&0\\
                         0&0&1\\
                         0&1&1\\
                         1&1&0
                       \end{pmatrix}$, $C_A$, $C(1,1,1)$ or $C(1,2,0)$ and their equivalent codes is the set of all the optimal codes,
                    where the first two codes are neither linear, Class-I nor Class-II;
                     \item
                       When $n>3$, the set formed by the optimal linear codes, Class-I codes, Class-II codes  and their equivalent codes contains all the optimal codes.
                     \end{enumerate}
                   \end{theorem}

For $n=2$ or $3$, all the optimal codes are characterized in Theorem~\ref{thm:nbig3}, and are universal. It is somehow surprising that when $n=3$, there exist optimal codes that are nonlinear or include an all-zero column. When $n>3$, Theorem~\ref{thm:nbig3} says that an optimal code must be equivalent to a linear, Class-I, or Class-II code.

\begin{theorem}\label{thm:n8}
  For $4\leq n \leq 300$, all the optimal $ (n,4) $ codes are equivalent to linear codes. 
\end{theorem}

Together with Theorem~\ref{thm:linearopt} and Theorem~\ref{thm:nbig3}, all the optimal $(n,4)$ codes for $n$ up to $300$ are characterized. 
When $n\leq 8$, the above result can be proved by the theoretical analysis of Class-I codes.
For the larger blocklength $n$, computer evaluations are used.

\section{Major Results about Code Comparison}\label{sec:approach}
In this section, we give the major results of comparing two
$(n,4)$ codes, derive some further results, and prove Theorems \ref{thm:two}, \ref{thm:linearopt}, \ref{thm:nbig3} and  \ref{thm:n8}.

\subsection{Comparison Results of  Two  $(n,4)$ Codes}
\label{sec:dec_formulation}
We present some comparison results of the ML decoding performance between two $(n,4)$ codes which differ in one or two columns, but leave the technical proofs to the following sections.
We give an example to illustrate how to obtain a new  code with more linear-type columns  for an $(n,4)$ code by sequentially changing columns of the code. This motivates us to compare a general code with a linear code by a series of comparisons of codes with a small difference.
\begin{definition}[Flipping Operation]
  For a binary $(n,4)$ code represented as a $4\times n$ binary matrix, a flipping operation $f_{i,j}$ is defined to map this code to a new code obtained by flipping the $(i,j)$ entry of the code. %
\end{definition}

\begin{example}\label{example:opt1}
	Given a code $C$  with codewords $\cv c_1,\cv c_2,\cv c_3,\cv c_4$, a linear code $ C'$ with codewords $\cv c_1,\cv c_2,\cv c_3',\cv c_4'$ is found  by a series of flipping operations, as presented in \eqref{eq:optcha2a}, \eqref{eq:optcha2b} and \eqref{eq:optcha2c}.  
	\begin{align}
	C \triangleq 
	\begin{pmatrix}
	\cv c_1 \\
	\cv c_2\\
	\cv c_3 \\
	\cv c_4
	\end{pmatrix}\triangleq 
	\begin{pmatrix}
	0 & 0 & 0 & 0 & 0 & 0 \\
	0 & 0 & 0 & 1 & 1 & 1 \\
	\color{blue}\fbox{$0$} & 1 & 1 & 0 & 0 & 1 \\
	1 & 0 & 1 & 0 & 1 & 0 
	\end{pmatrix}
	& \stackrel{ f_{3,1}}{\longrightarrow}
	\begin{pmatrix}
	0 & 0 & 0 & 0 & 0 & 0 \\
	0 & 0 & 0 & 1 & 1 & 1 \\
	\color{red}\fbox{$1$} & 1 & 1 & 0 & 0 & 1 \\
	1 & 0 & 1 & \color{blue}\fbox{$0$} & 1 & 0 
	\end{pmatrix} 
	\label{eq:optcha2a}
	\\
	& \stackrel{ f_{4,4}}{\longrightarrow}
	\begin{pmatrix}
	0 & 0 & 0 & 0 & 0 & 0 \\
	0 & 0 & 0 & 1 & 1 & 1 \\
	1 & 1 & 1 & 0 & 0 & 1 \\
	1 & \color{blue}\fbox{$0$} & 1 & \color{red}\fbox{$1$} & 1 & 0 
	\end{pmatrix} 
	\label{eq:optcha2b}\\
	& \stackrel{ f_{4,2}}{\longrightarrow}
	\begin{pmatrix}
	0 & 0 & 0 & 0 & 0 & 0 \\
	0 & 0 & 0 & 1 & 1 & 1 \\
	1 & 1 & 1 & 0 & 0 & 1 \\
	1 & \color{red}\fbox{$1$} & 1 & 1 & 1 & 0 
	\end{pmatrix} \triangleq 
	\begin{pmatrix}
	\cv c_1 \\
	\cv c_2\\
	\cv c_3' \\
	\cv c_4'
	\end{pmatrix} \triangleq C'. \label{eq:optcha2c}
	\end{align}
	Specifically, we have $C' = (f_{4,2} \circ f_{4,4} \circ f_{3,1}) (C).$
\end{example}

In the above example, if the new code obtained by each flipping operation is better than the old one,  a better linear code $C'$ than $C$ is found.
In the following theorem, we show that a better code can be obtained by changing one column of $\bspan{0}$ in an $(n,4)$ code to a nonzero column.

\begin{theorem} \label{thm:0column} Consider an $(n,4)$ code $C$ with $|0|\geq 1$, and code $C'$ obtained by changing one column $\bspan{0}$ in $C$ to a  column $\bspan{s'}$.
  \begin{enumerate}
  \item If $C$ is equivalent to an $(n,4)$ code $C_0$ with $ |0|_{C_0}+ |5|_{C_0}+ |6|_{C_0}= n$ and $|5|_{C_0}, |6|_{C_0}$ odd, we have $\lambda_{C'}=\ls_C$ for any $s'\in \{0,1,\dots, 15\}$.
  \item If $C$ is not equivalent to $C_0$, then $\lambda_{C'}>\ls_C$ for some $s'\in\{3,5,6\}$. %
  \end{enumerate}
\end{theorem}

\begin{IEEEproof}
  See~\S\ref{sec:0column}.
\end{IEEEproof}

\begin{corollary}\label{cor:0col}
  Consider an $(n,4)$ code $C$ with $|0|\geq 2$. Then
  there exists a code $C'$ which is obtained by changing two $\bspan{0}$ columns in $C$ to two nonzero linear-type columns such that $\ls_{C'}>\ls_C$.
\end{corollary}

\begin{remark}
  Theorem~\ref{thm:0column} and Corollary~\ref{cor:0col} apply to codes with columns of type $\bspan{15}$ as well, which become $\bspan{0}$ after column flipping. 
\end{remark}

\begin{IEEEproof}
  Let $C_0$ be an $(n,4)$ code with $ |0|_{C_0}+ |5|_{C_0}+ |6|_{C_0}= n$ and $|5|_{C_0}, |6|_{C_0}$ odd. If $C$ is not equivalent to $C_0$, by Theorem~\ref{thm:0column}-2), there exists a code $\tilde{C}$ which is obtained by changing one $\bspan{0}$ column in $C$ to a linear-type column such that $\ls_{\tilde{C}}>\ls_C$. Let $C'$ be the code 
obtained by changing one $\bspan{0}$ column in $\tilde{C}$ to any linear-type column. By Theorem~\ref{thm:0column}, $\ls_{C'}\geq \ls_{\tilde{C}} > \ls_{C}$.

If $C$ is equivalent to $C_0$, by Theorem~\ref{thm:0column}-1), we have
$\ls_{\tilde C}=\ls_C$ with $\tilde C$ being obtained by changing one
$\bspan{0}$ column in $C$ to $\bspan{5}$. Then $\tilde C$ is not
equivalent to $C_0$, and by Theorem~\ref{thm:0column}-2), there 
exists a code $C'$ obtained by changing one $\bspan{0}$ column in
$\tilde C$ to some linear-type column such that
$\ls_{C'}>\ls_{\tilde C}$.  Thus $\ls_{C'}> \ls_{C}$.

For both cases,  we have $\ls_{C'}>\ls_C$ where $C'$ is obtained by changing two $\bspan{0}$ columns in $C$ to two linear-type columns.
\end{IEEEproof}

For an $(n,4)$ code $C$ of columns
$\bspan{0},\bspan{1},\ldots,\bspan{7}$, we have
\begin{IEEEeqnarray}{rCl}
  w(\vv c_1 \oplus \vv c_2) & = & |4| + |5| + |6| + |7|, \label{eq:w1} \\
  w(\vv c_1 \oplus \vv c_3) & = & |2| + |3| + |6| + |7|, \label{eq:w2}\\
  w(\vv c_1 \oplus \vv c_4) & = & |1| + |3| + |5| + |7|, \label{eq:w3}\\
  w(\vv c_2 \oplus \vv c_3) & = & |2| + |3| + |4| + |5|, \label{eq:w4}\\
  w(\vv c_3 \oplus \vv c_4) & = & |1| + |2| + |5| + |6|. \label{eq:w5}
\end{IEEEeqnarray}
In the following theorem, we give the sufficient conditions so that a  better code is obtained by a certain flipping operation. See an illustration in Fig.~\ref{fig:dong1}.

\begin{figure}[t]
  \begin{equation*}
    \begin{pNiceMatrix}[first-col,last-row]
      \vv c_1 & 0 & * & \cdots & *\\
      \vv c_2 & 0 & * & \cdots & *\\
      \vv c_3 & \fbox{0} & * & \cdots & *\\
      \vv c_4 & 1 & * & \cdots & *\\
      & \bspan{1} & 
    \end{pNiceMatrix}
    \xrightarrow{f_{3,1}}
    \begin{pNiceMatrix}[first-col,last-row]
      \vv c_1 & 0 & * & \cdots & *\\
      \vv c_2 & 0 & * & \cdots & *\\
      \vv c_3' & \fbox{1} & * & \cdots & *\\
      \vv c_4 & 1 & * & \cdots & *\\
      & \bspan{3} & 
    \end{pNiceMatrix}
  \end{equation*}
	\caption{An example of Theorem \ref{thm:even}. $ C $ has codewords $ \cv c_1,\dots, \cv c_4 $ and $ C' $ has codewords $ \cv c_1,\cv c_2,\cv c_3', \cv c_4 $, where $\vv c_3'$ is obtained by flipping the first bit of $\vv c_3$. Theorem \ref{thm:even} says $ \lambda_{ C'} \geq \lambda_{ C}$ if $ w(\cv c_3\oplus \cv c_4)$ is even. }\label{fig:dong1}
\end{figure}

\begin{theorem}[$1$-bit Flip] \label{thm:even}
	Consider an $(n,4)$ code $ C$ of codewords $\vv c_1,\vv c_2,\vv c_3,\vv c_4$, and with a column being  type $\bspan{1}$ and 
        $w(\vv c_3\oplus \vv c_4)$ even. Let $ C'$ be the code obtained by replacing a column of type $\bspan{1}$ in $ C$ by $\bspan{3}$ (or  flipping the third entry in a column of type $\bspan{1}$ in $C$). Then $\lambda_{ C'}\geq \lambda_C$.
        Moreover, when $C$ has only columns of types $\bspan{0}, \bspan{1},\ldots, \bspan{7}$,  $\lambda_{ C'} = \lambda_C$ if and only if 
		one of the following conditions are satisfied
		\begin{enumerate}
			\item[i)] $w(\vv c_1\oplus \vv c_3)$ and $w(\vv c_2\oplus \vv c_3)$ are both odd;
			\item[ii)] $ |1|=1$, $|2|=|4|=|5|=|7|=0 $ and, $ |3|$ and $|6| $ are both odd;
			\item[iii)] $ |1|=1$, $|2|=|4|=|6|=|7|=0 $ and,  $ |3| $ and $ |5| $ are both odd.
                        \end{enumerate}
\end{theorem}
\begin{IEEEproof}
	See  \S\ref{sec:proof:even}.
\end{IEEEproof}
\begin{remark}\label{rm:code1}
  Theorem \ref{thm:even} can be applied to other codewords as well by applying the code equivalence relation.   
		  Consider an $(n,4)$ code $C$ of codewords $\vv c_1,\vv c_2,\vv c_3,\vv c_4$ with $w(\vv c_s\oplus \vv c_t)$ being even for certain $1\leq s, t \leq 4$ with $s\neq t$, and with a column of type $\bspan{2^{4-s}}$, i.e. a column with only the $s$-th entry being $1$.
	Let $C'$ be the code obtained by flipping the $t$-th entry in a column of type $\bspan{2^{4-s}}$ in $C$, or we say 
	replacing a column of type $\bspan{2^{4-s}}$ in $C$ by $\bspan{2^{4-s}+2^{4-t}}$. Since $C$ and $C'$ are respectively equivalent to the codes defined in Theorem~\ref{thm:even}, we have $\lambda_{C'}\geq \lambda_C$.
	The sufficient and necessary condition for $\lambda_{ C'}= \lambda_{C}$   can also be obtained by the equivalence relation, as illustrated in the following example.
\end{remark}
\begin{example}
  Consider an $(n,4)$ code $C$ of codewords $\vv c_1,\cv c_2,\cv c_3,\vv c_4$, and with the first column being  type $\bspan{1}$.
  Suppose $w(\vv c_2\oplus \vv c_4)$ is even. Let $ C'$ be the code obtained by replacing the first column of $C$ by $\bspan{5}$. 
  Suppose $\cv c_2'$ is obtained by flipping the first entry of $\cv c_2$. Let
  \begin{align*}
    C =\begin{pmatrix}
         \cv c_1\\
         \cv c_2\\
         \cv c_3\\
         \cv c_4
       \end{pmatrix},
    C_0 =\begin{pmatrix}
           \cv c_1\\
           \cv c_3\\
           \cv c_2\\
           \cv c_4
         \end{pmatrix},
    C' =\begin{pmatrix}
          \cv c_1\\
          \cv c_2'\\
          \cv c_3\\
          \cv c_4
        \end{pmatrix},
    C_0' =\begin{pmatrix}
            \cv c_1\\
            \cv c_3\\
            \cv c_2'\\
            \cv c_4
          \end{pmatrix},
  \end{align*}
  where $ C_0$ is obtained by exchanging the second row and the third row of $C$ and $ C_0'$ is obtained by exchanging the second row and the third row of $C'$, which induces
  $\lambda_{ C} = \lambda_{ C_0}$ and $\lambda_{ C'} = \lambda_{ C_0'}$. As $ C_0$ satisfies the conditions in Theorem \ref{thm:even},  $\lambda_{C_0} = \lambda_{C'_0}$ if and only if one of the following conditions holds
  \begin{enumerate}
  \item[i)] $w(\vv c_1\oplus \vv c_2)$ and $w(\vv c_3\oplus \vv c_2)$ are odd;
  \item[ii)] $ |1|_{ C_0}=1$, $|2|_{ C_0}=|4|_{C_0}=|5|_{C_0}=|7|_{C_0}=0 $, $ |3|_{ C_0}$ and $|6|_{ C_0} $ are odd;
  \item[iii)] $ |1|_{C_0}=1,\ |2|_{ C_0}=|4|_{ C_0}=|6|_{ C_0}=|7|_{ C_0}=0 $, $ |3|_{ C_0} $ and $ |5|_{C_0} $ are odd.
  \end{enumerate}
  Due to $C_0$ is obtained by exchanging the second row and the third row of $C$, we have
  \begin{align*}
    |1|_{C} = |1|_{C_0} , |2|_{C} =  |4|_{ C_0} ,  |3|_{C} =  |5|_{C_0} ,   |4|_{C} =  |2|_{C_0} ,  |5|_{C} =  |3|_{ C_0} ,  
    |6|_{ C} =  |6|_{C_0} ,  |7|_{C} =  |7|_{C_0} .
  \end{align*}
  Hence we have $\lambda_{C} = \lambda_{C'}$ if and only if one of the following conditions holds
  \begin{enumerate}
  \item[i)] $w(\vv c_1\oplus \vv c_2)$ and $w(\vv c_3\oplus \vv c_2)$ are odd;
  \item[ii)] $ |1|_{C}=1$, $|2|_{C}=|4|_{C}=|3|_{C}=|7|_{C}=0 $, $ |5|_{C}$ and $|6|_{C} $ are odd;
  \item[iii)] $ |1|_{C}=1,\ |2|_{ C}=|4|_{ C}=|6|_{C}=|7|_{C}=0 $, $ |3|_{C} $ and $ |5|_{ C} $ are odd.
  \end{enumerate}
\end{example}

Theorem \ref{thm:even} gives a sufficient condition to find a  better code with one more linear-type column. Using Theorem \ref{thm:even}, for Example \ref{example:opt1}, we can verify that the flipping operations $f_{3,1}$, $f_{4,4}$ and $f_{4,2}$ can all lead to a  better code  and thus the code $C'$ is  better  than $C$. 
Based on the inequality $\ls_{C'}\geq \ls_C$ in Theorem \ref{thm:even}, we can verify that for a code without columns of type $\bspan{7}$, there exists a   better code with only the linear-type and type $\bspan{1}$ columns in  the following corollary.
\begin{corollary}\label{cor:twopo}
  Consider an $(n,4)$ code $C$ with $\sum_{i=1}^6|i|_C=n$. There exists a code $C'$ with $\lambda_{C'}\geq \lambda_C$ and $|1|_{C'} + |3|_{C'}+|5|_{C'}+|6|_{C'}=n$.
\end{corollary}
\begin{remark}
  An $(n,4)$ code with only the linear-type and type $\bspan{1}$
  columns is equivalent to some code with only the linear-type and type
  $\bspan{2}$ (or $\bspan{4}$) columns. So the corollary still holds
  if $|1|_{C'} + |3|_{C'}+|5|_{C'}+|6|_{C'}=n$ is replaced by
  $|2|_{C'} + |3|_{C'}+|5|_{C'}+|6|_{C'}=n$ or
  $|4|_{C'} + |3|_{C'}+|5|_{C'}+|6|_{C'}=n$.
\end{remark}
\begin{IEEEproof}
  In this proof, we write $|i|_C$ as $|i|$. 
  Suppose at least two of $|1|, |2|, |4|$ are positive, since otherwise, the proof is done by interchanging rows of $C$. We argue the case that $|1|$ and $|2|$ are positive. Other cases can be converted to this case by interchanging rows. Write
\begin{IEEEeqnarray*}{rCl}
  w(\vv c_2\oplus \vv c_3) & = & |2| + |3| + |4| + |5|,\\
  w(\vv c_3\oplus \vv c_4) & = & |1| + |2| + |5| + |6|,\\
  w(\vv c_2\oplus \vv c_4) & = & |1| + |3| + |4| + |6|.
\end{IEEEeqnarray*}
We claim that one of the above three weights must be even. Assume $w(\vv c_2\oplus \vv c_3)$ is odd. Then $|3|+|4|$ and $|2|+|5|$ are of different parity, so that one of $w(\vv c_3\oplus \vv c_4)$ and $w(\vv c_2\oplus \vv c_4)$ must be even. Therefore, one of the following three is possible when $|1|$ and $|2|$ are positive:
\begin{enumerate}
\item If $w(\vv c_2\oplus \vv c_3)$ is even,  Theorem~\ref{thm:even} implies a  better code with $|2|$ smaller and $|6|$ bigger.
\item If $w(\vv c_3\oplus \vv c_4)$ is even,  Theorem~\ref{thm:even} implies a  better code with $|1|$ smaller and $|3|$ bigger.
\item If $w(\vv c_2\oplus \vv c_4)$ is even,  Theorem~\ref{thm:even} implies a  better code with $|1|$ smaller and $|5|$ bigger.
\end{enumerate}
As long as $|1|$ and $|2|$ are positive, the above step can be repeated. 
Thus, there exists a code $C'$ with at most one of $|1|_{C'}$ and $|2|_{C'}$ being positive that is better than  the code $C$. 

If there are still two of $|1|, |2|, |4|$ positive, using the same argument, we can obtain a better code $C'$ where at most one of $|1|_{C'}, |2|_{C'}, |4|_{C'}$ is positive and $\sum_{i=1}^6|i|_{C'}=n$. 
The corollary is proved by properly interchanging rows of $C'$.
\end{IEEEproof}
	We can refine the better code in Corollary \ref{cor:twopo} to a smaller subset of $(n,4)$ codes.
In the following corollary, we  show  that for a code with $|1|_{C} + |3|_{C}+|5|_{C}+|6|_{C}=n$ that is both non-Class-I and nonlinear, a  better linear or Class-I code always exists.
\begin{corollary}\label{cor:onepo}
  Consider a non-Class-I, nonlinear $(n,4)$ code $C$ with $|1|_{C} + |3|_{C}+|5|_{C}+|6|_{C}=n$. There exists an either linear or Class-I code $C'$ with $\lambda_{C'}\geq \lambda_C$ and $|1|_{C'}<|1|_C$.
\end{corollary}
\begin{IEEEproof}
  In this proof, we write $|i|_C$ as $|i|$. Since $C$ is nonlinear, $|1|>0$.
  We claim that at least one of the following three weights are even:
\begin{IEEEeqnarray}{rCl}
  w(\vv c_1\oplus \vv c_4) & = & |1| + |3| + |5|, \label{eq:3w:1}\\
  w(\vv c_3\oplus \vv c_4) & = & |1| + |5| + |6|, \label{eq:3w:2}\\
  w(\vv c_2\oplus \vv c_4) & = & |1| + |3| + |6|. \label{eq:3w:3}
\end{IEEEeqnarray}
Consider two cases of $|1|$:
\begin{itemize}
\item When $|1|$ is odd, $|3|$, $|5|$ and $|6|$ are not of the same parity since $C$ is not of Class-I, which implies at least one of \eqref{eq:3w:1}, \eqref{eq:3w:2}, \eqref{eq:3w:3} is even. By Theorem~\ref{thm:even}, there is a  better code $C_1$ with $|1|_{C_1}=|1|-1$ even and $|1|_{C_1} + |3|_{C_1}+|5|_{C_1}+|6|_{C_1}=n$.
\item 
When $|1|$ is even, if \eqref{eq:3w:1}, \eqref{eq:3w:2}, \eqref{eq:3w:3} are all odd, then $|3|+|5|$, $|5|+|6|$ and $|3|+|6|$ are all odd, which is not possible for any integers $|3|$, $|5|$, $|6|$.
Then at least one of \eqref{eq:3w:1}, \eqref{eq:3w:2}, \eqref{eq:3w:3} is even, 
Theorem~\ref{thm:even} implies a  better code $C_1$ with $|1|_{C_1}=|1|-1$ odd and $|1|_{C_1} + |3|_{C_1}+|5|_{C_1}+|6|_{C_1}=n$.
\end{itemize}
For both cases, a better code $C'$ with $|1|$ strictly smaller always exists if $C$ is non-Class-I, nonlinear.
By repeating the similar argument on $C_1$, we eventually obtain a  better code $C'$ which is either linear (i.e., $|1|_{C'}=0$) or is of Class-I so that \eqref{eq:3w:1}, \eqref{eq:3w:2}, \eqref{eq:3w:3} are all odd.
\end{IEEEproof}

Theorem~\ref{thm:even} and the above two corollaries help us to find a better code which has only one nonlinear type column $\bspan{1}$. These results can be applied on codes with type $\bspan{7}$ columns as well by flipping columns and interchanging rows to change $\bspan{7}$ columns to $\bspan{1}$, $\bspan{2}$ or $\bspan{4}$. However, this approach is not effective when the code has all $|1|$, $|2|$ and $|4|$ positive. In the following example, we see that when changing the $\bspan{7}$ column to $\bspan{1}$, the original $\bspan{1}$ column is changed to $\bspan{7}$. 
\begin{IEEEeqnarray*}{rCl}
  \begin{pmatrix}
    0 & 0 & 0 & 0 & 0 & 0 & 0\\
    0 & 0 & 0 & 1 & 1 & 1 & 1\\
    0 & 1 & 1 & 0 & 0 & 1 & 1\\
    1 & 0 & 1 & 0 & 1 & 0 & 1
  \end{pmatrix}
  &\xrightarrow{\text{flip col } 7}&
  \begin{pmatrix}
    0 & 0 & 0 & 0 & 0 & 0 & 1\\
    0 & 0 & 0 & 1 & 1 & 1 & 0\\
    0 & 1 & 1 & 0 & 0 & 1 & 0\\
    1 & 0 & 1 & 0 & 1 & 0 & 0
  \end{pmatrix} \\
  & \xrightarrow{\text{interchange row } 1, 4} &
  \begin{pmatrix}
    1 & 0 & 1 & 0 & 1 & 0 & 0\\
    0 & 0 & 0 & 1 & 1 & 1 & 0\\
    0 & 1 & 1 & 0 & 0 & 1 & 0\\
    0 & 0 & 0 & 0 & 0 & 0 & 1
  \end{pmatrix} \\
  & \xrightarrow{\text{flip col } 1, 3, 5}&
  \begin{pmatrix}
    0 & 0 & 0 & 0 & 0 & 0 & 0\\
    1 & 0 & 1 & 1 & 0 & 1 & 0\\
    1 & 1 & 0 & 0 & 1 & 1 & 0\\
    1 & 0 & 1 & 0 & 1 & 0 & 1
  \end{pmatrix} 
\end{IEEEeqnarray*}
Theorem~\ref{thm:odd} provides an approach to handle this case, which can find a better code by changing two  columns being of type $ \bspan{1} $ and $ \bspan{7} $ together to linear types. 

\begin{theorem}[$2$-bit Flip in One Row]  \label{thm:odd}
	Consider an $(n,4)$ code $C$ with two columns being of type $ \bspan{1} $ and $ \bspan{7} $. Let $C' $ be the code obtained by replacing two columns of type $ \bspan{1} $ and $ \bspan{7} $ in $C $ by $ \bspan{3} $ and $ \bspan{5} $ (i.e., flipping  the bits in the third row of these two columns). Then %
	$ \lambda_{C'}\geq \lambda_{C} $.
	Moreover, when $C$ has only columns of types $\bspan{0}, \bspan{1},\ldots, \bspan{7}$,   the equality holds if and only if $ |1|_{C}=|7|_{C}=1$, $|2|_{C}=|4|_{C}=|6|_{C}=0 $ and at least one of $ |3|_{C}$ and $ |5|_{C} $ is odd.
\end{theorem}
\begin{IEEEproof}
See \S\ref{sec:fliptwobits}.
\end{IEEEproof}

Fig.~\ref{fig:dong2} illustrates an example of the above theorem.

\begin{remark}
	Consider an $(n,4)$ code $C$ with two columns of the types  $\bspan{7}$ and $\bspan{2^{4-s}}$, $s\in \{2,3,4\} $, i.e. a column with only the $s$-th entry being $1$. Let $C'$ be the code obtained by replacing these two columns  with
	$\bspan{7-2^{4-t}}$ and   $\bspan{2^{4-s}+2^{4-t}}$ for $t\in \{2,3,4\} \setminus \{s\}$, or we say flipping  the bits in $t$-th row of these two columns. Since $C$ and $C'$ are respectively equivalent to the codes defined in Theorem~\ref{thm:odd},  $\lambda_{C'}\geq \lambda_{C}$. 	The sufficient and necessary condition for $\lambda_{ C'}= \lambda_{C}$   can also be obtained by Theorem~\ref{thm:even} using code equivalence.
\end{remark}

\begin{figure}[t]

    \begin{equation*}
    \begin{pNiceMatrix}[first-col,last-row]
      \vv c_1 & 0 & 0 & * & \cdots & *\\
      \vv c_2 & 0 & 1 & * & \cdots & *\\
      \vv c_3 & \fbox{0} & \fbox{1} & * & \cdots & *\\
      \vv c_4 & 1 & 1 & * & \cdots & *\\
      & \bspan{1} & \bspan{7} & 
    \end{pNiceMatrix}
    \xrightarrow{f_{3,1}}
    \begin{pNiceMatrix}[first-col,last-row]
      \vv c_1 & 0 & 0 & * & \cdots & *\\
      \vv c_2 & 0 & 1 & * & \cdots & *\\
      \vv c_3' & \fbox{1} & \fbox{0} & * & \cdots & *\\
      \vv c_4 & 1 & 1 & * & \cdots & *\\
      & \bspan{3} & \bspan{5} &
    \end{pNiceMatrix}
  \end{equation*}
  \caption{An example of Theorem \ref{thm:odd}. $ C $ has codewords $ \cv c_1, \cv c_2, \cv c_3, \cv c_4 $ and $ C' $ has codewords $ \cv c_1, \cv c_2, \cv c_3', \cv c_4 $. Theorem \ref{thm:odd} says $ \lambda_{ C'} \geq \lambda_{ C}. $}\label{fig:dong2}
\end{figure}

\subsection{Optimal Code Inclusive Set: Proof of Theorem \ref{thm:two}}

The aforementioned  Theorem~\ref{thm:0column}, Theorem~\ref{thm:even} and Theorem~\ref{thm:odd} induce a partial order on all the $(n,4)$ codes. Based on this partial order, we are able to prove Theorem \ref{thm:two} by showing that for any $(n,4)$ code $C$, there exists an either linear or Class-I code better than $C$.

\begin{IEEEproof}[Proof of Theorem \ref{thm:two}]
  Consider an arbitrary $(n,4)$ code $C$. By Theorem~\ref{thm:0column}, if $C$ has $\bspan{0}$ or $\bspan{15}$ columns, there exists an $(n,4)$ code $C'$ without $\bspan{0}$ and $\bspan{15}$ columns such that $\ls_{C'}\geq \ls_{C}$. 
  Suppose $C$ has no $\bspan{0}$ or $\bspan{15}$  columns.
  As column flipping does not change the ML decoding performance, we can obtain an equivalence code with $\sum_{i=1}^7|i|=n$ by flipping the columns of $C$. We then discuss $C$ with $\sum_{i=1}^7|i|=n$ in two cases.

  If $0<|7|\leq |1|+|2|+|4|$ in $C$, by Theorem~\ref{thm:odd}, there exists a code $C_2$ with $\lambda_{C_2} \geq \lambda_{C}$ and $\sum_{i=1}^6|i|=n$ obtained by replacing, one-by-one, pairs of columns of types $\bspan{7}$ and $\bspan{2^s}$ ($s=0,1,2$).
Following Corollary~\ref{cor:twopo}, there exists code $C_3$, no worse than $C_2$, where $|1|_{C_3}+ |3|_{C_3}+ |5|_{C_3}+|6|_{C_3}=n$. Then by Corollary~\ref{cor:onepo}, there exists an either linear or Class-I code $C_4$ such that $\lambda_{C_4} \geq \lambda_{C_3} \geq \lambda_{C_2} \geq \lambda_C$.

If $|1|+|2|+|4|< |7|$ in $C$, by Theorem~\ref{thm:odd}, there exists a better code $C_2'$ with $|1|+|2|+|4|=0$. By flipping columns, we can obtain a code $C_3'$ of the same performance of $C_2'$ that has $|1|_{C_3'}>0$ and $|1|_{C_3'}+|3|_{C_3'}+|5|_{C_3'}+|6|_{C_3'}=n$. Again, by Corollary~\ref{cor:onepo}, the proof is completed.
\end{IEEEproof}

\subsection{All Optimal Codes Among Linear Codes: Proof of Theorem \ref{thm:linearopt}}

Now we move on to compare two linear codes with only one column difference. 

\begin{theorem}\label{thm:linearcompare}
  Consider an $(n,4)$ linear code $C(n_3,n_5,n_6)$ with $n_3>0$.
  Let $ C'$ be the code obtained by replacing a column of type $\bspan{3}$ of $C$ by $\bspan{5}$.
  \begin{enumerate}
  \item When $n_3, n_5+1,n_6$ have the same parity, $\lambda_{ C'}= \lambda_ C$;
  \item When $n_3, n_5,n_6$ have the same parity,%
    \begin{itemize}
    \item if $n_3=1$,  $\lambda_{C'}= \lambda_ C$, and
    \item  if $n_3\geq 2$,  $\lambda_{ C'}>\lambda_C$.
    \end{itemize}
  \item When $n_5\leq \min\{n_3, n_6\} $ and  $n_3-1, n_5,n_6$ have the same parity, 
    \begin{itemize}
    \item if $n_3= n_5+1$,  $\lambda_{ C'}= \lambda_C$, and
    \item if $n_3> n_5+1$,  $\lambda_{C'}> \lambda_ C$.
    \end{itemize}
  \end{enumerate}
\end{theorem}
\begin{IEEEproof}
  See \S\ref{sec:linearcode}.
\end{IEEEproof}

\begin{remark}
  Our results of comparing two linear codes are stronger than the one 
  in~\cite{chen2013optimal} since we give the sufficient conditions for the strict inequality and equality between the decoding performance of two codes.
\end{remark}
\begin{remark}\label{rm:equilinear}
	Since $\ls_{C(n_3,n_5,n_6)} = \ls_{C(n_r,n_s,n_t)}$ with $\{r,s,t\}=\{3,5,6\}$,  Theorem~\ref{thm:linearcompare} holds when the roles of $n_3$, $n_5$ and $n_6$ are switched. For example, if we want to compare $C(u,v,w)$ with $C(u-1,v,w+1)$ when 	$u, v,w+1$ have the same parity, we can get   $ \ls_{C(u,w,v)}=\ls_{ C(u-1,w+1,v)}$  first by Theorem~\ref{thm:linearcompare}-1) and following this
	we have  $ \ls_{C(u,v,w)}=\ls_{ C(u-1,v,w+1)}$. 
\end{remark}
\begin{corollary}\label{cor:linear1}
  For a linear $(n,4)$ code $C(n_3,n_5,n_6)$ with $n_3\geq 2$, let $C'=C(n_3-2,n_5,n_6+2)$.
  \begin{enumerate}
  \item When $n_3$, $n_5+1$ and $n_6$ have the same parity, $n_3\geq  n_6+3$ and $n_5\geq n_6-1$, we have $\lambda_{ C'}>\lambda_C$;
  \item When $n_3$, $n_5$ and $n_6+1$ have the same parity, $n_3\geq  n_6+4$ and $n_5\geq n_6+1$, we have $\lambda_{ C'}>\lambda_C$.
  \end{enumerate}
\end{corollary}
\begin{IEEEproof}
  \begin{enumerate}
  \item
    By Theorem \ref{thm:linearcompare}--1), we have $$\lambda_{C(n_3-1,n_5+1,n_6)}=\lambda_{C(n_3,n_5,n_6)}.$$ Due to
    $n_3-1>n_6 +1$ and $n_6 \leq \min\{n_3-1,n_5+1\}$, by Theorem \ref{thm:linearcompare}--3) and  Remark~\ref{rm:equilinear}, $$\lambda_{C(n_3-2,n_5+1,n_6+1)}> \lambda_{C(n_3-1,n_5+1,n_6)}.$$ Hence we have $\lambda_{C(n_3-2,n_5+1,n_6+1)}>\lambda_{C(n_3,n_5,n_6)}$. As $(n_3-2),$ $(n_5+1)$ and $(n_6+1)+1$ have the same parity, by Theorem \ref{thm:linearcompare}--1) and Remark~\ref{rm:equilinear}, $$\lambda_{C(n_3-2,n_5+1,n_6+1)} = \lambda_{C(n_3-2,n_5,n_6+2)} .$$
  \item 
    By Theorem \ref{thm:linearcompare}--1) and Remark~\ref{rm:equilinear},  $$\lambda_{C(n_3-1,n_5,n_6+1)}=\lambda_{C(n_3,n_5,n_6)}.$$ Due to
    $n_3-1>(n_6 +1)+1$ and $n_6+1 \leq \min\{n_3-1,n_5\}$, by Theorem \ref{thm:linearcompare}--3) and Remark~\ref{rm:equilinear},
    \begin{equation*}
      \lambda_{C(n_3-2,n_5,n_6+2)}>\lambda_{C(n_3-1,n_5,n_6+1)}.\qedhere  
    \end{equation*}
  \end{enumerate}
\end{IEEEproof}

Now we are ready to give the proof of Theorem~\ref{thm:linearopt}.
\begin{IEEEproof}[Proof of Theorem \ref{thm:linearopt}]
	Due to the code equivalence, we can search the optimal codes among codes with only  columns $\bspan{0},$ $\bspan{3}$, $\bspan{5}$ and $\bspan{6}$.
	We first find all the optimal $(n,4)$ codes among all the linear codes  with only columns $\bspan{3}$, $\bspan{5}$ and $\bspan{6}$ for the three cases of $n$. After that, we will discuss the general linear codes that may contain $\bspan{0}$ columns. In this proof, we write $\max\{a,b\}$ as $a\lor b$ and $\min\{a,b\}$ as $a\land b$. 
	
First, consider $n =  3k-1$ for a positive integer $k$. Observe that  $C(n_3,n_5,n_6)$ with $n_3\lor n_5\lor n_6 -n_3\land n_5\land n_6 \leq 1$ is equivalent to $C(k,k,k-1)$.
If  we can show that $C(k,k,k-1)$ is strictly better than any linear code $C(n_3,n_5,n_6)$ with $n_3\lor n_5\lor n_6 -n_3\land n_5\land n_6 > 1$, then $C(k,k,k-1)$ is universally optimal among linear codes with only columns $\bspan{3}$, $\bspan{5}$ and $\bspan{6}$. It can be verified that $n_3\lor n_5\lor n_6 -n_3\land n_5\land n_6 > 1$ is equivalent to $n_3\land n_5\land n_6 >k. $
WLOG, we  consider the linear code $ C(n_3,n_5,n_6) $ with $n_3>k$ and $n_3\geq n_5\geq n_6$. There are totally four cases for the parities of $n_3$, $n_5$ and $n_6$, and we can find a strictly better code than $ C(n_3,n_5,n_6) $  in each case:
\begin{enumerate}
\item[1-1)] When $n_3$, $n_5$ and $n_6$ have the same parity, we have $n_6\leq k-1\leq n_3-2$. By Theorem \ref{thm:linearcompare}--2), $\lambda_{C(n_3-1,n_5,n_6+1)} > \lambda_{C(n_3,n_5,n_6) }$ since $n_3>k\geq 1$.  %
\item[1-2)] When $n_3+1$, $n_5$ and $n_6$ have the same parity,  we have $n_6\leq k-2$ and $n_3\geq n_6 +3$. By Theorem \ref{thm:linearcompare}--3),  $\lambda_{C(n_3-1,n_5,n_6+1)} > \lambda_{C(n_3,n_5,n_6) }$. %
\item[1-3)] When $n_3$, $n_5+1$ and $n_6$ have the same parity, we have $n_6\leq k-2$, $n_3\geq k+2$ and  
  $n_3-n_6\geq 4$. By Corollary \ref{cor:linear1}--1),  $\lambda_{C(n_3-2,n_5,n_6+2)}>\lambda_{C(n_3,n_5,n_6)}$. %
 \item[1-4)] When $n_3$, $n_5$ and $n_6+1$ have the same parity, we have $n_6\leq k-3$, $n_3\geq n_6+5$ and $n_5>n_6$. By Corollary \ref{cor:linear1}--2),
  $\lambda_{C(n_3-2,n_5,n_6+2)}>\lambda_{C(n_3,n_5,n_6)}$. %
\end{enumerate}
Denote by $C(n_3',n_5,n_6')$ the strictly better code than code $C(n_3,n_5,n_6)$ obtained above. When $n_3>n_5$, for all the four cases above, either $ n_3'=n_3-1,n_6'= n_6+1$ or $n_3'=n_3-2,n_6'= n_6+2$  and hence
\begin{equation*}
n_3'\lor n_5\lor n_6'\leq (n_3-1)\lor n_5 \lor k =n_3-1 < n_3\lor n_5\lor n_6
\end{equation*}  and $n_3'\land n_5\land n_6' \geq n_3\land n_5\land n_6$. When $n_3=n_5$, we have 
\begin{equation*}
n_3'\lor n_5\lor n_6' =n_5 = n_3\lor n_5\lor n_6 \text{ and } n_3'\land n_5\land n_6'= n_6+1  > n_3\land n_5\land n_6.
\end{equation*}
Therefore, in each case above, the strictly better code $C(n_3',n_5,n_6')$ satisfies that $n_3' \lor n_5\lor n_6’ - n_3' \land n_5\land n_6’$ is strictly smaller than $n_3\lor n_5\lor n_6 -n_3\land n_5\land n_6 $. We can repeat the above argument 
until a strictly better code $C(n_3'',n_5'',n_6'')$ with $n_3''\lor n_5''\lor n_6'' -n_3''\land n_5''\land n_6'' \leq 1$ is obtained.
Therefore, we have $\lambda_{C(k,k,k-1)} > \lambda_{C(n_3,n_5,n_6)}$. %

Second, consider $n =  3k$ for a positive integer $k$.
For a code $C(n_3,n_5,n_6)$ with $n_3\geq  k+2$ and $n_3\geq n_5\geq n_6$, we can verify that either $C(n_3-1,n_5,n_6+1)$ or $ C(n_3-2,n_5,n_6+2) $ is  strictly better than $C(n_3,n_5,n_6)$ in the following four cases:
\begin{enumerate}
\item[2-1)] When $n_3,n_5$ and $n_6$ have the same parity, by Theorem \ref{thm:linearcompare}-2), $C(n_3-1,n_5,n_6+1)$  is  strictly better than $C(n_3,n_5,n_6)$.
\item[2-2)] When $n_3+1,n_5$ and $n_6$ have the same parity, by Theorem \ref{thm:linearcompare}-3), the code $C(n_3-1,n_5,n_6+1)$  is strictly better than $C(n_3,n_5,n_6)$.
\item[2-3)] When $n_3,n_5+1$ and $n_6$ have the same parity, we have $n_3\geq n+3$ and $n_6\leq k-3$.
  By Corollary \ref{cor:linear1}-1), $C(n_3-2,n_5,n_6+2)$ is strictly better than $C(n_3,n_5,n_6)$. %
\item[2-4)] When $n_3,n_5$ and $n_6+1$ have the same parity, we have $n_3\geq k+3$ and $n_6\leq k-2$.
  By Corollary \ref{cor:linear1}-2), $C(n_3-2,n_5,n_6+2)$ is strictly better than $C(n_3,n_5,n_6)$. %
\end{enumerate}
Denote by $C(n_3',n_5,n_6')$ the strictly better code than code $C(n_3,n_5,n_6)$ obtained above. Observe that $n_6'\leq k.$
By the above four cases, we now verify that there exists a strictly better code $C(n_3',n_5',n_6')$ with $n_3'\lor n_5'\lor n_6'<n_3\lor n_5\lor n_6$  for the code $ C(n_3,n_5,n_6)$.
When $n_3>n_5$, we have 
\begin{equation*}
n_3'\lor n_5\lor n_6'\leq (n_3-1)\lor n_5 \lor k < n_3=n_3\lor n_5\lor n_6. 
\end{equation*}
When $n_3=n_5,$ we have $n_6\leq k-4$ and by the four cases 2-1)-4) above, there exists $C(n_3^*,n_5,n_6^*)$ with either $ n_3^*=n_3-1,n_6^*= n_6+1$ or $n_3^*=n_3-2, n_6^*= n_6+2$ such that $\ls_{C(n_3,n_5,n_6)}< \ls_{C(n_3^*,n_5,n_6^*)}$. Then we have $n_3^*\geq k$ and $n_6^*\leq k-2$, which implies $n_6^*< n_3^*$. Now for the code $C(n_3^*,n_5,n_6^*)$  with $n_5\geq k+2$ and $n_5 > n_3^*>n_6^*$, by the equivalence between $C(n_3^*,n_5,n_6^*)$ and $C(n_5,n_3^*,n_6^*)$ and applying the above four cases 2-1)-4) to $C(n_5,n_3^*,n_6^*)$, we get that there exists  $C(n_3',n_5',n_6')$ with either $ n_3'=n_3^*,n_5'= n_5-1,n_6'= n_6^*+1$ or $n_3'=n_3^*, n_5'=n_5-2,n_6'= n_6^*+2$ such that $\ls_{C(n_3^*,n_5,n_6^*)}< \ls_{C(n_3',n_5',n_6')}$ and 
\begin{equation*}
n_3'\lor n_5'\lor n_6'\leq  n_3^*\lor n_5'\lor (n_6^*+2)\leq n_3-1<n_3=n_3\lor n_5\lor n_6.
\end{equation*}
Therefore, there always exists a strictly better code $C(n_3',n_5',n_6')$ with $n_3'\lor n_5'\lor n_6'<n_3\lor n_5\lor n_6$  for the code $ C(n_3,n_5,n_6)$.
Then we can repeat the above argument until a strictly better code $C(n_3'',n_5'',n_6'')$ with $n_3''\lor n_5''\lor n_6''\leq k+1$ is obtained. 
Thus there always exists a strictly better code $C(n_3'',n_5'',n_6'')$ with $n_3''\lor n_5''\lor n_6''\leq k+1$  for any code $ C(n_3,n_5,n_6) $ with $ n_3\lor n_5\lor n_6\geq k+2$.
When $k\geq 2$, excluding equivalent codes, there are totally three possibilities of $C(n_3'',n_5'',n_6'')$ with $ n_3''\lor n_5''\lor n_6''\leq k+1$: $$C(k,k,k), C(k+1,k,k-1) \text{ and } C(k+1,k+1,k-2).$$  By Theorem \ref{thm:linearcompare}-1), $\lambda_{C(k+1,k+1,k-2)}=\lambda_{C(k+1,k,k-1)}$. By Theorem \ref{thm:linearcompare}-2),  $\lambda_{C(k,k,k)} < \lambda_{C(k,k-1,k+1)}$.
Therefore, $ C(k+1,k,k-1) $ and $ C(k+1,k+1,k-2) $ are strictly better than their non-equivalent codes with only columns from  $\bspan{3}$, $\bspan{5}$ and $\bspan{6}$. 
When $k=1$,  the code  $C(n_3'',n_5'',n_6'')$ with $ n_3''\lor n_5''\lor n_6'' \leq k+1$  can be either $C(1,1,1)$ or $C(1,2,0)$. By Theorem \ref{thm:linearcompare}-2), $C(1,1,1)$ has the same performance as $C(1,2,0)$. Hence, $C(1,1,1)$ and $C(1,2,0)$ are strictly better than their non-equivalent codes with only columns from  $\bspan{3}$, $\bspan{5}$ and $\bspan{6}$.

Third, consider  $n =  3k+1$ for a positive integer $k$.
For a code $C(n_3,n_5,n_6)$ with $n_3\geq  k+3$ and $n_3\geq n_5\geq n_6$, we can verify that either $C(n_3-1,n_5,n_6+1)$ or $ C(n_3-2,n_5,n_6+2) $ is  strictly better than $C(n_3,n_5,n_6)$ in the following four cases:
\begin{enumerate}
\item[3-1)] When $n_3,n_5$ and $n_6$ have the same parity, by Theorem \ref{thm:linearcompare}--2), $C(n_3-1,n_5,n_6+1)$  is  strictly better than $C(n_3,n_5,n_6)$.
\item[3-2)] When $n_3+1,n_5$ and $n_6$ have the same parity, by Theorem \ref{thm:linearcompare}--3), the code $C(n_3-1,n_5,n_6+1)$  is strictly better than $C(n_3,n_5,n_6)$.
\item[3-3)] When $n_3,n_5+1$ and $n_6$ have the same parity, we have $n_6\leq k-2$.
  By Corollary \ref{cor:linear1}--1), $C(n_3-2,n_5,n_6+2)$ is strictly better than $C(n_3,n_5,n_6)$. %
\item[3-4)] When $n_3,n_5$ and $n_6+1$ have the same parity, we have $n_6\leq k-2$ and $n_5\geq n_6+1$.
  By Corollary \ref{cor:linear1}--2), $C(n_3-2,n_5,n_6+2)$ is strictly better than $C(n_3,n_5,n_6)$. %
\end{enumerate}
Similar to the analysis when $n=3k,$ we can verify that  there exists a strictly better code $C(n_3',n_5',n_6')$ with $n_3'\lor n_5'\lor n_6'<n_3\lor n_5\lor n_6$  for the code $ C(n_3,n_5,n_6)$.
Then we can repeat the above argument until 
a strictly better code $C(n_3',n_5',n_6')$ with
$ n_3'\lor n_5'\lor n_6' \leq k+2$ is obtained.
Excluding the equivalent codes, there are five possibilities for $C(n_3',n_5',n_6')$ with $ n_3'\lor n_5'\lor n_6' \leq k+2$: $C(k+1,k,k)$, $C(k+1,k+1,k-1)$, $C(k+1,k+2,k-2)$, $C(k+2,k,k-1)$ and $C(k+2,k+2,k-3)$.  By Theorem \ref{thm:linearcompare}--1),
\begin{IEEEeqnarray*}{rCl}
  \lambda_{C(k+1,k,k)}&=&\lambda_{C(k+2,k,k-1)}, \\
  \lambda_{C(k+2,k+2,k-3)}&=&\lambda_{C(k+2,k+1,k-2)}=\lambda_{C(k+1,k+2,k-2)}.
\end{IEEEeqnarray*}
By Theorem \ref{thm:linearcompare}--2),
\begin{equation*}
 \lambda_{C(k+1,k+1,k-1)}<\lambda_{C(k+2,k,k-1)}, 
\end{equation*}
and by Theorem \ref{thm:linearcompare}--3),
\begin{equation*}
  \lambda_{C(k+1,k+2,k-2)}<\lambda_{C(k,k+2,k-1)} = \lambda_{C(k+2,k,k-1)}.
\end{equation*}
Therefore, $C(k+1,k,k)$ and $C(k+2,k,k-1)$ are strictly better than their non-equivalent codes with only columns from  $\bspan{3}$, $\bspan{5}$ and $\bspan{6}$. 

Now we consider general linear codes that can have $\bspan{0}$ columns. For a linear code $C$ with at least two $\bspan{0}$ columns, by Corollary~\ref{cor:0col}, $C$ is not optimal among all the linear codes.
For a linear code $C$ with exactly one $\bspan{0}$ column, we consider two cases: When $C$ is not equivalent to $C_0$  with $ |0|_{C_0}+ |5|_{C_0}+ |6|_{C_0}= n$ and $|5|_{C_0}, |6|_{C_0}$ odd, by Theorem~\ref{thm:0column}, there exists a strictly better linear code with no $\bspan{0}$ column. 
When $C$ is equivalent to $C_0$, WLOG, we suppose $C=C_0$. By Theorem~\ref{thm:0column}, we have $\ls_{C'} =\ls_C$ with $C'$ being obtained by changing the $\bspan{0}$ column in $C$ to $\bspan{3}$. Then $|3|_{C'}+|5|_{C'}+|6|_{C'}=n$  and $|3|_{C'}$, $|5|_{C'}$ and $|6|_{C'}$ are all odd. 
\begin{enumerate}
\item
When $n\neq 3$, we have shown that the best codes among $(n,4)$ codes with only columns from $\bspan{3}$, $\bspan{5}$ and $\bspan{6}$ are equivalent to some $C(r,s,t)$ where at least one of $r$, $s$ and $t$ is even. Thus, there exists a linear $(n,4)$ code that is strictly better than $C'$, and hence $C$ is not optimal among linear codes. 
\item 
When $n=3$, we have $|0|_{C}= |5|_{C}=|6|_{C}=1$, $C' = C(1,1,1)$ and $\ls_{C} = \ls_{C(1,1,1)}$. 
\end{enumerate}

Therefore, when  $n=3$, $C(1,1,1)$, $C(1,2,0)$ and $\begin{pmatrix}
	0&0&0\\
	0&1&1\\
	0&0&1\\
	0&1&0
\end{pmatrix}$  are universally optimal among all linear codes and strictly better than  other $(n,4)$ linear codes that are not equivalent to them. When $n\neq 3$, the codes with $\bspan{0}$ columns are not optimal among all linear codes.
  \end{IEEEproof}

\subsection{All Optimal Codes: Proof of Theorem \ref{thm:nbig3}}\label{sec:enhanced}

Note that Theorem~\ref{thm:0column}, Theorem~\ref{thm:even} and
Theorem~\ref{thm:odd} also have the necessary and sufficient condition
such that the two codes in comparison have the same performance. In
other words, these theorems also induce a strict partial order on all the
$(n,4)$ codes. Based on this partial order, together with Theorem~\ref{thm:linearopt}, we can eventually prove Theorem~\ref{thm:nbig3}.
We first define a new class of nonlinear code.

\begin{table}[t]
	\centering
	\caption{Definition of Class-III codes. }\label{tab:class3}

		\begin{tabular}{ccccccccccc}
			\hline
			\rowcolor[HTML]{EFEFEF} 
			class & subclass  & blocklength &$|0|$& $|1|$  & $|2|$  & $|3|$    & $|4|$  & $|5|$     & $|6|$     & $|7|$                        \\ \hline
			& a                                            & odd         &$0$   & $1  $     & $0$    & even   &$0 $   & odd     & $0  $     & $1$        
			\\ \cline{2-11} 
			\multirow{-2}{*}{III} 	&  b       & odd    &$0$   & 1                           & $0 $                       & odd                        & $0$                        & $0$                           & odd                         & $0$             \\ \hline
		\end{tabular}
	\end{table}

\begin{definition}
	An $ (n,4) $ code is said to be \emph{Class-III} if it satisfies $|1|+|3|+|5|+|6|+|7|=n$ and one of the following two conditions:
	\begin{enumerate}
		\item[a)] $ |1|=|7|=1$, $|6|=0 $, $ |3|$ is even and $|5| $ is odd;
		\item[b)]  $ |1|=1$, $|5|=|7|=0 $ and $ |3|$ and $|6| $ are odd.
	\end{enumerate}
\end{definition}
In Table~\ref{tab:class3}, we list the Class-III codes in two subclasses. Following the same analysis in \S\ref{sec:n4f} about linear, Class-I and Class-II codes,
it holds that  there are no equivalent codes belonging to two different classes (or subclasses) among   linear, Class-I, Class-II codes and Class-III codes.
The following lemma is about a relation among linear, Class-I, Class-II and Class-III codes.

\begin{lemma}\label{thm:class2}
  \begin{enumerate}
  \item For a Class-II code $C$, there exists a Class-I code $C'$ with $\ls_C = \ls_{C'}$ and $|1|_{C'} = |1|_{C}-1$.
  \item For a Class-III code $C$, there exists a linear code $C'$ with $\ls_C = \ls_{C'}$ and $|1|_{C'} = |1|_{C}-1$.
  \end{enumerate}
\end{lemma}
\begin{IEEEproof}
  \begin{enumerate}
  \item Let $ C$ be a Class-II $(n,4)$ code, and let $C'$ be obtained by changing one column of type $\bspan{1}$ of $C$ to $\bspan{3}$. If $ C$ is Class-II-a, then $C'$  is Class-I-b  and by the condition i) of equality in Theorem \ref{thm:even}, $C'$ has the same correct decoding probability as $ C$. If $ C$ is Class-II-b, then $C'$ is Class-I-a and by the condition i) of equality in Theorem \ref{thm:even} has the same correct decoding probability as $ C$.
  \item Let $ C$ be a Class-III $(n,4)$ code.
    When $C$ is of subclass a, let $ C''$ be obtained by replacing two columns of the types $\bspan{1}$  and $\bspan{7}$ of $C$ with $\bspan{3}$ and $\bspan{5}$. Then $C''$   is linear, and by Theorem~\ref{thm:odd}, $C'$ has the same correct decoding probability as $C$.
  When $C$ is of subclass b, $C'$ be obtained by changing one column of type $\bspan{1}$ of $C$ to $\bspan{3}$. Then $C''$ is linear, and by the condition ii) of equality in Theorem~\ref{thm:even}, $C'$ has the same correct decoding probability as $C$.
  \end{enumerate}
\end{IEEEproof}

In the following, we first argue that a set formed by linear,  Class-I, Class-II, Class-III codes and their equivalent codes contains all the optimal codes. We then apply Theorem \ref{thm:linearopt} and Lemma \ref{thm:class2} to further reduce the set that includes all the optimal codes, proving Theorem \ref{thm:nbig3}.

\begin{lemma}\label{lm:all}
  All
  the optimal codes among $(n,4)$ codes with only columns in
  $\{\bspan{1},\dots, \bspan{14}\}$ are equivalent to the linear codes, Class-I codes,
  Class-II codes or Class-III codes.
\end{lemma}
\begin{IEEEproof}
	  Let $\mc S$ be the set formed by the linear codes, Class-I codes,
	Class-II codes, Class-III codes and all their equivalent codes. 
  Let $C_0 \notin \mc S$ be an $(n,4)$ code   with only columns in $\{\bspan{1},\dots, \bspan{14}\}$, and let $C$ be the equivalent code of  $C_0$ obtained by flipping all the columns of type $\bspan{i}$, $i>7$. Then $C$ has only columns in $\{\bspan{1},\dots, \bspan{7}\}$.
   We will show there exists a strictly better code than $C$ for the following cases:
  \begin{enumerate}
  \item one of $|1|$, $|2|$, $|4|$ and $|7|$ is positive;
  \item at least two of $|1|$, $|2|$, $|4|$ and $|7|$ are positive.
  \end{enumerate}
Note that if $|1|$, $|2|$, $|4|$ and $|7|$ are all $0$, $C$ is linear.

For Case-1), we only argue the case $|1|>0$ and $ |2| = |4| = |7| = 0 $ since other cases can be transformed to this case by interchanging rows.
Referring to the proof of Corollary \ref{cor:onepo}, we see that at least one of  $w(\vv c_1\oplus \vv c_4)$, $w(\vv c_2\oplus \vv c_4)$ and $w(\vv c_3\oplus \vv c_4)$ are even due to $ C$ is non-Class-I. 
Here we assume $w(\vv c_3\oplus \vv c_4)$ is even since other cases can be transformed to this case by interchanging rows. Hence, $C$ does not satisfy all the following conditions:
\begin{itemize}
\item[i)] $w(\vv c_1\oplus \vv c_3)$ and $w(\vv c_2\oplus \vv c_3)$ are odd. Otherwise, $C$ is Class-II. 
\item[ii)] $ |1|=1$, $|5|=0 $, $ |3|$ and $|6| $ are odd. Otherwise, $C$ is Class-II-b.
\item[iii)] $ |1|=1$, $|6|=0 $, $ |3|$ and $|5| $ are odd. Otherwise, $C$ is equivalent to Class-II-b. 
\end{itemize}
Then by Theorem \ref{thm:even}, replacing a column of type $ \bspan 1 $ of $  C $ by $ \bspan 3 $ can give a strictly better code.

For Case-2), we first argue the case $|7|>0$ and $|1|>0$. Denote as Condition-A that $ |1|=|7|=1$, $|2|=|4|=|6|=0 $ and at least one of $ |3|$ and $ |5| $ is odd.
If Condition-A is not satisfied, by Theorem \ref{thm:odd}, changing columns $  \bspan 1$ and  $ \bspan 7 $ of $C$ to $ \bspan 3 $ and $ \bspan 5 $ can  give a strictly better code.
If Condition-$\mathrm A$ is satisfied, we have $|3|$ and $|5|$ are bot odd since $ C$ is not equivalent to a Class-III-a code.
Then we have $w(\vv c_3\oplus \vv c_4)$,  $w(\vv c_1\oplus \vv c_3)$ and $w(\vv c_2\oplus \vv c_3)$ are  all even, and hence by Theorem \ref{thm:even}, replacing a column of type $\bspan{1}$ of $ C$ by $\bspan{3}$  gives a strictly better code.

The case $|7|>0$ and $|2|>0$ or $|4|>0$ can be transformed to the case $|7|>0$ and $|1|>0$ by interchanging rows. 
Last, we consider the case $ |7|=0 $ and at least two of $ |1|,|2| $ and $ |4| $ are positive, which can be transformed equivalently to the case with $|7|>0$ and $|1|+|2|+|4|>0$. 
We give the transformation when $ |1| $ and $ |2| $ are positive and the other cases can be transformed to this case by interchanging rows. 
Let $ C' $ be the code obtained by flips all the columns of type $ \bspan 2,\bspan 3 $ and $ \bspan 6 $ in $C$ and then exchanging the third and  the first rows.
See Fig.~\ref{fig:iwla} for an illustration of this transformation. 
Observe that 
\begin{align*}
  |1| = |1|_{C'}, |2| = |7|_{C'} , |3| = |6|_{C'}, |4| = |4|_{C'}, |5| = |5|_{C'}, |6| = |3|_{C'}.
\end{align*}
Then we have $|7|_{C'} >0$ and $|1|_{C'}>0$. The proof is completed. 
\end{IEEEproof}
              
\begin{figure}
  \centering
  \begin{IEEEeqnarray*}{rCl}
    \begin{pNiceMatrix}[last-row]
      0 & \color{blue}0 &\color{blue}0 &0 &0 &\color{blue}0\\
      0&\color{blue}0&\color{blue}0&1&1&\color{blue}1\\
      0&\color{blue}1&\color{blue}1&0&0&\color{blue}1\\
      1&\color{blue}0&\color{blue}1&0&1&\color{blue}0\\
      \bspan{1} & \bspan{2} & \bspan{3} & \bspan{4} & \bspan{5} & \bspan{6} 
    \end{pNiceMatrix}
    &\xrightarrow{\text{flip col } 2,3,6 }&
    \begin{pNiceMatrix}[last-row]
      0 & \color{red}1  &\color{red}1  &0 &0 &\color{red}1\\
      0 &\color{red}1   &\color{red}1 &1&1    &\color{red}0\\
      0 &\color{red}0  &\color{red}0 &0&0   &\color{red}0\\
      1 &\color{red}1   &\color{red}0 &0&1   &\color{red}1\\
      \bspan 1 & \bspan{13} & \bspan{12} & \bspan 4 & \bspan 5 & \bspan 9
    \end{pNiceMatrix}\\
    &\xrightarrow{\text{interchange row } 1,2 }&
    \begin{pNiceMatrix}[last-row]
      \color{green}0 &\color{green}0  &\color{green}0 &\color{green}0&\color{green}0   &\color{green}0\\
      0 & 1  &1  &1 &1 &0\\
      \color{green}0 &\color{green}1   &\color{green}1 &\color{green}0&\color{green}0    &\color{green}1\\
      1 &1   &0 &0&1   &1 \\
      \bspan 1 & \bspan{7} & \bspan{6} & \bspan 4 & \bspan 5 & \bspan 3
    \end{pNiceMatrix}
  \end{IEEEeqnarray*}
  \caption{Illustration of the code transformation used in the proof of Lemma~\ref{lm:all}. Suppose $C$ has only one column for each type $\bspan{1}, \ldots, \bspan{6}$. Then $ C'$
		is obtained by flipping all the columns of type $ \bspan 2$, $\bspan 3 $ and $ \bspan 6 $ and interchanging the third row and  the first row. }
  \label{fig:iwla}
\end{figure}

\begin{IEEEproof}[Proof of Theorem \ref{thm:nbig3}]
First, consider $n=2$. By Theorem~\ref{thm:0column}, $(2,4)$ codes with $\bspan 0$ or $\bspan{15}$ are not optimal. By definition, there is no $(2,4)$ codes of Class-I/II/III. Hence, by Lemma~\ref{lm:all}, all the optimal $(2,4)$ codes are equivalent to linear codes.
Last, by Theorem~\ref{thm:linearopt}, $C(1,1,0)$ is universally strictly optimal among all $(2,4)$ linear codes.

Second, consider $n=3$. By Theorem~\ref{thm:0column}, if a $(3,4)$ code with a column $\bspan 0$ or $\bspan{15}$ is optimal, it must be equivalent to $C_A$ defined in \eqref{eq:0columnlinear}, and has the same performance as $C(1,1,1)$ and $C(1,2,0)$, which, and their equivalent codes, are actually all the optimal linear codes by Theorem~\ref{thm:linearopt}.
When $n=3$, all the possible Class-III codes are equivalent to $C_1$ with $|1|_{C_1}=|5|_{C_1}=|7|_{C_1}=1$ or $C_2$ with $|1|_{C_2}=|3|_{C_2}=|6|_{C_2}=1$.
By Theorem~\ref{thm:odd}, we have $\ls_{C_1}=\ls_{C(1,2,0)}$  and by Theorem~\ref{thm:even}, we have $\ls_{C_2}=\ls_{C(1,2,0)}$.  
When $n=3$, all the possible Class-I codes are equivalent to a code $C_3$ with $|1|_{C_3}=1$ and $|3|_{C_3}=2$ or a code $C_4$ with $|1|_{C_4}=3$. By Theorem~\ref{thm:even}, we have $\ls_{C(1,2,0)}>\ls_{C_3}>\ls_{C_4}$. Then all the Class-I codes are not optimal. Last, 
 by Lemma~\ref{thm:class2}, Class-II codes are also not optimal since each of them has the same performance as some Class-I code. 
Therefore, all the optimal $(3,4)$ are equivalent to $C_A$, $C(1,1,1)$, $C(1,2,0)$, $C_1$ and $C_2$.

Third, consider $n>3$. By Theorem~\ref{thm:0column} and Corollary~\ref{cor:0col}, $(n,4)$ codes with $\bspan 0$ or $\bspan{15}$ are not optimal.
Based on Lemma~\ref{lm:all}, we only need to show that a Class-III code $C$ is not optimal. By Lemma \ref{thm:class2}, there exists a linear $(n,4)$ code $C'$ that has the same correct decoding probability as $C$.
By further checking the proof of Lemma \ref{thm:class2}, we can find such a $C'$ without type $\bspan{5}$ or $\bspan{6}$ columns.
When $n>3$, $C'$ cannot be an optimal code by Theorem~\ref{thm:linearopt} since the blocklength of a Class-III code is odd. The proof is completed.
      \end{IEEEproof}

\subsection{More Results about Class-I Codes and Proof of Theorem \ref{thm:n8}}

Up to here, the optimal $(n,4)$ code problem has been completely solved for $n=2$ and $3$. For $n\geq 4$, Theorem~\ref{thm:nbig3} and Lemma~\ref{thm:class2} reduce the essential problem to whether a Class-I code is optimal. Here we present some further results about Class-I codes, which enable us to prove Theorem \ref{thm:n8} that exactly characterizes all the optimal codes for $n$ up to $300$. 

The following lemma presents a sufficient condition for the
existence of linear codes that are optimal among all codes.

\begin{lemma}[Sufficient Condition for Existence of Optimal Codes]\label{thm:ccs}
  Fix a blocklength $n$.  If for any Class-I $(n,4)$ code $C$, there
  exists an $(n,4)$ code $C'$ such that $|1|_{C'}<|1|_C$,
  $|1|_{C'}+|3|_{C'}+|5|_{C'}+|6|_{C'} = n$ and
  $\lambda_{C} \leq \lambda_{C'}$, then there exists an optimal
  $(n,4)$ code that is linear. 
\end{lemma}
\begin{IEEEproof}
  Fix an optimal $(n,4)$ code $C$ that is Class-I. If such a code does
  not exist, by Theorem~\ref{thm:two}, there must exist an optimal
  $(n,4)$ code that is linear, and the proof is done. According to the
  statement of this theorem, there exists an optimal code $C'$ such
  that $|1|_{C'}< |1|_C $ and
  $|1|_{C'}+|3|_{C'}+|5|_{C'}+|6|_{C'} = n$. If $|1|_{C'}=0$, then
  $C'$ is linear, and the proof is done. If $C'$ is Class-I, we
  repeat the above argument. If $C'$ is non-Class-I and nonlinear,
  then by Corollary~\ref{cor:onepo}, there exists an optimal code
  $C''$ with $|1|_{C''}<|1|_{C'}$ that is either linear or Class-I. If
  $C''$ is linear, the proof is done. If $C''$
  is Class-I, we can repeat the above argument. As $|1|_C$ is finite,
  the process will eventually stop with an optimal linear code.
\end{IEEEproof}

 In the following two theorems, it will be shown that there always exists a better code with more linear-type columns for the Class-I codes $ C$ that have only one $\bspan{1}$ column, or have at most one column of some  linear-type, i.e. $\min_{i = 3,5,6}|i|_C = 0$ or $1$.
\begin{theorem}\label{thm:11}
	Let $C$ be a Class-I $(n,4)$ code with $|1|_C=1$. Let $C'$ be the code obtained by replacing the $\bspan{1}$ column of $C$ by $\bspan{s}$, where $s= \argmin_{i = 3,5,6}|i|_C$. Then $\lambda_{C'} > \lambda_{C}$ when $n\neq 3$, and $ \lambda_{C'} = \lambda_{C}$ when $n=3$.
\end{theorem}
\begin{IEEEproof}
  See \S\ref{sec:11}.
\end{IEEEproof}

In the above theorem, code $C'$ is linear.

\begin{theorem}\label{thm:301}
	Let $C$ be a Class-I $(n,4)$ code with $\min_{i = 3,5,6}|i|_C = 0$ or $1$. Let $C'$ be the code obtained by replacing one $\bspan{1}$ column of $C$ by $\bspan{s}$, where $s\in\{3,5,6\}$ has $|s|_C=0$ or $1$. Then $\lambda_{C'}\geq \lambda_{C}$.
\end{theorem}
\begin{IEEEproof}
  See \S\ref{sec:301}.
\end{IEEEproof}

\begin{theorem}\label{thm:18}
  For each Class-I $(n,4)$ code $C$ with $|1|_C\geq 3$, let $ C'$ be the code obtained by
replacing one $\bspan{1}$ column of $ C$ by $\bspan{s}$ with
$s=\argmin_{i=3,5,6} |i|_C$. Then for $n\leq 300$,
$\lambda_{C} \leq \lambda_{C'}$.
\end{theorem}
\begin{IEEEproof}
 This theorem is proved using computer evaluations. 
  See \S\ref{sec:comp}.
\end{IEEEproof}

Using  Lemma~\ref{thm:ccs}, and Theorem~\ref{thm:18}, we can show the existence of an optimal code that is linear for $n\leq 300$. We improve the technique to derive a stronger result that all the optimal codes are equivalent to linear codes for $4\leq n \leq 300$.

\begin{IEEEproof}[Proof of Theorem \ref{thm:n8}]
  By Theorem~\ref{thm:11}, for any Class-I $(n,4)$ code $C$ with $n\geq 4$ and $|1|_C=1$, there
  exists an $(n,4)$ linear code $C'$ such that $\lambda_{C} < \lambda_{C'}$.
  By Theorem \ref{thm:nbig3} and Lemma~\ref{thm:class2},
  if we can further show that for any Class-I $(n,4)$ code $C$ with $|1|_C>1$, there
  exists a Class-I code $C'$ such that $|1|_{C'}=1$ and
  $\lambda_{C} \leq \lambda_{C'}$, 
  then all the optimal $(n,4)$ codes are equivalent to linear codes.

  We first consider $4\leq n\leq 8$. For a Class-I $(n,4)$ code $C$ with $|1|\geq 3$, we have $|3|+|5|+|6|\leq 5$ which implies $\min\{|3|, |5|, |6|\}\leq 1$. By Theorem~\ref{thm:301},
  we have $\lambda_{C} \leq \lambda_{C'}$ for the $(n,4)$ code $C'$ obtained by replacing one $\bspan{1}$ column of $C$ by $\bspan{s}$, where $s\in\{3,5,6\}$ with $|s|_C=0$ or $1$.
  Note that $C'$ is equivalent to a Class-II code. By Lemma~\ref{thm:class2}, there exists a Class-I code $C''$ with $\ls_{C''} = \ls_{C'} \geq \ls_C$ and $|1|_{C''} = |1|_C -2$. Repeat the argument, we can find a Class-I code $C''$ such that $|1|_{C''}=1$ and $\lambda_{C} \leq \lambda_{C''}$.  
Hence, all the optimal $(n,4)$ codes are equivalent to linear codes.

When $8< n \leq 300$, by Theorem~\ref{thm:18}, following the same argument for $n\leq 8$, we can derive that all the optimal $(n,4)$ codes are equivalent to linear codes.
\end{IEEEproof}

\section{An Approach of Comparing Two $(n,4)$ Codes}
\label{sec:formulation}

In this section, we study the ML decoding performance comparison of two $(n,4)$ codes, and prove Theorem~\ref{thm:even}. %

\subsection{Code Comparison Problem Formulation}
\label{sec:compgen}

Following the formulation of $(n,4)$ codes in \S\ref{sec:formulationresults}, 
we further define some notations.
For a binary vector $\vv y$, denote $(\vv y)_i$ or $y_i$ as the $i$th entry of $\vv y$. For example, the $3$rd entry of $\bspan{2}=[0\ 0\ 1\ 0]^\top$ is $(\bspan{2})_3=1$. 
Let $C$ be an $(n,4)$ code with the $j$th codeword/row $\vv c_j$, $j=1,\ldots,4$. We use $\{i\}_C$ to denote the index set of the columns of $C$ equal to $\bspan{i}$.
For $i=0,1,\ldots,15$, define $w_i(\vv y) = \sum_{j\in \{i\}_C}  y_j$ for $\vv y \in \{0,1\}^n$.  Let $\overline{w_i} (\cv y)= |i|_C - w_i(\cv y)$. When $\vv y$ is clear from the context, we write $w_i = w_i(\vv y)$  and $\overline{w_i}=\overline{w_i}(\cv y)$.
For a vector $\vv y\in \{0,1\}^n$, denote
\begin{equation} %
  d_j(\vv y)  =  w(\vv c_j \oplus \vv y) \label{eq:d}
   = \sum_{i=0}^{15} \hat{w}_{ij} %
 \end{equation}
 where
 \begin{equation*}
   \hat{w}_{ij}=
   \begin{cases}
     w_i ,& (\bspan{i})_j = 0, \\
     \overline{w_i}, & (\bspan{i})_j = 1.
   \end{cases}
 \end{equation*}
 We also write $d_j = d_j(\vv y)$ when $\vv y$ is clear from the context.

 \begin{example}\label{ex:dw}
For example, when $C$ only has only the columns of types  $\bspan{0},\bspan{1},\ldots,\bspan{7}$, 
\begin{IEEEeqnarray}{rCl}
  d_1(\vv y) & = & w_0+w_1+w_2+w_3+w_4+w_5+w_6+w_7, \label{eq:d1}\\
  d_2(\vv y) & = & w_0+w_1+w_2+w_3+\overline{w_4}+ \overline{w_5}+\overline{w_6}+\overline{w_7}, \label{eq:d2} \\
  d_3(\vv y) & = & w_0+w_1+\overline{w_2}+\overline{w_3}+ w_4+w_5+\overline{w_6}+\overline{w_7}, \label{eq:d3} \\
  d_4(\vv y) & = & w_0+\overline{w_1}+w_2+\overline{w_3}+ w_4+\overline{w_5}+w_6+\overline{w_7}, \label{eq:d4}
\end{IEEEeqnarray}    
 \end{example}

We compare $C$ with another
$(n,4)$ code $C'$ obtained by modifying $C$ as follows.  
Let $\mc O$ be a nonempty, proper subset of $\{1,2,3,4\}$
and let $\mc P$ be its complement, which is also nonempty.
Let $C'$ be the code obtained by
flipping the first $t$ bits of $\vv c_i$ for each $i\in \mc P$.
Denote by $\vv c'_i$ the $i$th codeword/row of $C'$, $i=1,\ldots,4$.
For $\vv y \in \{0,1\}^n$, let $F_t(\vv y)$ be the vector obtained by
flipping the first $t$ bits of $\vv y$.  We see that
$\vv c_i'=\vv c_i$ for $i\in \mc O$ and $\vv c_i' = F_t(\vv c_i)$ for
$i\in \mc P$.

The ML decoding performance depends on the function $d_C(\vv y)$ defined in \eqref{eq:dcy}.
Denote by
$\vv s_{\tau}$, $\tau = 1,2,\ldots,t$ the $\tau$th column of
$C$.  For $\vv y \in \{0,1\}^n$, let
\begin{equation}\label{eq:dprim}
  d_i'(\vv y) = d_i(F_t(\vv y)) = d_i(\vv y) + \sum_{\tau=1}^t (-1)^{(\vv s_{\tau})_i} (\overline{y_\tau}- y_\tau),
\end{equation}
where $ \overline{y_\tau} = y_\tau\oplus 1$.
For a nonempty subset $\mc S \subset \{1,\ldots, 4\}$, let
\begin{equation*}
  d_{\mc S}(\vv y)  =  \min_{i\in \mc S} d_i(\vv y) \ \ \text{and}\ \ 
  d_{\mc S}'(\vv y)  =  \min_{i\in \mc S} d_i'(\vv y).
\end{equation*}
We have
\begin{IEEEeqnarray}{rCl}
  d_C(\vv y) & = & \min\{d_{\mc O}(\vv y), d_{\mc P}(\vv y)\}, \label{eq:dc} \\
  d_{C'}(\vv y) & = &  \min\{d_{\mc O}(\vv y), d_{\mc P}'(\vv y)\}, \label{eq:dcp} \\
  d_{C'}(F_t(\vv y)) & = & \min\{d_{\mc O}(F_t (\vv y)), d_{\mc P}'(F_t(\vv y))\} \nonumber\\
  & = & \min\{d_{\mc O}'(\vv y), d_{\mc P}(\vv y)\}. \label{eq:dcpf}
\end{IEEEeqnarray}

In the following lemma, we demonstrate a special case of our technique to compare two codes.

\begin{lemma}\label{thm:com}
	For two $ (n,4) $ codes $ C $ and $ C' $,  if there exists a subset $ \mc S\subseteq \{0,1\}^n $ and an one-to-one and onto mapping $ g:\{0,1\}^n\to \{0,1\}^n $ such that $ d_C(\vv y)> d_{C'}(g(\vv y)) $ for $ \vv y\in \mc S$ and $ d_ C(\vv y)= d_{C'}(g(\vv y)) $ for $ \vv y\in \{0,1\}^n\setminus\mc S$, we have $ \lambda_{ C'}\geq \lambda_{ C} $. 
	 Moreover, $ \lambda_{ C'}>\lambda_{C} $ if $ \mc S\neq \emptyset $ and $ \lambda_{ C'}=\lambda_{ C} $  if $ \mc S= \emptyset $ .
\end{lemma}
\begin{IEEEproof}
  Since $g$ is an one-to-one and onto mapping,
  by \eqref{eq:lambda:1},
  \begin{align*}
    \lambda_{ C'} - \lambda_{C}
    &= \frac{1}{4}\sum_{\vv y\in \{0,1\}^n}  \big(1-\epsilon\big)^{n-d_{C'}(\vv y)}\epsilon^{d_{ C'}(\vv y)} - \frac{1}{4}\sum_{\vv y\in \{0,1\}^n}  \big(1-\epsilon\big)^{n-d_C(\vv y)}\epsilon^{d_ C(\vv y)}\\
    &=  \frac{1}{4}\sum_{\vv y\in \{0,1\}^n} \bigg(\Big(1-\epsilon\Big)^{n-d_{ C'}\big(g(\vv y)\big)}\epsilon^{d_{C'}\big(g(\vv y)\big)} -\big(1-\epsilon\big)^{n-d_{ C}(\vv y)}\epsilon^{d_{ C}(\vv y)}    \bigg) \\
    & = \frac{1}{4}\sum_{\vv y\in \mc S} \bigg(\Big(1-\epsilon\Big)^{n-d_{ C'}\big(g(\vv y)\big)}\epsilon^{d_{C'}\big(g(\vv y)\big)} -\big(1-\epsilon\big)^{n-d_{ C}(\vv y)}\epsilon^{d_{ C}(\vv y)}    \bigg).
  \end{align*}
  As for $\vv y\in \mc S$, $ d_C(\vv y)> d_{C'}(g(\vv y)) $ and  $ (1-\epsilon)^{n-x}\epsilon^x $ is strictly decreasing over $ x $ when $ 0<\epsilon <1/2 $, $ \lambda_{ C'}\geq \lambda_{ C} $ with equality if and only if $\mc S = \emptyset$. 
\end{IEEEproof}

In general, our approach to compare the ML decoding performance of $C$
and $C'$ is based on an one-to-one and onto mapping
$ g:\{0,1\}^n\to \{0,1\}^n $ and a partition
$\{\mc Y_i,i=1,\ldots,i_0\}$ of $\{0,1\}^n$, where $i_0$ indicates the
number of subsets in the partition. The mapping $g$ and the partition satisfy
the following property: for each $i=1,\ldots,i_0$,  one of the following
conditions holds:
\begin{enumerate}
\item for all $\vv y \in \mc Y_i$, $d_C(\vv y) = d_{C'}\left(g(\vv y)\right)$;
\item for all $\vv y \in \mc Y_i$, $d_C(\vv y) < d_{C'}\left(g(\vv y)\right)$;
\item for all $\vv y \in \mc Y_i$, $d_C(\vv y) > d_{C'}\left(g(\vv y)\right)$.
\end{enumerate}
Such a mapping $g$ and a partition always exist. For example, the identity mapping $g(\vv y) = \vv y$ and the partition including only the singleton sets. But this example does not help to simplify the problem. For the two special cases used to prove Theorem~\ref{thm:even} and~\ref{thm:odd}, there exists such a partition with $i_0=5$.
Lemma~\ref{thm:com} applies to the case that all the partitions satisfy only conditions 1) and 3). 

In the remainder of this paper, we will discuss two ways of generating $C'$ with $t=1$ and $2$, respectively. 
We write $\min\{a,b\}$ as $a \land b$.
For a function $h:\{0,1\}^n\rightarrow \mathbb{R}$, we write $\{\vv y \in \{0,1\}^n: h(\vv y) \geq 0\}$ as $\{h\geq 0\}$ to simplify the notations.

\subsection{Change of One Column}
\label{sec:1col}

We study how the ML decoding performance is affected after changing one column of an $(n,4)$ code. 
Consider an $(n,4)$ code $C$ with the first column $\bspan{s}$, $0\leq s \leq 15$. Let $C'$ be the code formed by changing the first column of $C$ to $\bspan{s'}$, $s'\neq s$. Let $\mc O$ be the set of index $j$ such that $(\bspan{s})_j = (\bspan{s'})_j$, and $\mc P$ be the set of index $j$ such that $(\bspan{s})_j \neq (\bspan{s'})_j$. When $s'=15-s$, the bits in the first column are all flipped and hence $C$ and $C'$ are equivalent. 
Assume $s'\neq s$ and $s'\neq 15-s$, and hence both $\mc O$ and $\mc P$ are nonempty. 
In this case, $d_i'$ defined in \eqref{eq:dprim} becomes
\begin{equation}\label{eq:d1di}
  d_i'(\vv y) = d_i(\vv y) + (-1)^{(\bspan{s})_i}(\overline{y_1} - y_1).
\end{equation}

\begin{example}\label{ex:s1s3}
Consider an example with
$s=1$ and $s'=3$. Now $\mc O=\{1,2,4\}$ and $\mc P=\{3\}$. Substituting $\bspan{1}$ into \eqref{eq:d1di}, 
\begin{IEEEeqnarray*}{rCl}
	d_1'(\vv y)&=& d_1(\vv y) - y_1 + \overline{y_1}, \label{eq:s1d1}  \\
	d_2'(\vv y)&=& d_2(\vv y) - y_1 + \overline{y_1},  \\
	d_3'(\vv y)&=& d_3(\vv y) - y_1 + \overline{y_1},  \\
	d_4'(\vv y)&=& d_4(\vv y) + y_1 - \overline{y_1}. \label{eq:s1d4}
\end{IEEEeqnarray*}
and hence
\begin{IEEEeqnarray}{rCl}
  d_{\mc O}(\vv y) & = & d_1 \land d_2 \land d_4,  \label{eq:op1}\\
  d_{\mc P}(\vv y) & = & d_3, \\
  d_{\mc O}'(\vv y) & = & [(d_1\land d_2)- y_1 + \overline{y_1}]\land (d_4 + y_1 - \overline{y_1}), \\
  d_{\mc P}'(\vv y) & = & d_3 - y_1 + \overline{y_1}. \label{eq:op4}
\end{IEEEeqnarray}
\end{example}

The crucial part of our technique for comparing $C$ and $C'$ is
the following $5$ subsets of $\{0,1\}^n$:
\begin{IEEEeqnarray}{rCl}
  \mc Y_1 & = & \{ d_{\mc O} \leq d_{\mc P} < d_{\mc P}' \} \cup \{ d_{\mc O} \leq d_{\mc P}' \leq d_{\mc P}, d_{\mc O}' \leq d_{\mc P}' \},  \label{eq:y1}\\
  \mc Y_2 & = & \{ d_{\mc P} \leq d_{\mc P}', d_{\mc P} < d_{\mc O}\} \cup  \{ d_{\mc P}' < d_{\mc P} \leq d_{\mc O}, d_{\mc P} \leq d_{\mc O}'\}, \\
  \mc Y_3 & = & \{d_{\mc P}' = d_{\mc O}' < d_{\mc P} = d_{\mc O}\},  \label{eq:y3}  \\
  \mc Y_4 & = & \{d_{\mc P} = d_{\mc P}' = d_{\mc O} < d_{\mc O}' \}, \\
  \mc Y_5 & = & \{ d_{\mc P}' = d_{\mc O} < d_{\mc O}' = d_{\mc P} \}.  \label{eq:y5}
\end{IEEEeqnarray}
Recall that $F_1$ (defined in \S\ref{sec:compgen}) flips the first bit of a binary vector.
Define a mapping $g_1: \{0,1\}^n \rightarrow \{0,1\}^n$ as
\begin{equation*}
  g_1(\vv y)  =
  \begin{cases}
    \vv y  & \vv y \in \mc Y_1 \cup \mc Y_3, \\
    F_1(\vv y) & \text{otherwise}.
  \end{cases}
\end{equation*}
The next lemma shows that $\{\mc Y_1, \mc Y_2, \mc Y_3, \mc Y_4, \mc Y_5\}$ and $g_1$ satisfy the properties described in \S\ref{sec:compgen} for $C$ and $C'$.

\begin{lemma}\label{lemma:1}
  $\{\mc Y_1, \mc Y_2, \mc Y_3, \mc Y_4, \mc Y_5\}$ defined in \eqref{eq:y1}--\eqref{eq:y5} forms a partition of $\{0,1\}^n$, and $g_1$ is a one-to-one and onto mapping. Moreover, for the $(n,4)$ codes $C$ and $C'$ formulated above with only the first column different,  
  \begin{enumerate}
  \item for $\vv y \in \mc Y_1$, $d_C(\vv y) = d_{C'}(\vv y)=d_{\mc O}$;
  \item for $\vv y \in \mc Y_2$, $d_C(\vv y) = d_{C'}(\vv y') = d_{\mc P}$ where $\vv y' \triangleq F_1(\vv y)$;
  \item for $\vv y \in \mc Y_3$, $d_C(\vv y) = d_{\mc P} = d_{C'}(\vv y)+1=d_{\mc P}'+1$;
  \item for $\vv y \in \mc Y_4$, $d_C(\vv y) = d_{\mc O} = d_{C'}(\vv y')=d_{\mc P}$ where $\vv y'\triangleq F_1(\vv y)$;
  \item for $\vv y \in \mc Y_5$, $d_C(\vv y) +1 = d_{\mc O}+1 = d_{C'}(\vv y')=d_{\mc P}$ where $\vv y'\triangleq F_1(\vv y)$.
  \end{enumerate}
\end{lemma}
\begin{IEEEproof}
  By checking the definition, we see that $\mc Y_1,\ldots \mc Y_5$ are all disjoint. To show they form a partition, we can verify that
	\begin{IEEEeqnarray*}{rCl}
		\mc Y_1 \cup \mc Y_4 \cup \mc Y_5 & = &  \{d_{\mc O} \leq d_{\mc P}\land d_{\mc P}'\}, \\
		\mc Y_2 \cup \mc Y_3 & = &  \{d_{\mc O} > d_{\mc P}\land d_{\mc P}'\}
	\end{IEEEeqnarray*}
	and hence $(\mc Y_1 \cup \mc Y_4 \cup \mc Y_5)\cup(\mc Y_2 \cup \mc Y_3) = \{0,1\}^n$.
	
	We first prove that $\mc Y_1 \cup \mc Y_4 \cup \mc Y_5  =  \{d_{\mc O} \leq d_{\mc P}\land d_{\mc P}'\}$, where the three sets can be rewritten as follows: First,
	\begin{IEEEeqnarray}{rCl}
          \mc Y_1 &=&  \{ d_{\mc O} \leq d_{\mc P} < d_{\mc P}' \} \cup \{ d_{\mc O} \leq d_{\mc P}' \leq d_{\mc P}, d_{\mc O}' \leq d_{\mc P}' \}\nonumber\\
          &=&(\{ d_{\mc O}\leq d_{\mc P}\land d_{\mc P}'\}\cap \{  d_{\mc P} < d_{\mc P}'  \})\cup  (\{ d_{\mc O}\leq d_{\mc P}\land d_{\mc P}'\} %
          \cap \{d_{\mc P}' \leq d_{\mc P}, d_{\mc O}' \leq d_{\mc P}' \}  ) \label{eq:lemma5a}.
        \end{IEEEeqnarray}
        For all $\vv y\in \{0,1\}^n$, we have
	\begin{equation}
          |d_{\mc S}(\vv y) -d_{\mc S}'(\vv y)| \leq 1.  \label{eq:lemma5d}
	\end{equation}
        Hence, if $d_{\mc O}\leq d_{\mc P}'<d_{\mc O}'$, then $d_{\mc O}= d_{\mc P}'$. So
        \begin{IEEEeqnarray}{rCl}
          \mc Y_4&=& \{d_{\mc P} = d_{\mc P}' = d_{\mc O} < d_{\mc O}' \}\nonumber\\
          & = & \{ d_{\mc O}\leq d_{\mc P}\land d_{\mc P}'\}\cap \{d_{\mc P}'=d_{\mc P}, d_{\mc O}'>d_{\mc P}' \}.
          \label{eq:lemma5b}
        \end{IEEEeqnarray}
        Furthermore, if $d_{\mc O}'>d_{\mc P}' $ we have $d_{\mc O}\geq d_{\mc P}'$. So
        \begin{IEEEeqnarray}{rCl}
          \mc Y_5 &=&\{ d_{\mc P}' = d_{\mc O} < d_{\mc O}' = d_{\mc P} \} \nonumber\\
          & = & \{ d_{\mc O}\leq d_{\mc P}\land d_{\mc P}'\}\cap \{d_{\mc P}' < d_{\mc P},d_{\mc O}'>d_{\mc P}'\}.
        \label{eq:lemma5c}
	\end{IEEEeqnarray}
	By \eqref{eq:lemma5b} and \eqref{eq:lemma5c}, we have
	\begin{IEEEeqnarray}{rCl}
		\mc Y_4\cup\mc Y_5 &=& \{ d_{\mc O}\leq d_{\mc P}\land d_{\mc P}'\} \cap \{d_{\mc P}' \leq d_{\mc P},d_{\mc O}'>d_{\mc P}' \}.\nonumber
	\end{IEEEeqnarray}
	From \eqref{eq:lemma5a}, this further implies 
	\begin{IEEEeqnarray*}{rCl}
		\mc Y_1 \cup \mc Y_4 \cup \mc Y_5  =  \{d_{\mc O} \leq d_{\mc P}\land  d_{\mc P}'\}.
	\end{IEEEeqnarray*}
	
	Similarly, we can prove $\mc Y_2\cup \mc Y_3=\{d_{\mc O} > d_{\mc P}\land  d_{\mc P}'\}$ by rewriting the two sets as follows: 
	\begin{IEEEeqnarray*}{rCl}
		\mc Y_2 &=& \{ d_{\mc P} \leq d_{\mc P}', d_{\mc P} < d_{\mc O}\} \cup  \{ d_{\mc P}' < d_{\mc P} \leq d_{\mc O}, d_{\mc P} \leq d_{\mc O}'\}\nonumber \\
                & = & (\{d_{\mc O} > d_{\mc P}\land  d_{\mc P}'\}\cap \{d_{\mc P} \leq d_{\mc P}' \}  )\cup ( \{d_{\mc O} > d_{\mc P}\land  d_{\mc P}'\} \cap \{d_{\mc P} > d_{\mc P}' , d_{\mc P} \leq d_{\mc O}'\}   ), \label{eq:lemma5e}\\
		\mc Y_3&=& \{d_{\mc P}' = d_{\mc O}' < d_{\mc P} = d_{\mc O}\}\nonumber\\
                & = & \{d_{\mc O} > d_{\mc P}\land  d_{\mc P}'\}\cap \{d_{\mc P} > d_{\mc P}' , d_{\mc P} > d_{\mc O}'\}.\label{eq:lemma5f} 
	\end{IEEEeqnarray*}

        For $i=2,4,5$, define
        $\mc Y_i'  =  \left\{F_1(\vv y):\vv y\in \mc Y_i\right\}$, which can be 
        rewritten as
	\begin{IEEEeqnarray*}{rCl}
		\mc Y_2' & = & \{ d_{\mc P}' \leq d_{\mc P}, d_{\mc P}' < d_{\mc O}'\} \cup \{ d_{\mc P} < d_{\mc P}' \leq d_{\mc O}', d_{\mc P}' \leq d_{\mc O}\}, \\
		\mc Y_4' & = & \{d_{\mc P} = d_{\mc P}' = d_{\mc O}' < d_{\mc O}\}, \\
		\mc Y_5' & = & \{ d_{\mc P} = d_{\mc O}' < d_{\mc O}= d_{\mc P}'\}.
	\end{IEEEeqnarray*}
	It can be verified that $(\mc Y_2'\cup \mc Y_4'\cup \mc Y_5')\cap (\mc Y_1\cup \mc Y_3) = \emptyset$.
	As $F_1$ is a one-to-one mapping, $\mc Y_2'\cup \mc Y_4'\cup \mc Y_5' = \mc Y_2\cup \mc Y_4\cup \mc Y_5$. Hence, we conclude that $g_1$ is a one-to-one and onto mapping. %
	
	We use the following facts in the remaining part of the proof  (ref.~\eqref{eq:dc}--\eqref{eq:dcpf}):
	\begin{IEEEeqnarray*}{rCl}
          d_C(\vv y) & = & \min\{d_{\mc O}(\vv y), d_{\mc P}(\vv y)\} \\
          d_{C'}(\vv y) & = &  \min\{d_{\mc O}(\vv y), d_{\mc P}'(\vv y)\}\\
          d_{C'}(F_1(\vv y)) & = &  \min\{d_{\mc O}'(\vv y), d_{\mc P}(\vv y)\}.
        \end{IEEEeqnarray*}
        The claims 1)--5) can be proved as follows:
        \begin{itemize}
        \item  
	For $\vv y \in \mc Y_1$, as $d_{\mc O}(\vv y) \leq \min\{d_{\mc P}(\vv y), d_{\mc P}'(\vv y)\}$, $d_C(\vv y) = d_{C'}(\vv y)=d_{\mc O}(\vv y)$.
      \item
	For $\vv y \in \mc Y_2 $, as $d_{\mc P}(\vv y) \leq \min\{d_{\mc O}(\vv y),d_{\mc O}'(\vv y)\}$, $d_C(\vv y) = d_{C'}(\vv y')= d_{\mc P}(\vv y)$.
      \item 
	For $\vv y\in \mc Y_3$, $d_C(\vv y)=d_{\mc O}(\vv y)\land d_{\mc P}(\vv y)=d_{\mc P}(\vv y)$. Hence,
	\begin{IEEEeqnarray*}{rCl}
		d_{ C'}(\vv y)&=& d_{\mc O}(\vv y)\land  d_{\mc P}'(\vv y)=d_{\mc P}'(\vv y)<d_{C}(\vv y).
	\end{IEEEeqnarray*}
	By \eqref{eq:lemma5d}, we have $d_{\mc P(\vv y)}=d_{\mc P}'(\vv y)+1$.
      \item For $\vv y\in \mc Y_4$, as $d_{\mc P}(\vv y) = d_{\mc O}(\vv y) < d'_{\mc O}(\vv y)$, $d_C(\vv y) = d_{\mc O}(\vv y) = d_{\mc P}(\vv y) = d_{C'}(\vv y')$. 
      \item 
	For $\vv y\in \mc Y_5$, $d_C(\vv y)=d_{\mc O}(\vv y)\land d_{\mc P}(\vv y)=d_{\mc O}(\vv y)$. Hence,
	\begin{IEEEeqnarray*}{rCl}
		d_{ C'}(\vv y')&=& d_{\mc O}'(\vv y)\land  d_{\mc P}(\vv y)=d_{\mc O}'(\vv y) = d_{\mc P}(\vv y) >d_{C}(\vv y).
	\end{IEEEeqnarray*}
	By \eqref{eq:lemma5d}, we have $d_{\mc O}'(\vv y)=d_{\mc O}(\vv y)+1$.
        \end{itemize}
	The proof is completed.
\end{IEEEproof}

Now we move on to compare $\ls_{\mc C}(\epsilon)$ and $\ls_{\mc C'}(\epsilon)$ as defined in \eqref{eq:lambda}.
Define for $i=1,\ldots, 5$ and $d=0,1,\ldots,n$,
\begin{equation*}
  \alpha_{C}^i(d)= |\{\vv y\in \mc Y_i: d_{C}(\vv y) = d\}|.
\end{equation*}
As $\{\mc Y_1, \mc Y_2, \mc Y_3, \mc Y_4, \mc Y_5\}$ is a partition of $\{0,1\}^n$, we have
\begin{equation}\label{eq:10}
  \alpha_{C}(d) = \sum_{i=1}^5 \alpha_{C}^i(d).  
\end{equation}
As we will show in the following theory, the comparison of $\ls_{\mc C}(\epsilon)$ and $\ls_{\mc C'}(\epsilon)$ uses only $\alpha_C^3$ and $\alpha_C^5$. 
We have $\alpha_{C}^3(0) = 0$ and $\alpha_{C}^5(n)=0$:
\begin{itemize}
\item For $\vv y \in \mc Y_3$, $d_C(\vv y) = d_{\mc P} \land d_{\mc O} > d_{\mc P}' \geq 0$, and hence $\alpha_{C}^3(0) = 0$. 
\item For $\vv y \in \mc Y_5$, $d_C(\vv y ) = d_{\mc P} \land d_{\mc O} \leq d_{\mc O} < d_{\mc P} \leq n$, and hence $\alpha_{C}^5(n)=0$.
\end{itemize}

\begin{theorem}\label{the:1}
  Given BSC($\epsilon$), $0<\epsilon<1/2$, and two $(n,4)$ codes $C$ and $C'$ with only one column difference,
  \begin{equation*}
    \ls_{C'}(\epsilon)-\ls_{ C} (\epsilon) = \frac{(1-\epsilon)^n}{4}\left(1-\frac{\epsilon}{1-\epsilon}\right) \sum_{d=1}^n\left[\alpha_{C}^3(d)-\alpha_{C}^5(d-1)\right] \left(\frac{\epsilon}{1-\epsilon}\right)^{d-1}.
  \end{equation*}
  Moreover,
\begin{enumerate}
\item 	$\ls_{C'}(\epsilon) > \ls_{C}(\epsilon)$ if and only if
\begin{IEEEeqnarray*}{rCl}
	\sum_{d=1}^n\left[\alpha_{C}^3(d)-\alpha_{C}^5(d-1)\right] \left(\frac{\epsilon}{1-\epsilon}\right)^{d-1} >0;
\end{IEEEeqnarray*}
\item 	$\ls_{C'} (\epsilon)= \ls_{ C}(\epsilon)$ if and only if
\begin{IEEEeqnarray*}{rCl}
	\sum_{d=1}^n\left[\alpha_{C}^3(d)-\alpha_{C}^5(d-1)\right] \left(\frac{\epsilon}{1-\epsilon}\right)^{d-1} = 0.
\end{IEEEeqnarray*}
\end{enumerate}
\end{theorem}
\begin{IEEEproof}
  Due to code equivalence, we suppose that the difference of $C$ and
  $C'$ is in the first column.  As shown in the proof of
  Lemma~\ref{lemma:1},
  $\{\mc Y_1, \mc Y_2', \mc Y_3, \mc Y_4', \mc Y_5'\}$ is a partition
  of $\{0,1\}^n$. Define
\begin{IEEEeqnarray}{rCl}
	\alpha_{C'}^1(d) & = & |\{\vv y\in \mc Y_1: d_{C'}(\vv y) = d\}| = \alpha_C^1(d), \label{eq:1-1}\\
	\alpha_{C'}^2(d) & = & |\{\vv y\in \mc Y_2': d_{C'}(\vv y) = d\}| = \alpha_C^2(d), \label{eq:1-2}\\
	\alpha_{C'}^3(d) & = & |\{\vv y\in \mc Y_3: d_{ C'}(\vv y) = d\}| =
	\begin{cases}
		\alpha_{C}^3(d+1) & d<n,\\
		0 & d=n,
	\end{cases} \label{eq:1-3}\\
	\alpha_{C'}^4(d) & = & |\{\vv y\in \mc Y_4': d_{ C'}(\vv y) = d\}| = \alpha_d^4( C), \label{eq:1-4}\\
	\alpha_{C'}^5(d) & = & |\{\vv y\in \mc Y_5': d_{ C'}(\vv y) = d\}| =
	\begin{cases}
		\alpha_{C}^5(d-1) & d\geq 1,\\
		0 & d=0,
	\end{cases}\label{eq:1-5}
\end{IEEEeqnarray}
Where the second equality in each line follows from Lemma~\ref{lemma:1}. We have $\alpha_{C'}(d) = \sum_{i=1}^5 \alpha_{C'}^i(d)$. Together with \eqref{eq:10}, we write
\begin{IEEEeqnarray*}{rCl}
  \ls_{C'}(\epsilon)-\ls_{ C} (\epsilon)
  & = & \frac{1}{|C|} \sum_{d=0}^n( \alpha_{C'}(d)-\alpha_{C}(d) (1-\epsilon)^{n-d}\epsilon^{d}\\
  & = & \frac{1}{|C|} \sum_{d=0}^n\sum_{i=1}^5( \alpha_{C'}^i(d)-\alpha_{C}^i(d)) (1-\epsilon)^{n-d}\epsilon^{d} \\
  & = & \frac{1}{|C|} \sum_{d=0}^n\sum_{i=3,5}( \alpha_{C'}^i(d)-\alpha_{C}^i(d)) (1-\epsilon)^{n-d}\epsilon^{d}, \IEEEyesnumber \label{eq:1-6}
\end{IEEEeqnarray*}
where the last equality follows from \eqref{eq:1-1}, \eqref{eq:1-2} and \eqref{eq:1-4}.
By substituting \eqref{eq:1-3} and \eqref{eq:1-5} into \eqref{eq:1-6}, we get  %
\begin{IEEEeqnarray*}{rCl}
  \ls_{C'}(\epsilon)-\ls_{C} (\epsilon)
  & = &\frac{(1-\epsilon)^n}{|C|}\Bigg(
    \sum_{d=1}^{n-1}\left[\alpha_{C}^3(d+1)-\alpha_{C}^3(d)+\alpha_{C}^5(d-1)-\alpha_{C}^5(d)\right] \left(\frac{\epsilon}{1-\epsilon}\right)^{d}  \\
    & & + \alpha_{C}^3(1) - \alpha_{C}^5(0) + \left[-\alpha_C^3(n) + \alpha_C^5(n-1)\right] \left(\frac{\epsilon}{1-\epsilon}\right)^{n} \Bigg)\\
  & = &\frac{(1-\epsilon)^n}{|C|}
  \sum_{d=1}^n\left[\alpha_{C}^3(d)-\alpha_{C}^5(d-1)\right] \left(\frac{\epsilon}{1-\epsilon}\right)^{d-1}\left(1-\frac{\epsilon}{1-\epsilon}\right).
\end{IEEEeqnarray*}
The proof is completed by further checking when $\ls_{ C'}(\epsilon) > \ls_{ C}(\epsilon)$ and $\ls_{ C'}(\epsilon) = \ls_{ C}(\epsilon)$.
\end{IEEEproof}

The comparison in Theorem~\ref{the:1} depends on the crossover probability $\epsilon$ and hence is not universal. We further derive some universal code comparison results using Theorem~\ref{the:1}. 

\begin{corollary}\label{cor:1}
  For two $(n,4)$ codes $C$ and $C'$ with only one column different, 
  \begin{enumerate}
  \item 	$\ls_{ C'} = \ls_{ C}$ if for $d=1,\ldots, n$,
    \begin{equation*}
      \sum_{i=1}^d\alpha_{C}^3(i) = \sum_{i=0}^{d-1}  \alpha_{C}^5(i);
    \end{equation*} 
  \item $\ls_{ C'} > \ls_{C}$ if $      \sum_{i=1}^{d'}\alpha_{C}^3(i) \geq  \sum_{i=0}^{d'-1}  \alpha_{C}^5(i)$ for $d'=1,\dots,n$ and there exists  $d\in \{1,\ldots, n\}$ such that
    \begin{equation*}
      \sum_{i=1}^d\alpha_{C}^3(i) > \sum_{i=0}^{d-1}  \alpha_{C}^5(i);
    \end{equation*} 
  \item When $\mc Y_5=\emptyset$, $\ls_{ C'} \geq \ls_{C}$, where the equality holds if and only if $\mc Y_3=\emptyset$.
  \end{enumerate}
\end{corollary}
\begin{IEEEproof}
  Let $\epsilon_0=\frac{\epsilon}{1-\epsilon}$ and let $\Psi_d = \sum_{i=1}^d\left[\alpha_{C}^3(i) -  \alpha_{C}^5(i-1)\right]$ for $d=1,\ldots, n$ and $\Psi_0=0$.
  Write
  \begin{IEEEeqnarray*}{rCl}
    \sum_{d=1}^n[\alpha_{C}^3(d)-\alpha_{C}^5(d-1)] \left(\frac{\epsilon}{1-\epsilon}\right)^{d-1}
    & = & 
    \sum_{d=1}^n(\Psi_d-\Psi_{d-1}) \epsilon_0^{d-1} \\
    & = & \Psi_n\epsilon_0^{n-1}+ \sum_{d=1}^{n-1} \Psi_d (\epsilon_0^{d-1}-\epsilon_0^d).
  \end{IEEEeqnarray*}
  Note that for $0<\epsilon<\frac{1}{2}$, $\epsilon_0^d=\left(\frac{\epsilon}{1-\epsilon}\right)^{d}$ is a strictly  decreasing function of $d$. The first two claims can be proved as follows:
  \begin{enumerate}
  \item When $\Psi_d=0$, for $d=1,\ldots,n$, we have $\lambda_{C'}(\epsilon)= \lambda_{C}(\epsilon)$ for any $\epsilon$ by Theorem~\ref{the:1}--1).
  \item When $\Psi_d'\geq 0$ for all $d'=1,\dots,n$ and $\Psi_d>0$, for some $d\in\{1,\ldots,n\}$ we have $\lambda_{C'}(\epsilon)> \lambda_{C}(\epsilon)$ for any $\epsilon$ by Theorem~\ref{the:1}--2) due to $\epsilon_0^{n-1}>0$ and $\epsilon_0^{d-1}-\epsilon_0^d>0$ for $\forall d\in \{1,\ldots, n-1\}$.
  \end{enumerate}
  To prove the last claim, as $\mc Y_5=\emptyset$, $\alpha_d^5(C)=0$ and hence $\Psi_d\geq 0$ for $d=1,\ldots,n$. Thus, by the first two claims, $\ls_{ C'} \geq \ls_{C}$.
  Suppose $\mc Y_3=\emptyset$, which implies $\alpha_d^3(C)=0$. By claim 1), $\ls_{ C'} = \ls_{C}$.
  Suppose $\mc Y_3\neq \emptyset$. Since $|\mc Y_3| = \sum_{d=1}^n \alpha_{C}^3(d)>0$, $\alpha_{C}^3(d)>0$ for some $d\in \{1,\dots,n\}$. By Theorem~\ref{the:1}--1), $\ls_{ C'}(\epsilon) > \ls_{C}(\epsilon)$ for all $0<\epsilon<1/2$.
\end{IEEEproof}

The general approach described above can be used to study the optimality of codes with  $\bspan{0}$ columns (see the  \S\ref{sec:0column}), to find optimal linear $(n,4)$ codes (see \S\ref{sec:linearcode}), and to analyze Class-I codes (see \S\ref{sec:classI}). %
In \S\ref{sec:proof:even}, we will use Corollary~\ref{cor:1} to prove Theorem~\ref{thm:even}.

\subsection{Proof of Theorem~\ref{thm:0column}}\label{sec:0column}

Suppose an $(n,4)$ code $C$ has the first column $\bspan{0}$. Let code $C'$ be the $(n,4)$ code obtained by flipping bits in the rows of $\mc O\subset \{1,2,3,4\}$ of the first column of $C$. Theorem~\ref{thm:0column} states that $\lambda_{C'}\geq \lambda_C$. In the subsection, we give the proof for Theorem~\ref{thm:0column} based on the discussion in \S\ref{sec:1col}.

Let $\mc P=\{1,2,3,4\}\setminus \mc O$. 
Substituting $s=0$ to the discussion in \S\ref{sec:1col}, we have 
	\begin{IEEEeqnarray}{rCl}
		d_{\mc O}'(\vv y) & = & d_{\mc O}(\vv y) + \overline{y_1} -y_1, \label{eq:0column2}\\
		d_{\mc P}'(\vv y) & = & d_{\mc P}(\vv y) + \overline{y_1} -y_1.\label{eq:0column3}
	\end{IEEEeqnarray}
	As $d_{\mc O} - d_{\mc P} = d_{\mc O}' - d_{\mc P}'$, by the definition in \eqref{eq:y5}, $\mc Y_5 = \emptyset$. Hence by Corollary~\ref{cor:1}, we have $\lambda_{C'}\geq \lambda_C$.
We then check whether $\ls_{C'} =\ls_C$ through verifying whether $\mc Y_3 =\emptyset$.
  
  We first prove 1). WLOG, suppose $C=C_0$, i.e.,  $ |0|_{C}+ |5|_{C}+ |6|_{C}= n$ and $|5|_{C}, |6|_{C}$ are odd. For any $\cv y\in \{0,1\}^n$,
\begin{align}
  d_1(\cv y) &= w_0 + w_5 + w_6 , \\ d_2(\cv y) &= w_0 + \overline{w_5} + \overline{w_6} ,\label{eq:0columnc}\\
  d_3(\cv y) &= w_0 + w_5 +\overline{ w_6} , \\ d_4(\cv y) &= w_0 + \overline{w_5} + w_6.\label{eq:0columnd}
\end{align}
By the definition in \eqref{eq:y3},
\begin{equation}
  \label{eq:6}
  \mc Y_3 = \{\vv y\in \{0,1\}^n:y_1=1, d_{\mc O}(\vv y) = d_{\mc P}(\vv y)\}.
\end{equation}
For $\cv y\in \mc Y_3$, there must exist $i, j,h,k$ such that $\{i,j,h,k\} = \{1,2,3,4\}$ and $d_i = d_j\leq d_h \land d_k$ by the definition of $\mc Y_3$.
Consider the $6$ cases of $(i,j)$:
\begin{itemize}
\item When $(i,j) = (1,3)$, $(1,4)$, $(2,3)$ or $(2,4)$, as $|5|_{C}$ and $|6|_{C}$ are odd, $d_i+d_j$ is odd (checking~\eqref{eq:0columnc}--\eqref{eq:0columnd}), which is a contradiction to $d_i=d_j$. 
\item When $(i,j) = (1,2)$,  $d_1=d_2\leq d_3\land d_4$. By~\eqref{eq:0columnc}--\eqref{eq:0columnd}, we have $w_6 = \overline{w_6}$ and $w_5 = \overline{w_5}$, which implies $	w_6= \frac{|6|_{C}}{2} $ and $w_5 = \frac{|5|_{C}}{2}$. We get a contradiction to the assumption that $|5|_{C}$ and $|6|_{C}$ are odd.
\item When $(i,j) = (3,4)$, $d_3=d_4\leq d_1\land d_2$. The same contradiction as the previous case can be obtained. 
\end{itemize}
Thus we have $\mc Y_3=\emptyset$ for any $\mc O$ and $\mc P$.

Now we prove 2). For the code $C$, if the distances between a codeword with two other codewords $\cv c, \cv c'$ are both odd, the $\cv c$ and $ \cv c'$  must be of an even distance. Thus there exist two codewords in $C$ of an even distance. WLOG, suppose $\cv c_1=\cv 0$ and $w(\vv c_1\oplus \vv c_2)$ is even.
We prove 2) for $s'=5$, where $\mc O=\{1,3\}$ and $\mc P=\{2,4\}$.
Following \eqref{eq:6},
\begin{IEEEeqnarray*}{rCl}
  \mc Y_3 
  & = & \{\vv y\in \{0,1\}^n:y_1=1, d_1 \land d_3 = d_2\land d_4\}.
\end{IEEEeqnarray*}
$\mc Y_3$ is not empty if we find $\cv y \in \{0,1\}^n$ such that 
\begin{equation}\label{eq:0column1}
  d_3\land d_4 \geq d_1= 
  d_2 = \frac{w(\cv c_1\oplus \cv c_2)}{2}  +1 =  \frac{w_4+w_5+w_6+w_7}{2}  +1,%
\end{equation}
where the last equality follows from $\vv c_0=\cv 0$ and \eqref{eq:w1}. 
Let $\cv y\in \{0,1\}^n$ satisfy
\begin{equation*}
	y_j=
\begin{cases}
1 & j=1,\\
0  & j\geq 2\text{ and } j\notin \cup_{i=4}^7 \{i\}_C.
\end{cases}
\end{equation*} 
Substituting this form of $\vv y$ into \eqref{eq:d1}--\eqref{eq:d4}, we get
\begin{IEEEeqnarray*}{rCl}
  d_1(\vv y)&=&1+w_4+w_5+w_6+w_7, \\ 
  d_2(\vv y)&=&1+\overline{w_4}+ \overline{w_5}+\overline{w_6}+\overline{w_7}, \label{eq:0columna}\\
  d_3(\vv y)&=& 1+|2|_C+|3|_C+w_4+w_5+\overline{w_6}+\overline{w_7}, \\  d_4(\vv y)&=& 1+|1|_C+|3|_C+w_4+\overline{w_5}+w_6+\overline{w_7}. \label{eq:0columnb}
\end{IEEEeqnarray*}

We assign values of $w_4, w_5, w_6, w_7$ to satisfy \eqref{eq:0column1} by considering difference cases of $|4|_C$, $|5|_C$, $|6|_C$ and $|7|_C$. As 
$w(\vv c_1\oplus \vv c_2)= \sum_{i=4}^7 |i|_C$ is even, there are totally $8$ feasible parity combinations of $|4|_C$, $|5|_C$, $|6|_C$ and $|7|_C$.
Except for the case with $|4|_C=|7|_C=0$ and $|5|_C,|6|_C$ odd, we give the assignment of $w_4, w_5, w_6, w_7$ in Table~\ref{table:0column}.
For the case $|4|_C=|7|_C=0$ and $|5|_C,|6|_C$ odd, as $C$ is not equivalent to $C_0$, at least one of $|1|_C$, $|2|_C$ and $|3|_C$  is nonzero. In this case $w_4=w_7=0$. We give the assignment of $w_5$ and $w_6$ for $|1|_C>0$, $|2|_C>0$ and $|3|_C>0$ respectively in Table~\ref{table:0column2}. Thus, $\mc Y_3\neq \emptyset$ in all the cases above and we have $\ls_{C'}>\ls_C$.

Last, a general $(n,4)$ code $C$ can be converted to one with $\vv c_1=\cv 0$ and $w(\vv c_1\oplus \vv c_2)$ by interchanging rows and flipping all the bits in some columns. Hence, we know that there exists a code $C'$ such that $\ls_{C'}>\ls_C$, where $C'$ is obtained by changing a $\bspan{0}$ of $C$ by $\bspan{s'}$ with $s'\in \{2^r+2^t: 0\leq r\neq t\leq 3\}$. If $s'>7$, we can flip all the bits in the $\bspan{s'}$ column in $C'$.

\begin{table}[t]
	\centering
	\caption{Assignment of $w_4, w_5, w_6, w_7$ under different parities of $|4|_C$, $|5|_C$, $|6|_C$ and $|7|_C$}\label{table:0column}
	\begin{tabular}{ccccccccc}
		\hline
		\rowcolor[HTML]{EFEFEF} 
	case&	$|4|_C$ & $|5|_C$ & $|6|_C$ & $|7|_C$ & $w_4-\frac{|4|_C}{2}$ & $w_5-\frac{|5|_C}{2}$ & $w_6-\frac{|6|_C}{2}$ & $w_7-\frac{|7|_C}{2}$ \\ \hline
	$1$&	even   & even   & even   & even   & $0$                           & $0$                           & $0$                           & $0$                           \\ \hline
		\rowcolor[HTML]{EFEFEF} 
		$2$&odd    & odd    & even   & even   & $\frac{1}{2}$                 & $-\frac{1}{2}$                & $0$                           & $0$                           \\ \hline
		$3$&odd    & even   & odd    & even   & $\frac{1}{2}$                 & $0$                           & $-\frac{1}{2}$                & $0$                           \\ \hline
		\rowcolor[HTML]{EFEFEF} 
	$4$&	odd    & even   & even   & odd    & $\frac{1}{2}$                 & $0$                           & $0$                           & $-\frac{1}{2}$                \\ \hline
	$5a$&	even, $>0$   & odd    & odd    & even   & $1$                           & $-\frac{1}{2}$                 & $-\frac{1}{2}$                & $0$                            \\ \hline
	$5b$&	even   & odd    & odd    & even, $>0$   & $0$                           & $\frac{1}{2}$                 & $\frac{1}{2}$                & $-1$                          \\ \hline
		\rowcolor[HTML]{EFEFEF} 
	$6$&	even   & odd    & even   & odd    & $0$                           & $\frac{1}{2}$                & $0$                           & $-\frac{1}{2}$                 \\ \hline
	$7$&	even   & even   & odd    & odd    & $0$                           & $0$                           & $\frac{1}{2}$                & $-\frac{1}{2}$                 \\ \hline
		\rowcolor[HTML]{EFEFEF} 
	$8$&	odd    & odd    & odd    & odd    & $\frac{1}{2}$                 & $\frac{1}{2}$                           & $-\frac{1}{2}$                           & $-\frac{1}{2}$                 \\ \hline
	\end{tabular}
\end{table}
\begin{table}[]
  \centering
  \caption{Assignment of $w_5, w_6$ when $|4|_C=|7|_C=0$, and $|5|_C,|6|_C$ are odd. }\label{table:0column2}
  \begin{tabular}{ccc}
    \hline
    \rowcolor[HTML]{EFEFEF} 
    case    & $w_5-\frac{|5|_C}{2}$ & $w_6-\frac{|6|_C}{2}$  \\ \hline
    $|1|_C>0$ & $\frac{1}{2}$                & $-\frac{1}{2}$                     \\ \hline
    \rowcolor[HTML]{EFEFEF} 
    $|2|_C>0$ &  $-\frac{1}{2}$                & $\frac{1}{2}$                           \\ \hline
    $|3|_C>0$ & $\frac{1}{2}$                & $-\frac{1}{2}$                                          \\ \hline
  \end{tabular}
\end{table}

\subsection{Proof of  Theorem~\ref{thm:even}}
\label{sec:proof:even}

Consider an $(n,4)$ code $C$ with a type $\bspan{1}$ column and
$w(\vv c_3\oplus \vv c_4)$ even. Let $C'$ be the code obtained by
replacing a $\bspan{1}$ column of $C$ by
$\bspan{3}$. Theorem~\ref{thm:even} states that
$\lambda_{C'} \geq \lambda_C$ together with a necessary and sufficient
condition such that the equality holds.  Now we give the proof of
Theorem~\ref{thm:even}. WLOG, we assume that the first column of $C$
is of type $\bspan{1}$ and is replaced by $\bspan{3}$ in $C'$. In the
proof, we write $|i|_C =|i|$. 

  Substituting $s=1$ and $s'=3$ to the discussion in \S\ref{sec:1col} (ref. Example~\ref{ex:s1s3}), we have $\mc O=\{1,2,4\}$ and $\mc P=\{3\}$ for $C$ and $C'$, and hence 
  \begin{equation*}
    \mc Y_5 = \{d_{3}'  =  d_{\{1,2,4\}} < d_{3} = d_{\{1,2,4\}}'\}.
  \end{equation*}
Assume $\mc Y_5$ is nonempty and fix $\vv y \in \mc Y_5$.
  As $d_{3}'(\vv y) = d_3(\vv y) - \vv y_1 + \overline{\vv y_1}$, we have $\vv y_1=1$.
  Further, due to
\begin{IEEEeqnarray*}{rCl}
  d_{\{1,2,4\}}(\vv y) & = & d_1 \land d_2 \land d_4, \\
  d_{\{1,2,4\}}'(\vv y) & = & (d_1-1) \land (d_2-1) \land (d_4+1), 
\end{IEEEeqnarray*}
we have $d_{\{1,2,4\}} = d_4$ and hence $d_3 = d_4 + 1$.
By \eqref{eq:d},
\begin{IEEEeqnarray*}{rCl}
  d_3(\vv y) + d_4(\vv y)
  & = &
   \sum_{i}\hat{w}_{i3} + 
    \sum_{i}\hat{w}_{i4}\\
  & = & \sum_{i:(\bspan{i})_3\neq (\bspan{i})_4} |i| + 
   2 \sum_{i:(\bspan{i})_3=(\bspan{i})_4} \hat{w}_{i3} \\
  & = & w(\vv c_3\oplus \vv c_4) + 2 \sum_{i:(\bspan{i})_3=(\bspan{i})_4} \hat{w}_{i3}.
\end{IEEEeqnarray*}
As $w(\vv c_3\oplus \vv c_4)$ is even, we see that $d_3+d_4$ is even, which is a contradiction to $d_3 = d_4+1$. Therefore, $\mc Y_5 = \emptyset$ and hence by Corollary~\ref{cor:1}, $\lambda_{C'}\geq \lambda_{C}$.

Now we study the condition of  $\lambda_{C'}= \lambda_C$.
To simplify the discussion, WLOG, we further assume that
the code $C$ has only the columns of types $\bspan{0},\bspan{1},\ldots,\bspan{7}$.
Also by Corollary~\ref{cor:1},  a sufficient and necessary condition for
$\lambda_{C'}= \lambda_C$ is 
$\mc Y_3=\emptyset$. Similar to $\mc Y_5$, $\mc Y_3$ can be rewritten as
\begin{IEEEeqnarray*}{rCl}
  \mc Y_3&=&\{d_{3}' = d_{\{1,2,4\}}' < d_{3} = d_{\{1,2,4\}}\}\\
  &=&\{y_1=1,d_3=d_1\land d_2\leq d_4   \}\\
  &=& \mc Y_3^1 \cup \mc Y_3^2,
\end{IEEEeqnarray*}	
where $ \mc Y_3^1= \{y_1=1,d_3=d_1\leq d_2\land d_4\} $ and $ \mc Y_3^2= \{y_1=1,d_3=d_2\leq d_1\land d_4\}$.
We verify the condition such that  $\mc Y_3=\emptyset$ for different parity of $w(\cv c_1\oplus \cv c_3)$ and 
$w(\cv c_2\oplus \cv c_3)$.
When $ w(\vv c_1\oplus \vv c_3) $ and $ w(\vv c_2\oplus \vv c_3)$ are both odd.   
For any $ \vv y \in \{0,1\}^n$, it holds that
\begin{align*}
  &	d_3(\vv y)  = w(\cv c_3\oplus\cv y)\neq 
    \begin{cases}
      w(\cv c_1\oplus \cv y) =d_1(\cv y),\\
      w(\cv c_2\oplus \cv y) =d_2(\cv y),
    \end{cases}
\end{align*}
which implies that $\mc Y_3^1=\mc Y_3^2=\emptyset$ and thus we have  $ \mc Y_3=\emptyset $. We henceforth discuss the cases with either $ w(\vv c_1\oplus \vv c_3) $ or $ w(\vv c_2\oplus \vv c_3)$ even.

When $ w(\vv c_1\oplus \vv c_3)$ is even, we will show that  $\mc Y_3=\emptyset$ if and only if $ |1|=1$, $|2|=|4|=|5|=|7|=0 $ and $|3|,|6|$ are odd.
By \eqref{eq:d1}--\eqref{eq:d4}, $\vv y \in \mc Y_3^1$, if it has $y_1=1$ and
\begin{IEEEeqnarray}{rCl}
	w_2+w_3 +w_6+w_7 & = &  \frac{|2|+|3|+|6|+|7|}{2}, \label{eq:y31} \\
	w_4+w_5+w_6+w_7 & \leq & \frac{|4|+|5|+|6|+|7|}{2},\label{eq:y32} \\
	w_1+w_3+w_5+w_7 & \leq & \frac{|1|+|3|+|5|+|7|}{2}.\label{eq:y33}
\end{IEEEeqnarray}
When $ w(\vv c_1\oplus \vv c_3) = |2|+|3|+|6|+|7| $ (ref. \eqref{eq:w2}) is even,
there are totally eight possible cases corresponding to the parities of $|2|$, $|3|$, $|6|$ and $|7| $. 
For the following seven cases, we have values of $w_1\geq 1,w_2,\ldots,w_6,w_7$ satisfying \eqref{eq:y31}--\eqref{eq:y33} and hence $ \mc Y_3^1\neq \emptyset $:
\begin{enumerate}
\item When $|2|$, $|3|$, $|6|$ and $|7| $ are all even, set $\cv y \in \{0,1\}^n$ with  $ y_1=1$, $w_1=1$, $w_2=\frac{|2|}{2}$, $w_3=\frac{|3|}{2}$, $w_4=0$, $w_5=0$, $w_6=\frac{|6|}{2}$, $w_7=\frac{|7|}{2}$.
Since $w(\vv c_1\oplus \vv c_3) = |1|+|2|+|5|+|6| $ is even and $|2|+|6|$ is even, $ |1|+|5| $ is even and hence $ |1|+|5|\geq 2 $. Thus \eqref{eq:y31}--\eqref{eq:y33} are satisfied, $ \vv y\in \mc Y_3^1 $   and then $\mc Y_3\neq \emptyset$.

\item When $|2|$, $|3|$, $|6|$ and $|7| $ are all odd, $w_1=1$, $w_2=\frac{|2|+1}{2}$, $w_3=\frac{|3|-1}{2}$, $w_4=0$, $w_5=0$, $w_6=\frac{|6|+1}{2}$, $w_7=\frac{|7|-1}{2}$.
\item When $|2|$, $|3|$ are odd,  and $|6|$, $|7| $ are even, $w_1=1$, $w_2=\frac{|2|+1}{2} $, $w_3=\frac{|3|-1}{2}$, $w_4=0$, $w_5=0$, $w_6=\frac{|6|}{2}$, $w_7=\frac{|7|}{2}$.

\item When $ |2|$, $|6|$ are odd, and $|3|$, $|7| $ are even, set $\cv y\in \{0,1\}^n$ with  $ y_1=1$, $w_1=1$, $w_2=\frac{|2|+1}{2} $, $w_3=\frac{|3|}{2}$, $w_4=0$, $w_5=0$, $w_6=\frac{|6|-1}{2}$, $w_7=\frac{|7|}{2}$.
 We can verify similar as case 1) that $ \vv y\in \mc Y_3^1 $.
  
\item When $ |2|$, $|7|$ are odd, and $|3|$, $|6| $ are even,  $w_1=1$, $w_2=\frac{|2|+1}{2} $, $w_3=\frac{|3|}{2}$, $w_4=0$, $w_5=0$, $w_6=\frac{|6|}{2}$, $w_7=\frac{|7|-1}{2}$. %
	
\item When $|3|$, $|7| $ are odd, and $|2|$, $|6|$ are even, set $\cv y\in \{0,1\}^n$ with  $y_1=1$, $w_1=1$, $w_2=\frac{|2|}{2} $, $w_3=\frac{|3|+1}{2}$, $w_4=0$, $w_5=0$, $w_6=\frac{|6|}{2}$, $w_7=\frac{|7|-1}{2}$, we can verify similar as case 1) that $ \vv y\in \mc Y_3^1 $.
	
\item When  $|6|$, $|7| $ are odd, and $|2|$, $|3|$ are even,  $w_1=1$, $w_2=\frac{|2|}{2} $, $w_3=\frac{|3|}{2}$, $w_4=0$, $w_5=0$, $w_6=\frac{|6|+1}{2}$, $w_7=\frac{|7|-1}{2}$. %
\end{enumerate}

For the remaining case that $ |3|$, $|6| $ are odd, and $ |2|$, $|7| $ are even, we discuss it in five sub-cases. In the first $4$ sub-cases, we have values of $ w_1\geq 1,w_2,\ldots,w_6,w_7$ satisfying \eqref{eq:y31}--\eqref{eq:y33} and hence $ \mc Y_3^1\neq \emptyset $:
\begin{enumerate}
\item[8-1)] When $ |2|> 0 $, $w_1=1$, $w_2=\frac{|2|}{2}+1 $, $w_3=\frac{|3|-1}{2}$, $w_4=0$, $w_5=0$, $w_6=\frac{|6|-1}{2}$, $w_7=\frac{|7|}{2}$.

\item[8-2)] When $ |7|> 0 $, $w_1=1$, $w_2=\frac{|2|}{2} $, $w_3=\frac{|3|+1}{2}$, $w_4=0$, $w_5=0$, $w_6=\frac{|6|+1}{2}$, $w_7=\frac{|7|}{2}-1$.%
  
\item[8-3)] When $ |4|+|5|> 0 $, $w_1=1$, $w_2=\frac{|2|}{2} $, $w_3=\frac{|3|-1}{2}$, $w_4=0$, $w_5=0$, $w_6=\frac{|6|+1}{2}$, $w_7=\frac{|7|}{2}$.

\item[8-4)] When $ |1|\geq 3 $, $w_1=1$, $w_2=\frac{|2|}{2} $, $w_3=\frac{|3|+1}{2}$, $w_4=0$, $w_5=0$, $w_6=\frac{|6|-1}{2}$, $w_7=\frac{|7|}{2}$.
\end{enumerate}
In the last sub-case, $ |1|\leq 2$, $|2|=|4|=|5|=|7|=0 $. First, in this case, we must have $|1|=1$ since $w(\cv c_3\oplus \cv c_4) = |1|+|6|$ is even.
For any $ \vv y\in\{0,1\}^n $ with $ y_1=1 $, 
\begin{equation}
	d_1(\vv y)+d_3(\vv y)=|3|+|6|+2w_0+2, \quad d_2(\vv y)+d_4(\vv y)=|3|+|6|+2w_0+1 .\label{eq:st1}
\end{equation}
If $ \vv y\in \mc Y_3^1$, $ d_1(\vv y)=d_3(\vv y)=\frac{|3|+|6|+2w_0+2}{2}\leq d_2(\vv y)\land d_4(\vv y) $. Then $ d_2(\vv y)+d_4(\vv y)\geq |3|+|6|+2w_0+2 $, which is a contradiction with \eqref{eq:st1}. Thus $ \vv y \not\in \mc Y_3^1 $, which implies $ \mc Y_3^1=\emptyset$.
If $ \vv y\in \mc Y_3^2 $, $ d_3(\vv y)=d_2(\vv y) $ and thus $\mc Y_3^2\neq \emptyset $ only if $ w(\vv c_2\oplus \vv c_3) $ is even, which never holds since  $ w(\vv c_2\oplus \vv c_3) =|3|$ is odd. Therefore  $ \vv y \not\in \mc Y_3^2 $ and then $ \mc Y_3^2 =\emptyset$. As a result, $ \mc Y_3=\emptyset $ when $ C $ satisfies $ |1|=1$, $|2|=|4|=|5|=|7|=0 $, $ |3|$ and $|6| $ are odd.

The case when $ w(\vv c_2\oplus \vv c_3)$ is even can be shown by the equivalence relation. Denote by $\tilde{C}$ the code obtained by modifying $C$ as follows:
\begin{itemize}
\item First, interchange $\vv c_1$ and $\vv c_2$;
\item Then, flip all the bits of columns of type $\bspan{i}$, $i>7$.
\end{itemize}
$\tilde C$ is equivalent to $C$. Let 
$\tilde{\cv c}_1, \tilde{\cv c}_2,\tilde{\cv c}_3,\tilde{\cv c}_4 $ be the four codewords of $\tilde C$. We can check that $|i|_{\tilde{C}} = |i|$ for $i=0,1,2,3$, $|4|_{\tilde{C}}=|7|$, $|7|_{\tilde{C}}=|4|$,  $|5|_{\tilde{C}}=|6|$ and $|6|_{\tilde{C}}=|5|$. Then  $w(\tilde{\cv c}_1 \oplus \tilde{\cv c}_3) = w(\cv c_2\oplus \cv c_3) $ and $w(\tilde{\cv c}_3 \oplus \tilde{\cv c}_4) = w(\cv c_3\oplus \cv c_4) $ are both even.
Let $\tilde{C}'$ be the code obtained by replacing the first column of type $\bspan{1}$ in $\tilde{C}$ by $\bspan{3}$. 
Thus we have $\ls_{\tilde{C}} = \ls_{\tilde{C}'}$ if and  only if $ |1|_{\tilde{C}}=1$, $|2|_{\tilde{C}}=|4|_{\tilde{C}}=|5|_{\tilde{C}}=|7|_{\tilde{C}}=0 $ and $|3|_{\tilde{C}}$, $|6|_{\tilde{C}}$ are odd. As $\tilde C'$ is equivalent to $C'$, we have
$\ls_{C} = \ls_{C'}$ if and only if $ |1|=1$, $|2|=|4|=|6|=|7|=0 $ and $|3|$, $|5|$ are odd. 

\section{Change of Two Columns: Proof of Theorem~\ref{thm:odd}}
\label{sec:fliptwobits}

Consider an $(n,4)$ code  $C$ with codewords $\cv c_1,\cv c_2,\cv c_3,\cv c_4$ and the first two columns being $\bspan{1}$ and $\bspan{7}$. Let $C'$ be the code obtained by flipping the first two bits of $\vv c_3$ in $C$, so that the first two columns of $C'$ are $\bspan{3}$ and $\bspan{5}$. Theorem~\ref{thm:odd} says that $\ls_{ C'} \geq \ls_{ C}$ together with a necessary and sufficient
condition such that the equality holds. Now we give the proof of
Theorem~\ref{thm:odd} for $C$ and $C'$. Other cases of Theorem~\ref{thm:odd} can be obtained due to code equivalence.

\subsection{Proof of $\ls_{ C'} \geq \ls_{ C}$}\label{sssec:geq}

Following the notations in \S\ref{sec:compgen}, for $C$ and $C'$, we have $\mc O=\{1,2,4\}$, $\mc P=\{3\}$, and by \eqref{eq:dprim}
\begin{IEEEeqnarray*}{rCl}
  d_1'(\vv y) & = & d_1(\vv y) +  (\overline{y_1}-y_1)+  (\overline{y_2}-y_2),\\
  d_2'(\vv y) & = & d_2(\vv y) +  (\overline{y_1}-y_1) - (\overline{y_2}-y_2),\\
  d_3'(\vv y) & = & d_3(\vv y) +  (\overline{y_1}-y_1) - (\overline{y_2}-y_2),\\
  d_4'(\vv y) & = & d_4(\vv y) - (\overline{y_1}-y_1) - (\overline{y_2}-y_2). 
\end{IEEEeqnarray*}
When $y_1=y_2 $, we have
\begin{equation}
d_{\mc P}'(\vv y)=d_{\mc P}(\vv y).  \label{eq:thm2b}
\end{equation}
When $y_1\neq y_2$, we have $d_1'(\vv y)=d_1(\vv y)$,
$d_4'(\vv y)=d_4(\vv y)$, and
\begin{equation}
d_2'(\vv y)-d_2(\vv y)= d_3'(\vv y)-d_3(\vv y)=\pm 2, \label{eq:thm2a}
\end{equation}
and hence
\begin{IEEEeqnarray}{rCl}
	\IEEEeqnarraymulticol{3}{l}{(d_{\mc O}'(\vv y)- d_{\mc O}(\vv y)) (d_{\mc P}'(\vv y)- d_{\mc P}(\vv y))}\nonumber \\
	&=& (d_{\{1,2,4\}}'(\vv y)-d_{\{1,2,4\}}(\vv y))(d_3'(\vv y)-d_3(\vv y))
	\geq 0.   \label{eq:ineq1}
\end{IEEEeqnarray}

The crucial part of comparing $C$ and $C'$ is
the following $5$ subsets of $\{0,1\}^n$:
\begin{IEEEeqnarray}{rCl}
	\mc Z_1 &=& \{y_1=y_2  \}, \label{eq:y21} \\
	\mc Z_2&=& \{y_1\neq y_2,d_{\mc O} \leq d_{\mc P} \land d_{\mc P}'  \},\\
	\mc Z_3 &=& \{y_1\neq y_2,d_{\mc O} > d_{\mc P} \land d_{\mc P}',d_{\mc P}\leq d_{\mc O}\land d_{\mc O}'   \},\\
	\mc Z_4 &=& \{y_1\neq y_2,d_{\mc P}'<d_{\mc O}\land d_{\mc O}' <d_{\mc P}  \},\\
	\mc Z_5&=& \{y_1\neq y_2,d_{\mc O}' \leq  d_{\mc P} \land d_{\mc P}' <d_{\mc O},d_{\mc O}'<d_{\mc P}  \}.\label{eq:y25}
\end{IEEEeqnarray}
Recall the function $F_2$ (defined in \S\ref{sec:compgen}) flips the first two bits of a binary vector. Define a mapping $g_2: \{0,1\}^n \rightarrow \{0,1\}^n$ as
\begin{equation*}
  g_2(\vv y)  =
  \begin{cases}
    \vv y  & \vv y \in \mc Z_1 \cup \mc Z_2 \cup \mc Z_5, \\
    F_2(\vv y) & \text{otherwise}.
  \end{cases}
\end{equation*}
The next lemma shows that $\{\mc Z_1, \mc Z_2, \mc Z_3, \mc Z_4, \mc Z_5\}$ and $g_1$ satisfy the properties described in \S\ref{sec:compgen} for $C$ and $C'$.

\begin{lemma}\label{lemma:2}
  $\{\mc Z_1, \mc Z_2, \mc Z_3, \mc Z_4, \mc Z_5\}$ defined in \eqref{eq:y21}--\eqref{eq:y25} forms a partition of $\{0,1\}^n$, and $g_2$ is a one-to-one and onto mapping. Moreover, for the $(n,4)$ codes $C$ and $C'$ formulated above,  
  \begin{enumerate}
  \item For $\vv y \in \mc Z_1$, $d_C(\vv y) = d_{C'}(\vv y)$;
  \item For $\vv y \in \mc Z_2$, $d_C(\vv y) = d_{C'}(\vv y)=d_{\mc O}$;
  \item For $\vv y \in \mc Z_3$, $d_C(\vv y) = d_{C'}( F_2(\vv y))=d_{\mc P}$;
  \item 
    For $\vv y \in \mc Z_4$, $d_C(\vv y) = d_{\mc O}\land d_{\mc P}\geq d_{C'}(F_2(\vv y))=d_{\mc O}'$; 
  \item For $\vv y \in \mc Z_5$, $d_C(\vv y) = d_{\mc O}\land d_{\mc P}> d_{C'}(\vv y)=d_{\mc P}'$.
  \end{enumerate}
\end{lemma}
\begin{IEEEproof}
For $i=3,4$, let 
\begin{equation*}
\mc Z_i'=\{F_2(\vv y): \vv y \in \mc Z_i \}.
\end{equation*}

We justify that $\mc Z_1,\mc Z_2,\mc Z_3,\mc Z_4,\mc Z_5$ form a partition of $\{0,1\}^n$ and $\mc Z_1,\mc Z_2,\mc Z_3',\mc Z_4',\mc Z_5$ form a partition of $\{0,1\}^n$:
First, we show that 
\begin{equation}\label{eq:y45}
\mc Z_4 \cup \mc Z_5 =\{y_1\neq y_2,d_{\mc O} > d_{\mc P} \land d_{\mc P}',d_{\mc P}> d_{\mc O}\land d_{\mc O}'   \},
\end{equation}
and then we obtain $\bigcup_{i=1}^5\mc Z_i=\{0,1\}^n$. Moreover, $\mc Z_1, \ldots, \mc Z_{5}$ are all disjoint by checking the definition. Thus $\mc Z_1, \ldots, \mc Z_{5}$ form a partition of $\{0,1\}^n$.

To show \eqref{eq:y45}, 
since $\mc Z_4\subseteq \{ d_{\mc O} > d_{\mc P} \land d_{\mc P}'\}$, we have
\begin{IEEEeqnarray}{rCl}
	\mc Z_4
	&=& \{y_1\neq y_2,d_{\mc P}'<d_{\mc O}\land d_{\mc O}' <d_{\mc P}  \}\cap  \{ d_{\mc O} > d_{\mc P} \land d_{\mc P}'\} \nonumber \\
	&=&	\{y_1\neq y_2,d_{\mc O} > d_{\mc P} \land d_{\mc P}',d_{\mc P}> d_{\mc O}\land d_{\mc O}'   \}\cap \nonumber\\
	&&\{d_{\mc O}\land d_{\mc O}'>d_{\mc P}' \}. \label{eq:ineq2}
\end{IEEEeqnarray}
Denote 
\begin{IEEEeqnarray*}{rCl}
	\mathcal A_1&=&\{y_1\neq y_2,	d_{\mc O} > d_{\mc P} \land d_{\mc P}',d_{\mc P}> d_{\mc O}\land d_{\mc O}'\} \cap\\
	&& \{d_{\mc O}\land d_{\mc O}'\leq d_{\mc P}'\}. \IEEEyesnumber  \label{eq:ineq3}
\end{IEEEeqnarray*}
For $\vv y\in \mathcal A_1$, we have $d_{\mc O}(\vv y)>d_{\mc O}'(\vv y)$ which implies $d_{\mc P}(\vv y)>d_{\mc P}'(\vv y)$ by \eqref{eq:thm2a} and \eqref{eq:ineq1}, and hence
\begin{equation*}
d_{\mc O}'(\vv y)\leq d_{\mc P}(\vv y)\land d_{\mc P}'(\vv y)<d_{\mc O}(\vv y),\quad d_{O}'(\vv y)<d_{\mc P}(\vv y).
\end{equation*}
Thus we have $\vv y\in \mc Z_5$ and then $\mathcal A_1\subseteq \mc Z_5$. 
For $\vv y \in \mc Z_5$, we have  $d_{\mc O}(\vv y)>d_{\mc O}'(\vv y)$ by the definition above, which implies $d_{\mc P}(\vv y)>d_{\mc P}'(\vv y)$ by \eqref{eq:thm2a} and \eqref{eq:ineq1}. Then we obtain 
\begin{equation*}
d_{\mc P}(\vv y)>d_{\mc O}(\vv y)\land d_{\mc O}'(\vv y)=d_{\mc O}'(\vv y)\leq d_{\mc P}'(\vv y)<d_{\mc O}(\vv y).
\end{equation*}
Thus $\vv y\in \mathcal A_1$ and then $\mc Z_5\subseteq \mathcal A_1$.
Therefore, $\mc Z_5= \mathcal A_1$.
From \eqref{eq:ineq2} and \eqref{eq:ineq3}, we obtain \eqref{eq:y45}.

We further show that
\begin{equation}
\label{eq:y34p}
\mc Z_3'\cup \mc Z_4' \subseteq\mc Z_3 \cup \mc Z_4.
\end{equation}
Since $F_2$ is an one-to-one mapping, we get $\mc Z_3'\cup \mc Z_4' =\mc Z_3 \cup \mc Z_4$. Therefore, $\mc Z_1, \mc Z_2, \mc Z_3', \mc Z_4', \mc Z_5$ form a partition of $\{0,1\}^n$.

To show \eqref{eq:y34p}, we see
\begin{equation}
\mc Z_4'=\{y_1 \neq y_2, d_{\mc P}<d_{\mc O}\land d_{\mc O}' <d_{\mc P}' \}%
\subseteq \mc Z_3, \label{eq:thm2a4}
\end{equation}
and
\begin{IEEEeqnarray*}{rCl}
	\mc Z_3'\setminus  \mc Z_4
	&=& \{y_1\neq y_2,d_{\mc O}' > d_{\mc P} \land d_{\mc P}',d_{\mc P}'\leq d_{\mc O}\land d_{\mc O}' \} \cap  \\
	&& (\{d_{\mc O}\land d_{\mc O}' \geq d_{\mc P}   \}  \cup   \{ d_{\mc P}'\geq d_{\mc O}\land d_{\mc O}'  \})\\
	&=& \{ y_1\neq y_2,d_{\mc O}' > d_{\mc P} \land d_{\mc P}',d_{\mc P}'\leq d_{\mc O}\land d_{\mc O}' ,  \\
	& & d_{\mc O}\land d_{\mc O}' \geq d_{\mc P}   \}  \cup 
	\{y_1\neq y_2,d_{\mc O}' > d_{\mc P} \land d_{\mc P}',\\
	& & d_{\mc P}'\leq d_{\mc O}\land d_{\mc O}', d_{\mc P}'\geq d_{\mc O}\land d_{\mc O}'  \}  \\
	&= &  \{y_1\neq y_2 ,d_{\mc P}\land d_{\mc P}'<\max (d_{\mc P},d_{\mc P}')\leq d_{\mc O}\land d_{\mc O}'  \}\cup\\
	&& \{ y_1\neq y_2 , d_{\mc O}\land d_{\mc O}'=d_{\mc P}', d_{\mc P}\land d_{\mc P}'<d_{\mc O}'  \}  \IEEEyesnumber  \label{eq:thm2a1} %
\end{IEEEeqnarray*}
where in the last equality $d_{\mc P}\land d_{\mc P}'<\max (d_{\mc P},d_{\mc P}')$ follows from \eqref{eq:thm2a}.
By \eqref{eq:ineq1}, when $y_1\neq y_2$, if $d_{\mc O}<d_{\mc O}'$, then $d_{\mc P}<d_{\mc P}'$;  and if $d_{\mc P}>d_{\mc P}'$, then $d_{\mc O}\geq d_{\mc O}'$. Hence, we can verify that both terms to union in \eqref{eq:thm2a1} are subsets of $\mc Z_3$. Therefore, $\mc Z_3'\setminus  \mc Z_4 \subset \mc Z_3$, which together with \eqref{eq:thm2a4}, proves \eqref{eq:y34p}.
The above claims are justified as follows:
\begin{enumerate}
	\item For $\vv y \in \mc Z_1$, as $y_1=y_2$, we have $d_{\mc P}'=d_{\mc P}$ by \eqref{eq:thm2b}, and hence
	\begin{IEEEeqnarray*}{rCl}
		d_{C}(\vv y)&=& d_{\mc O}(\vv y)\land d_{\mc P}(\vv y)=d_{\mc O}(\vv y)\land d_{\mc P}'(\vv y)=d_{C'}(\vv y).
	\end{IEEEeqnarray*}
	
	\item For $\vv y \in \mc Z_2$, by the definition of $\mc Z_2$, we have $d_{\mc O}\leq d_{\mc P}\land d_{\mc P}'$, and hence $d_C(\vv y) = d_{C'}(\vv y)=d_{\mc O}$.
	
	\item  For $\vv y \in \mc Z_3$, we have $d_{\mc P}\leq d_{\mc O}\land d_{\mc O}'$ by the definition of $\mc Z_3$. We then have
	\begin{IEEEeqnarray*}{rCl}
		d_{C}(\vv y)&=& d_{\mc O}(\vv y)\land d_{\mc P}(\vv y) = d_{\mc P}(\vv y),\\
		d_{C'}(F_2(\vv y))&=&  d_{\mc O}'(\vv y)\land d_{\mc P}(\vv y) = d_{\mc P}(\vv y).
	\end{IEEEeqnarray*}
	
	\item For $\vv y \in \mc Z_4$, we have $y_1\neq y_2,d_{\mc P}'<d_{\mc O}\land d_{\mc O}'<d_{\mc P}$ by the definition of $\mc Z_4$. By \eqref{eq:ineq1}, 
	$d_{\mc O}\land d_{\mc O}'=d_{\mc O}'$,
	which implies
	$d_{C'}(F_2(\vv y))=d_{\mc O}'(\vv y)\land d_{\mc P}(\vv y)=d_{\mc O}'(\vv y)$
	and hence
	\begin{IEEEeqnarray*}{rCl}
		d_C(\vv y)=d_{\mc O}(\vv y)\land d_{\mc P}(\vv y)\geq d_{\mc O}'(\vv y)=	d_{C'}(F_2(\vv y)).
	\end{IEEEeqnarray*}
	
	\item For $\vv y\in \mc Z_5$, we have $y_1\neq y_2,d_{\mc O}' \leq  d_{\mc P} \land d_{\mc P}' <d_{\mc O} $. By \eqref{eq:thm2a} and \eqref{eq:ineq1}, $d_{\mc P}'<d_{\mc P}$. Then we have
	\begin{IEEEeqnarray*}{rCl}
		d_C(\vv y)=d_{\mc O}\land d_{\mc P}\geq d_{\mc O}\land d_{\mc P}'=d_{\mc P}'=d_{C'}(\vv y),
	\end{IEEEeqnarray*}
	which implies
	\begin{equation*}
	d_{C'}(\vv y)< d_C(\cv y).
	\end{equation*}
\end{enumerate}

\end{IEEEproof}

Define two subsets of $\mc Z_4$:
\begin{IEEEeqnarray}{rCl}
  \mc Z_4^1 & = & \{ y_1\neq y_2, d_{\mc P}'< d_{\mc O}\land d_{\mc O}'<d_{\mc P}, d_{\mc O} \land d_{\mc P}>d_{\mc O}' \}, \label{eq:y41}\\
  \mc Z_4^2 & = & \{ y_1\neq y_2, d_{\mc P}'< d_{\mc O}\land d_{\mc O}'<d_{\mc P}, d_{\mc O} \land d_{\mc P}=d_{\mc O}' \}.\nonumber
\end{IEEEeqnarray}
We see that $\mc Z_4^1 \cup \mc Z_4^2 = \mc Z_4$ and $\mc Z_4^1 \cap \mc Z_4^2 = \emptyset$. By Lemma~\ref{lemma:2}, $d_C(\vv y) > d_C'(g_2(\vv y))$ for $\vv y \in \mc Z_4^1 \cup \mc Z_5$ and $d_C(\vv y) = d_C'(g_2(\vv y))$ for $\vv y \notin \mc Z_4^1 \cup \mc Z_5$. 
Hence by Lemma \ref{thm:com}, $\lambda_{C'} \geq \lambda_C$, where the equality holds if and only if $\mc Z_4^1 \cup \mc Z_5 = \emptyset$.

\subsection{Proof of $\ls_{ C'} = \ls_{ C}$}

Now we verify the necessary and sufficient conditions for $\mc Z_4^1\cup \mc Z_5=\emptyset$ when $C$'s columns are all from $\bspan{0},\bspan{1},\dots, \bspan{7}$.
Recall that $w(\vv c_2\oplus \vv c_3) = |2| + |3| + |4| + |5|$ (ref. \ref{eq:w4}). We discuss whether $\mc Z_4^1\cup \mc Z_5$ is empty or not according to the different parity of $|2|$, $|3|$, $|4|$ and $|5|$. See the cases listed in Table~\ref{tab:cases2col}. In the following Lemma~\ref{lemma:y4}, it is shown that for cases 9)--16) (i.e., when $ w(\vv c_2\oplus \vv c_3) $ is even), $\mc Z_4^1 = \emptyset$, and for cases 1), 4), 5) and 6), $\mc Z_4^1$ is nonempty. In the following Lemma~\ref{lemma:z5}, it is further verified that $\mc Z_5$ is empty for cases 7) -- 15). Thus we have $\mc Z_4^1\cup \mc Z_5\neq \emptyset$ for cases 1) and 4)-15).
 The necessary and sufficient conditions for $\mc Z_4^1\cup \mc Z_5=\emptyset$ is obtained by further analyzing cases 2), 3) and 16).

\begin{table}
  \centering
  \caption{Cases for discussing when $\mc Z_4^1\cup \mc Z_5=\emptyset$ in Theorem~\ref{thm:odd}. In the column of $\mc Z_4^1$ (resp. $\mc Z_5$), ``*'' means that the set can be empty. In the column of $\mc Z_5$, ``-'' means that whether $\mc Z_5$ is empty is not required to be verified as $\mc Z_4^1$ is nonempty.}
  \label{tab:cases2col}
  \begin{tabular}{c|c|cccc|c|c}
    \hline \rowcolor[HTML]{EFEFEF} 
    case&$ w(\vv c_2\oplus \vv c_3) $ & $|2|$ & $|3|$ & $|4|$ & $|5|$ & $\mc Z_4^1$ & $\mc Z_5$ \\
    \hline
    1 & \multirow{8}{*}{odd} & odd & even & even & even & nonempty & - \\ 
    2 & & even & odd & even & even & * & * \\
    3 & & even & even & even & odd & * & * \\
    4 & & odd & odd & odd & even & nonempty & - \\
    5 & & odd & even & odd & odd & nonempty & - \\
    6 & & even & odd & odd & odd & nonempty & -\\
    7 & & odd & odd & even & odd & - & nonempty \\
    8 & & even & even & odd & even & - & nonempty \\
    \hline
    9 & \multirow{8}{*}{even} & even & even & even & even & \multirow{8}{*}{$\emptyset$} & nonempty  \\ 
    10 & & odd & odd & even & even & & nonempty \\
    11 & & odd & even & even & odd & & nonempty \\ 
    12 & & even & odd & odd & even & & nonempty \\
    13 & & odd & even & odd & even & & nonempty \\
    14 & & even & even & odd & odd & & nonempty \\
    15 & & odd & odd & odd & odd & & nonempty \\
    16 & & even & odd & even & odd & & * \\
    \hline
  \end{tabular}
\end{table}

\begin{lemma}\label{lemma:y4}
  When $ w(\vv c_2\oplus \vv c_3) $ is even, 	$\mc Z_4^1 = \emptyset$. 
  When $ w(\vv c_2\oplus \vv c_3) $ is odd, for the following four cases, $\mc Z_4^1$ is nonempty: 
  \begin{enumerate}
  \item[1)] $ |2| $ is odd and $ |3|$, $|4|$, $|5| $ are even, 
  \item[4)] $ |2|$, $|3|$, $|4| $ are odd and $ |5| $ is even, 
  \item[5)]  $ |2|$, $|4|$, $|5| $ are odd and $ |3| $ is even, and
  \item[6)]  $ |3|$, $|4|$, $|5| $ are odd and $ |2| $ is even.
  \end{enumerate}
\end{lemma}

\begin{IEEEproof}
  By~\eqref{eq:y41}, we have
  \begin{align*}
    \mc Z_4^1&=
               \{ y_1\neq y_2, d_{\mc P}'< d_{\mc O}\land d_{\mc O}'<d_{\mc P}, d_{\mc O} \land d_{\mc P}>d_{\mc O}' \}	
    \\
             &=   	\{ y_1\neq y_2, d_3'< d_{1,2,4}\land d_{1,2,4}'<d_{3}, d_{1,2,4} \land d_{3}>d_{1,2,4}' \}	  .
  \end{align*}
  For $ \cv y\in\mc Z_4^1 $,
  due to $ d_3'<d_3 $, we have $ y_1=1$, $y_2=0 $
  and it can be easily verified that $ \mc Z_4^1 $ can be rewritten as 
  \begin{equation*}
    \mc Z_4^1=\{y_1=1,y_2=0,  d_3\leq d_1\land d_4, d_2=d_3+1\}.
  \end{equation*}
  Then by \eqref{eq:d1}--\eqref{eq:d4}, $\vv y \in \mc Z_4^1$ if and only if $y_1=1$, $y_2=0$, and
  \begin{IEEEeqnarray}{rCl}
    w_2+w_3-\frac{|2|+|3|}{2} & = & w_4+w_5 - \frac{|4|+|5|-1}{2}, \label{eq:2cw1} \\
    w_2+w_3+w_6+w_7 & \geq & \frac{|2|+|3|+|6|+|7|}{2},\label{eq:2cw2} \\
    w_2+w_6 - \frac{|2|+|6|}{2} & \geq & w_1+w_5-\frac{|1|+|5|}{2}.\label{eq:2cw3}
  \end{IEEEeqnarray}
  
  When $ w(\vv c_2\oplus \vv c_3)=|2|+|3|+|4|+|5| $ (ref. \eqref{eq:w4}) is even, we have $ \mc Z_4^1=\emptyset $ since \eqref{eq:2cw1} cannot be satisfied for any $\vv y$. %
  When $ w(\vv c_2\oplus \vv c_3) = |2|+|3|+|4|+|5| $ is odd, there are eight cases corresponding to the parities for $ |2|$, $|3|$, $|4|\text{ and }|5| $ (see Table~\ref{tab:cases2col}). For the cases 1), 4), 5) and 6), we have values of $w_1\geq 1,w_2,\ldots,w_6,w_7<|7|$ satisfying \eqref{eq:2cw1}--\eqref{eq:2cw3} and hence $ \mc Z_4^1\neq \emptyset $:
  \begin{itemize}
  \item[1)]
    When $ |2| $ is odd and $ |3|$, $|4|$, $|5| $ are even, $w_1=1$, $w_2=\frac{|2|+1}{2}$, $w_3=\frac{|3|}{2}, \ w_4=\frac{|4|}{2}, \ w_5=\frac{|5|}{2}$, $w_6=|6|$, $w_7=|7|-1 $.

  \item[4)] When $ |2|$, $|3|$, $|4| $ are odd and $ |5| $ is even, $w_1=1$, $w_2=\frac{|2|+1}{2}$, $w_3=\frac{|3|+1}{2}$, $w_4=\frac{|4|+1}{2}$,  $w_5=\frac{|5|}{2}$, $w_6=|6|$, $w_7=|7|-1 $.

  \item[5)] When  $ |2|$, $|4|$, $|5| $ are odd and $ |3| $ is even, $w_1=1$, $w_2=\frac{|2|+1}{2}$, $w_3=\frac{|3|}{2}, \ w_4=\frac{|4|+1}{2}, \ w_5=\frac{|5|-1}{2}$, $w_6=|6|$, $w_7=|7|-1 $. 
  \item[6)]
    When  $ |3|$, $|4|$, $|5| $ are odd and $ |2| $ is even, $w_1=1$, $w_2=\frac{|2|}{2}$, $w_3=\frac{|3|+1}{2}, \ w_4=\frac{|4|+1}{2}, \ w_5=\frac{|5|-1}{2}$, $w_6=|6|$, $w_7=|7|-1 $. 
  \end{itemize}
\end{IEEEproof}

\begin{lemma}
  \label{lemma:z5}
  For the following nine cases, $\mc Z_5$ is nonempty: 
  \begin{enumerate}
  \item[7)] $ |2| $, $|3|$ and $|5|$ are odd, $|4|$ is even;
  \item[8)] $ |2| $, $|3|$ and $|5|$ are even, $|4|$ is odd; and 
  \item[9--16)] $ w(\vv c_2\oplus \vv c_3) $ is even, except for the case that $ |2| $ and $|4|$ are even, $|3|$ and $|5|$ are odd.
  \end{enumerate}
  
\end{lemma}
\begin{IEEEproof}
  By~\eqref{eq:y25},  $ \mc Z_5 $ can be rewritten as 
  \begin{IEEEeqnarray*}{rCl}
    \mc Z_5&=&\{y_1=1,y_2=0, d_2\leq d_3\leq d_{\{1,2,4\}}+1 \}.
  \end{IEEEeqnarray*}

  When $ w(\vv c_2\oplus \vv c_3)=|2|+|3|+|4|+|5| $ is odd,
  $\mc Z_5 = \{y_1=1,y_2=0, d_2= d_3-1\leq d_{\{1,4\}} \}$.  By
  \eqref{eq:d1}--\eqref{eq:d4}, $\vv y \in \mc Z_5$ if and only if
  $y_1=1$, $y_2=0$ and
  \begin{IEEEeqnarray}{rCl}
    w_2+w_3-w_4-w_5& = &  \frac{|2|+|3|-|4|-|5|-1}{2}, \label{eq:4cw1} \\
    w_2+w_3+w_6+w_7 & \geq & \frac{|2|+|3|+|6|+|7|-1}{2},\label{eq:4cw2} \\
	w_1-w_2+w_5-w_6 & \leq & \frac{|1|-|2|+|5|-|6|+1}{2}.\label{eq:4cw3}
  \end{IEEEeqnarray}
  For the cases 7) and 8) as given in Table~\ref{tab:cases2col}, we have values of $w_1\geq 1,w_2,\ldots,w_6,w_7<|7|$ satisfying \eqref{eq:4cw1}--\eqref{eq:4cw3} and hence $ \mc Z_5\neq \emptyset $:
\begin{enumerate}
\item[7)] When $ |4| $ is even and $ |2|$, $|3|$, $|5| $ are odd,  $ w_1=1$, $w_2=\frac{|2|+1}{2}$, $w_3=\frac{|3|-1}{2}, \ w_4=\frac{|4|}{2}, \ w_5=\frac{|5|+1}{2}$, $w_6=|6|$, $w_7=|7|-1 $.  
\item[8)] When $ |4| $ is odd and $ |2|$, $|3|$, $|5| $ are even, $ w_1=1$, $w_2=\frac{|2|}{2}$, $w_3=\frac{|3|}{2}, \ w_4=\frac{|4|+1}{2}, \ w_5=\frac{|5|}{2}$, $w_6=|6|$, $w_7=|7|-1 $.  
\end{enumerate}

When $ w(\vv c_2\oplus \vv c_3) $ is even,
$\mc Z_5=\{y_1=1,y_2=0, d_2=d_3\leq d_{\{1,4\}}+1 \}$. By
\eqref{eq:d1}--\eqref{eq:d4}, $\vv y \in \mc Z_5$ if and only if
$y_1=1$, $y_2=0$, and
\begin{IEEEeqnarray}{rCl}
	w_2+w_3-w_4-w_5& = &  \frac{|2|+|3|-|4|-|5|}{2}, \label{eq:3cw1} \\
	w_2+w_3+w_6+w_7 & \geq & \frac{|2|+|3|+|6|+|7|-1}{2},\label{eq:3cw2} \\
	w_1-w_2+w_5-w_6 & \leq & \frac{|1|-|2|+|5|-|6|+1}{2}.\label{eq:3cw3}
\end{IEEEeqnarray}
For the cases 9) -- 15) as given in Table~\ref{tab:cases2col}, we have values of $w_1\geq 1,w_2,\ldots,w_6,w_7<|7|$ satisfying \eqref{eq:3cw1}--\eqref{eq:3cw3} and hence $ \mc Z_5 \neq \emptyset $:
\begin{enumerate}
\item[9)] When $ |2|$, $|3|$, $|4|$, $|5| $ are even, $ w_1=1$, $w_2=\frac{|2|}{2}$, $w_3=\frac{|3|}{2}, \ w_4=\frac{|4|}{2}, \ w_5=\frac{|5|}{2}$, $w_6=|6|$, $w_7=|7|-1 $.  
  
\item[10)] When $ |2|$, $|3| $ are odd and $ |4|$, $|5| $ are even, $ w_1=1$, $w_2=\frac{|2|+1}{2}$, $w_3=\frac{|3|-1}{2}, \ w_4=\frac{|4|}{2}, \ w_5=\frac{|5|}{2}$, $w_6=|6|$, $w_7=|7|-1 $.  
  
\item[11)] When $ |2|$, $|5| $ are odd and $ |3|$, $|4| $ are even, $ w_1=1$, $w_2=\frac{|2|+1}{2}$, $w_3=\frac{|3|}{2}, \ w_4=\frac{|4|}{2}, \ w_5=\frac{|5|+1}{2}$, $w_6=|6|$, $w_7=|7|-1 $.  
  
\item[12)] When $ |2|$, $|5| $ are even and $ |3|$, $|4| $ are odd,  $ w_1=1$, $w_2=\frac{|2|}{2}$, $w_3=\frac{|3|+1}{2}, \ w_4=\frac{|4|+1}{2}, \ w_5=\frac{|5|}{2}$, $w_6=|6|$, $w_7=|7|-1 $.  
  
\item[13)] When $ |2|$, $|4| $ are odd and $ |3|$, $|5| $ are even,  $ w_1=1$, $w_2=\frac{|2|+1}{2}$, $w_3=\frac{|3|}{2}, \ w_4=\frac{|4|+1}{2}, \ w_5=\frac{|5|}{2}$, $w_6=|6|$, $w_7=|7|-1 $.  
  
\item[14)] When $ |2|$, $|3| $ are even and $ |4|$, $|5| $ are odd,  $ w_1=1$, $w_2=\frac{|2|}{2}$, $w_3=\frac{|3|}{2}, \ w_4=\frac{|4|+1}{2}, \ w_5=\frac{|5|-1}{2}$, $w_6=|6|$, $w_7=|7|-1 $.  
\item[15)] When $ |2|$, $|3|$, $|4|$, $|5| $ are odd, $ w_1=1$, $w_2=\frac{|2|+1}{2}$, $w_3=\frac{|3|-1}{2}, \ w_4=\frac{|4|+1}{2}, \ w_5=\frac{|5|-1}{2}$, $w_6=|6|$, $w_7=|7|-1 $.  
\end{enumerate}
\end{IEEEproof}

In the reminder of the proof, we discuss cases 2), 3) and 16) as given in Table~\ref{tab:cases2col}.

\begin{lemma}\label{lm:2}
  For the case 2) that  $ |3| $ is odd and $ |2|$, $|4|$, $|5| $ are even, $\mc Z_4^1\cup \mc Z_5 = \emptyset$ if and only if $|1|=|7|=1$, $|2|=|4|=|6|=0$, $|3|$ is odd, and $|5|$ is even.
\end{lemma}
\begin{IEEEproof}
 For case 2), we discuss six sub-cases to verify whether $\mc Z_4^1$ is empty. For the first four sub-cases 2-1)--2-4), we have values of $w_1\geq 1,w_2,\ldots,w_6,w_7<|7|$ satisfying \eqref{eq:2cw1}--\eqref{eq:2cw3} and hence $ \mc Z_4^1\neq \emptyset $:
  \begin{enumerate}
  \item[2-1)] When $ |2|>0 $, $w_1=1$, $w_2=\frac{|2|}{2}+1$, $w_3=\frac{|3|-1}{2}, \ w_4=\frac{|4|}{2}, \ w_5=\frac{|5|}{2}$, $w_6=|6|$, $w_7=|7|-1 $.
  \item[2-2)] When $ |4|>0$ and $|5| >0$,  $ w_1=1$, $w_2=\frac{|2|}{2}$, $w_3=\frac{|3|+1}{2}, \ w_4=\frac{|4|}{2}+1, \ w_5=\frac{|5|}{2}-1$, $w_6=|6|$, $w_7=|7|-1 $.
  \item[2-3)] When $ |1|+|6|\geq 2$, $ w_1=1$, $w_2=\frac{|2|}{2}$, $w_3=\frac{|3|+1}{2}, \ w_4=\frac{|4|}{2}, \ w_5=\frac{|5|}{2}$, $w_6=|6|$, $w_7=|7|-1 $.
  \item[2-4)] When $ |5|>0\text{ and }|6|+|7|\geq 3 $,  $ w_1=1$, $w_2=\frac{|2|}{2}$, $w_3=\frac{|3|-1}{2}, \ w_4=\frac{|4|}{2}, \ w_5=\frac{|5|}{2}-1$, $w_6=|6|$, $w_7=|7|-1 $.
  \end{enumerate}
  For the remaining two cases, $ \mc Z_4^1$ can be empty: 
  \begin{enumerate}
\item[2-5)] When  $|1|=1$, $|2|=0$, $|5|=0$ and $|6|=0$, we must have $w_1=1$, but \eqref{eq:2cw3} is not satisfied. Thus $ \mc Z_4^1=\emptyset  $ if  $|1|=1$, $|2|=0$, $|5|=0$, $|6|=0$, $|3|$ is odd and $|4|$ is even.
  \item[2-6)] When  $|1|=1$, $|2|=0$, $|4|=0$, $|6|=0$ and $1\leq |7|\leq 2$, for $ \vv y\in \mc Z_4^1 $,  
  by~\eqref{eq:2cw1}, we get
  \begin{equation}
    w_3-w_5=\frac{|3|-|5|+1}{2} .\label{eq:o1a}
  \end{equation}
  Since  $|5|$ is even, by~\eqref{eq:2cw3} we get $ w_5\leq \frac{|5|-2}{2} $ and thus from \eqref{eq:o1a}, we get
  \begin{equation}
    w_3\leq \frac{|3|-1}{2} .\label{eq:o1b}
  \end{equation}
By~\eqref{eq:2cw2}, we get
  \begin{IEEEeqnarray}{rCl}
    w_3+w_7 &\geq & \frac{|3|+|7|}{2}.\label{eq:o1c}
  \end{IEEEeqnarray}
  By \eqref{eq:o1b} and \eqref{eq:o1c}, 
  \begin{equation}
    w_7\geq \frac{|7|+1}{2}. \label{eq:o1d}
  \end{equation}
  Since $ y_2=0 $, 
  $w_7\leq |7|-1$,
  which implies $ |7|\geq 3 $ together with \eqref{eq:o1d}. We get a contradiction with $ |7|\leq 2 $. Hence $ \mc Z_4^1=\emptyset $ when $|1|=1$, $|2|=0$, $|4|=0$, $|6|=0$, $1\leq |7|\leq 2$, $|3|$ is odd and $|5|$ is even.
  \end{enumerate}

 For case 2), we discuss six sub-cases to verify whether $\mc Z_5$ is empty. For the first four sub-cases 2-1$'$)--2-4$'$), we have values of $w_1\geq 1,w_2,\ldots,w_6,w_7<|7|$ satisfying \eqref{eq:4cw1}--\eqref{eq:4cw3} and hence $ \mc Z_5\neq \emptyset $:
  \begin{enumerate}
  \item[2-1$'$)] When $ |6|+|7|\geq 2 $, $ w_1=1$, $w_2=\frac{|2|}{2}$, $w_3=\frac{|3|-1}{2}, \ w_4=\frac{|4|}{2}, \ w_5=\frac{|5|}{2}$, $w_6=|6|$, $w_7=|7|-1 $.
  \item[2-2$'$)] When $ |4|>0 $, $ w_1=1$, $w_2=\frac{|2|}{2}$, $w_3=\frac{|3|+1}{2}, \ w_4=\frac{|4|}{2}+1, \ w_5=\frac{|5|}{2}$, $w_6=|6|$, $w_7=|7|-1 $.
  \item[2-3$'$)] When $ |2|>0$ and $|5|>0 $, $ w_1=1$, $w_2=\frac{|2|}{2}+1$, $w_3=\frac{|3|-1}{2}, \ w_4=\frac{|4|}{2}, \ w_5=\frac{|5|}{2}+1$, $w_6=|6|$, $w_7=|7|-1 $.
  \item[2-4$'$)]  When $ |5|>0$ and $|1|+|6|\geq 3 $, $ w_1=1$, $w_2=\frac{|2|}{2}$, $w_3=\frac{|3|+1}{2}, \ w_4=\frac{|4|}{2}, \ w_5=\frac{|5|}{2}+1$, $w_6=|6|$, $w_7=|7|-1 $.
  \end{enumerate}
  For the remaining two cases, $ \mc Z_5$ can be empty: 
\begin{enumerate}
\item[2-5$'$)] When $ |4|=|5|=|6|=0$ and $|7|=1 $, $w_7=0$ and hence \eqref{eq:4cw1} and \eqref{eq:4cw2} cannot be satisfied simultaneously. 
    Thus $ \mc Z_5=\emptyset $ when $ |4|=|5|=|6|=0$, $|7|=1 $, $|2|$ is even and $|3|$ is odd.
  \item[2-6$'$)] When $ |2|=|4|=|6|=0$, $|7|=1 $ and $1\leq |1|\leq 2$, $w_7=0$.
    For $ \vv y\in \mc Z_5$, 
    by~\eqref{eq:4cw1}, 
    \begin{equation}
      w_3-w_5=\frac{|3|-|5|-1}{2}. \label{eq:o3a}
    \end{equation}
    Since  $ |3| $ is odd, by~\eqref{eq:4cw2}, 
    \begin{equation}
      w_3\geq \frac{|3|+1}{2}. \label{eq:o3b}
    \end{equation}
    By~\eqref{eq:4cw3}, 
    \begin{equation}
      w_1+w_5\leq \frac{|1|+|5|+1}{2}. \label{eq:o3c}
    \end{equation}
    By \eqref{eq:o3a}, \eqref{eq:o3b} and \eqref{eq:o3c}, $ w_1=0 $ which contradicts with $ w_1
    \geq 1 $. Thus $ \mc Z_5=\emptyset $ when $ |2|=|4|=|6|=0$, $|7|=1 $, $1\leq |1|\leq 2$, $|3|$ is odd and $|5|$ is even.
  \end{enumerate}
  
  Combining the sufficient and necessary conditions for $\mc Z_4^1=\emptyset$ and the sufficient and necessary conditions  for $\mc Z_5=\emptyset$, we get that $\mc Z_4^1\cup \mc Z_5=\emptyset$ if and only if $|1|=|7|=1$, $|2|=|4|=|6|=0$, and $|3|$ is odd, $|5|$ is even.
\end{IEEEproof}

\begin{lemma}\label{lm:3}
  For the case 3) that $ |5| $ is odd and $ |2|$, $|3|$, $|4| $ are even, $\mc Z_4^1\cup \mc Z_5 = \emptyset$ if and only if $|1|=|7|=1$, $|2|=|4|=|6|=0$,  $|3|$ is even, and $|5|$ is odd.
\end{lemma}
\begin{IEEEproof}
 For the case 3), we discuss six sub-cases to verify whether $\mc Z_4^1$ is empty. For the first four sub-cases, we have values of $w_1\geq 1,w_2,\ldots,w_6,w_7<|7|$ satisfying \eqref{eq:2cw1}--\eqref{eq:2cw3} and hence $ \mc Z_4^1\neq \emptyset $:
  \begin{enumerate}
  \item[3-1)] When $ |2|>0 $,  $ w_1=1$, $w_2=\frac{|2|}{2}+1$, $w_3=\frac{|3|}{2}, \ w_4=\frac{|4|}{2}, \ w_5=\frac{|5|+1}{2}$, $w_6=|6|$, $w_7=|7|-1 $.
  \item[3-2)] When $ |6|+|7|\geq 2 $,  $ w_1=1$, $w_2=\frac{|2|}{2}$, $w_3=\frac{|3|}{2}, \ w_4=\frac{|4|}{2}, \ w_5=\frac{|5|-1}{2}$, $w_6=|6|$, $w_7=|7|-1 $.
  \item[3-3)] When $ |3|>0 $ and $ |1|+|6|\geq 3 $, $ w_1=1$, $w_2=\frac{|2|}{2}$, $w_3=\frac{|3|}{2}+1, \ w_4=\frac{|4|}{2}, \ w_5=\frac{|5|+1}{2}$, $w_6=|6|$, $w_7=|7|-1 $.
  \item[3-4)] When $ |3|>0$ and $|4|>0 $,  $ w_1=1$, $w_2=\frac{|2|}{2}$, $w_3=\frac{|3|}{2}+1, \ w_4=\frac{|4|}{2}+1, \ w_5=\frac{|5|-1}{2}$, $w_6=|6|$, $w_7=|7|-1 $.
  \end{enumerate}
  For the remaining two cases, $ \mc Z_4^1$ can be empty: 
  \begin{enumerate}
\item[3-5)] When $|2|=|3|=|6|=0, |7|=1 $, \eqref{eq:2cw2} cannot be satisfied. Then $ \mc Z_4^1=\emptyset $ when $|2|=|3|=|6|=0, |7|=1 $, $|4|$ is even and $|5|$ is odd;
  \item[3-6)] When $ |2|=|4|=|6|=0$, $|7|=1$ and $1\leq |1|\leq 2 $, $w_7=0$. By \eqref{eq:2cw1}--\eqref{eq:2cw3},
    \begin{IEEEeqnarray}{rCl}
      w_3&\geq& \frac{|3|+1}{2},  \label{eq:o2aa}\\
      w_3-w_5&=& \frac{|3|-|5|+1}{2},  \label{eq:o2b}\\
      w_1+w_5&\leq& \frac{|1|+|5|}{2},  \label{eq:o2c}
    \end{IEEEeqnarray}
  which implies $ w_1=1 $. Then $ w_5\leq \frac{|5|-1}{2} $ by \eqref{eq:o2c}. Thus $ w_3\leq \frac{|3|}{2} $ by \eqref{eq:o2b}, which is a contradiction with \eqref{eq:o2aa}. Thus we have $\mc Z_4^1=\emptyset$ when  $ |2|=|4|=|6|=0, |7|=1$, $1\leq |1|\leq 2 $, $|3|$ is even and $|5|$ is odd.
  \end{enumerate}

 For case 3), we discuss six sub-cases to verify whether $\mc Z_5$ is empty.
For the first four sub-cases 3-1$'$)--3-4$'$) we have values of $w_1\geq 1,w_2,\ldots,w_6,w_7<|7|$ satisfying \eqref{eq:4cw1}--\eqref{eq:4cw3} and hence $ \mc Z_5\neq \emptyset $:
  \begin{enumerate}
  \item[3-1$'$)] When $ |1|+|6|\geq 2 $, $ w_1=1$, $w_2=\frac{|2|}{2}$, $w_3=\frac{|3|}{2}, \ w_4=\frac{|4|}{2}, \ w_5=\frac{|5|+1}{2}$, $w_6=|6|$, $w_7=|7|-1 $.
  \item[3-2$'$)] When $ |4|>0 $ and $ w_1=1$, $w_2=\frac{|2|}{2}$, $w_3=\frac{|3|}{2}, \ w_4=\frac{|4|}{2}+1, \ w_5=\frac{|5|-1}{2}$, $w_6=|6|$, $w_7=|7|-1 $.
  \item[3-3$'$)] When $ |3|>0$ and $|6|+|7|\geq 3 $, $ w_1=1$, $w_2=\frac{|2|}{2}$, $w_3=\frac{|3|}{2}-1, \ w_4=\frac{|4|}{2}, \ w_5=\frac{|5|-1}{2}$, $w_6=|6|$, $w_7=|7|-1 $.
  \item[3-4$'$)] When $ |2|>0$ and $|3|>0 $, $ w_1=1$, $w_2=\frac{|2|}{2}+1$, $w_3=\frac{|3|}{2}-1, \ w_4=\frac{|4|}{2}, \ w_5=\frac{|5|+1}{2}$, $w_6=|6|$, $w_7=|7|-1 $.
    \end{enumerate}
For the remaining two cases, $ \mc Z_5$ can be empty: 
\begin{enumerate}
  \item[3-5$'$)] When $ |1|=1$ and $ |3|=|4|=|6|=0 $, \eqref{eq:4cw1} and \eqref{eq:4cw3}  cannot be satisfied simultaneously.   %
    Thus $ \mc Z_5=\emptyset $ when $ |1|=1$, $ |3|=|4|=|6|=0 $, $|2|$ is even and $|5|$ is odd;
  \item[3-6$'$)]
    Otherwise when $ |1|=1$, $|2|=|4|=|6|=0$ and $ 1\leq |7|\leq 2 $, $w_1=1$. For $\vv y \in \mc Z_5$,
    by~\eqref{eq:4cw1} and~\eqref{eq:4cw2}, we have
    \begin{IEEEeqnarray}{rCl}
      w_3-w_5&=&\frac{|3|-|5|-1}{2},\label{eq:o4a}\\
      w_3+w_7&\geq& \frac{|3|+|7|-1}{2}.\label{eq:o4b}
    \end{IEEEeqnarray}
    Since  $ |5| $ is odd, by~\eqref{eq:4cw3},
    \begin{equation}
      w_5\leq \frac{|5|-1}{2}.\label{eq:o4c}
    \end{equation}
    By \eqref{eq:o4a}, \eqref{eq:o4b}, \eqref{eq:o4c} and $w_7\leq |7|-1$, we get $ |7|\geq 3 $, which is a contradiction. Thus $ \mc Z_5=\emptyset $ when $ |1|=1$, $|2|=|4|=|6|=0$, $1\leq |7|\leq 2 $, $|3|$ is even and $|5|$ is odd.
  \end{enumerate}

Combining the sufficient and necessary conditions for $\mc Z_4^1=\emptyset$ and the sufficient and necessary conditions  for $\mc Z_5=\emptyset$, we get that $\mc Z_4^1\cup \mc Z_5=\emptyset$ if and only if $|1|=|7|=1$, $|2|=|4|=|6|=0$, and $|3|$ is even, $|5|$ is odd.
\end{IEEEproof}

\begin{lemma}\label{lm:16}
  For the case 16) that $ |2|$, $|4| $ are even and $ |3|$, $|5| $ are odd, $ \mc Z_5=\emptyset $ if and only if $|1|=|7|=1$, $|2|=|4|=|6|=0$, and $|3|$, $|5|$ are both odd.
\end{lemma}
\begin{IEEEproof}
We discuss this case in five sub-cases. For the first four sub-cases 16-1)--16-4) we have values of $w_1\geq 1,w_2,\ldots,w_6,w_7<|7|$ satisfying \eqref{eq:3cw1}--\eqref{eq:3cw3} and hence $ \mc Z_5\neq \emptyset $:
  \begin{enumerate}
  \item[16-1)] When $ |6|+|7|\geq 2 $,  $ w_1=1$, $w_2=\frac{|2|}{2}$, $w_3=\frac{|3|-1}{2}, \ w_4=\frac{|4|}{2}, \ w_5=\frac{|5|-1}{2}$, $w_6=|6|$, $w_7=|7|-1 $.
  \item[16-2)] When $ |1|+|6|\geq 2 $, $ w_1=1$, $w_2=\frac{|2|}{2}$, $w_3=\frac{|3|+1}{2}, \ w_4=\frac{|4|}{2}, \ w_5=\frac{|5|+1}{2}$, $w_6=|6|$, $w_7=|7|-1 $.
  \item[16-3)] When $ |4|>0 $, $ w_1=1$, $w_2=\frac{|2|}{2}$, $w_3=\frac{|3|+1}{2}, \ w_4=\frac{|4|}{2}+1, \ w_5=\frac{|5|-1}{2}$, $w_6=|6|$, $w_7=|7|-1 $.
  \item[16-4)] When $ |2|>0 $, $ w_1=1$, $w_2=\frac{|2|}{2}+1$, $w_3=\frac{|3|-1}{2}, \ w_4=\frac{|4|}{2}, \ w_5=\frac{|5|+1}{2}$, $w_6=|6|$, $w_7=|7|-1 $.
  \end{enumerate}
For the remaining sub-case, $ \mc Z_5$ can be empty: 
\begin{enumerate}
  \item[16-5)]   When $ |1|=|7|=1$ and $|2|=|4|=|6|=0 $, $w_1=1$ and $w_7=0$. 
By~\eqref{eq:3cw1} and \eqref{eq:3cw2}, $ w_5\geq \frac{|5|+1}{2} $. While  $ |5| $ is odd and by~\eqref{eq:3cw3}, $ w_5\leq \frac{|5|-1}{2} $, which is a contradiction. Thus $ \mc Z_5=\emptyset $ when $ |1|=|7|=1, |2|=|4|=|6|=0 $, and $|3|$, $|5|$ are odd.
  \end{enumerate}
\end{IEEEproof}

	By Lemma~\ref{lm:2}, \ref{lm:3} and \ref{lm:16}, we can conclude that $ \mc Z_4^1\cup \mc Z_5=\emptyset $ if and only if $ |1|=|7|=1$, $|2|=|4|=|6|=0 $ and at least one of $ |3|$ and $ |5| $ is odd.

\section{Optimal Linear $(n,4)$ Codes: Proof of Theorem~\ref{thm:linearcompare}}\label{sec:linearcode}
In this section, we apply the comparing technique of codes with difference in one column (ref. \S\ref{sec:1col}) on linear codes, and prove Theorem~\ref{thm:linearcompare}.

Recall the definition of the $(n,4)$ linear code $C(n_3,n_5,n_6)$ in Definition~\ref{def:linear}, where at two of $n_3$, $n_5$ and $n_6$ are nonzero. Consider an $(n,4)$ linear code $C(n_3,n_5,n_6)$ with $n_3>0$.
  Let $ C'$ be the code obtained by replacing a column of type $\bspan{3}$ of $C$ by $\bspan{5}$. Theorem~\ref{thm:linearcompare} claims that
  \begin{enumerate}
  \item When $n_3, n_5+1,n_6$ have the same parity, $\lambda_{ C'}= \lambda_ C$;
  \item When $n_3, n_5,n_6$ have the same parity,%
    \begin{itemize}
    \item if $n_3=1$,  $\lambda_{C'}= \lambda_ C$, and
    \item  if $n_3\geq 2$,  $\lambda_{ C'}>\lambda_C$.
    \end{itemize}
  \item When $n_5\leq \min\{n_3, n_6\} $ and  $n_3-1, n_5,n_6$ have the same parity, 
    \begin{itemize}
    \item if $n_3= n_5+1$,  $\lambda_{ C'}= \lambda_C$, and
    \item if $n_3> n_5+1$,  $\lambda_{C'}> \lambda_ C$.
    \end{itemize}
  \end{enumerate}
  We prove the above claims in this section.
  
\subsection{Formulae of $\alpha_C^3(d)$ and $\alpha_C^5(d)$}

WLOG, we assume the first $n_3$ columns of $C(n_3,n_5,n_6)$ are $\bspan{3}$ and the last $n_6$ columns are $\bspan{6}$.
By the discussion in Appendix~\ref{sec:1col}, we have $\mc O=\{1,4\}$ and $\mc P=\{2,3\}$. Substituting $s=\bspan{3}$ in \eqref{eq:d1di}, we have
\begin{IEEEeqnarray*}{rCl}
  d_1'(\vv y)&=& d_1(\vv y) - y_1 + \overline{y_1}, \\
  d_2'(\vv y)&=& d_2(\vv y) - y_1 + \overline{y_1},  \\
  d_3'(\vv y)&=& d_3(\vv y) + y_1 - \overline{y_1},  \\
  d_4'(\vv y)&=& d_4(\vv y) + y_1 - \overline{y_1},
\end{IEEEeqnarray*}
where by \eqref{eq:d}
\begin{IEEEeqnarray*}{rCl}
  d_1 & = & w_3+w_5+w_6, \\
  d_2 & = & w_3+\overline{w_5}+\overline{w_6}, \\
  d_3 & = & \overline{w_3}+w_5+\overline{w_6}, \\
  d_4 & = & \overline{w_3}+\overline{w_5}+w_6. 
\end{IEEEeqnarray*}

To use Corollary \ref{cor:1}, we rewrite $\mc Y_3$ and $\mc Y_5$ (defined in \eqref{eq:y3} and \eqref{eq:y5}, respectively) as
\begin{align*}
	\mc Y_3 & =  \{d_{\mc P}' = d_{\mc O}' < d_{\mc P} = d_{\mc O}\} \\
	& = \{d_2\land d_3 = d_1\land d_4 =  d_2'\land d_3'+1 = d_1'\land d_4'+1  \},
\end{align*}
and 
\begin{align*}
	\mc Y_5 &= \{ d_{\mc P}' = d_{\mc O} < d_{\mc O}' = d_{\mc P} \}\\
	&= \{d_2'\land d_3'+1=d_1\land d_4 +1 = d_2\land d_3 = d_1'\land d_4' \}.
\end{align*}
We analyze $ \mc Y_3 $ and $\mc Y_5$ to simplify the formula.
First, we have
\begin{IEEEeqnarray*}{rCl}
  \mc Y_3(d,1) & \triangleq & \{\vv y\in \mc Y_3: d_{C}(\vv y) = d, y_1 = 1\} \\
	&= &\{y_1=1, d_2\land d_3 = d_1\land d_4 =d, d_2'\land d_3'= d_1'\land d_4' = d-1\} \\
	&= &\Big\{ y_1=1,  d_2\land d_3 = d_1\land d_4 =d, (d_2-1)\land (d_3+1)=
	 (d_1-1)\land (d_4+1) = d-1 \Big\} \\
	&=& \{y_1=1, d_1=d_2=d, d_3\geq d, d_4\geq d \} \\
	&=&\Big \{y_1=1, w_3+w_5+w_6 = w_3+\overline{w_5} +\overline{w_6} = d, \overline{w_3} +w_5+\overline{w_6}\geq d,\overline{w_3}+\overline{w_5}+w_6\geq d \Big\} \IEEEyesnumber \label{eq:cli4}\\
	&=& \Big\{y_1= 1, w_3 =d- \frac{n_5+n_6}{2}, w_6= \frac{n_5+n_6}{2}-w_5, d- \frac{n_3+n_6}{2} \leq w_5\leq n_5-d+ \frac{n_3+n_6}{2}\Big\},\IEEEyesnumber \label{eq:cli7}
\end{IEEEeqnarray*}
and
\begin{IEEEeqnarray*}{rCl}
\mc Y_3(d,0) & \triangleq & \{\vv y\in \mc Y_3: d_{C}(\vv y) = d, y_1 = 0\} \\ 
&= &\{y_1=0, d_2\land d_3 = d_1\land d_4 =d, d_2'\land d_3'= d_1'\land d_4' = d-1\} \\
&= &\Big\{ y_1=0,  d_2\land d_3 = d_1\land d_4 =d, (d_2+1)\land (d_3-1)=(d_1+1)\land (d_4-1) = d-1 \Big\} \\
&=& \{y_1=0, d_3=d_4=d, d_1\geq d, d_2\geq d \} \\
&=& \Big\{y_1=0,  \overline{w_3}+\overline{w_5}+w_6 =\overline{w_3} +w_5+\overline{w_6}= d,  w_3+\overline{w_5} +\overline{w_6}\geq d, w_3+w_5+w_6 \geq d \Big\}.\IEEEyesnumber  \label{eq:cli5}
\end{IEEEeqnarray*}
From \eqref{eq:cli4} and \eqref{eq:cli5},
$\cv y \in \{\vv y\in \mc Y_3: d_{C}(\vv y) = d, y_1 = 1\}  $  if and only if $  F_{n_3+n_5}(\cv y) \in \{\vv y\in \mc Y_3: d_{C}(\vv y) = d, y_1 = 0\} $,
where $F_{n_3+n_5}(\cv y) $ is the vector obtained by flipping the first $n_3+n_5$ bits of $\cv y$.
Thus 
\begin{align}
	\alpha_C^3(d) = \left| \left\{\vv y\in \mc Y_3: d_{C}(\vv y) = d\right\} \right| = \left| \mc Y_3(d,0) \right| + \left| \mc Y_3(d,1) \right| = 2\left| \mc Y_3(d,1) \right|. \label{eq:cli8}
\end{align}
Similarly we get
\begin{IEEEeqnarray*}{rCl}
    \mc Y_5(d,0) & = & \{\vv y\in \mc Y_5: d_{C}(\vv y) = d, y_1 = 0\}  \\
    &=& \{y_1=0, d_2'\land d_3' = d_1\land d_4  =d, d_2\land d_3 = d_1'\land d_4' = d+1\} \\
    &=& \{y_1=0, (d_2+1)\land (d_3-1) = d_1\land d_4  =d, d_2\land d_3 = (d_1+1)\land (d_4-1) = d+1\} \\
    &=&\{ y_1=0, d_1=d,d_3=d+1,d_4\geq d+2, d_2\geq d+1 \}\\
    &=& \{y_1=0,w_3+w_5+w_6 = \overline{w_3} +w_5+\overline{w_6}-1=d,   w_3+\overline{w_5}+\overline{w_6}\geq d+1,\overline{w_3} +\overline{w_5}+w_6 \geq d+2  \} \\
    &=& \Big\{ y_1=0, w_5 = d- \frac{n_3+n_6-1}{2}, w_6=\frac{n_3+n_6-1}{2}-w_3, \\
    & & \quad  d-  \frac{n_5+n_6-1}{2}\leq w_3\leq n_3-d+\frac{n_5+n_6-3}{2}\Big\}, \IEEEyesnumber \label{eq:cli6}
\end{IEEEeqnarray*}
and
\begin{align}
	\alpha_C^5(d) = 2\big| \mc Y_5(d,0)  \big|. \label{eq:cli9}
\end{align}

\subsection{Proof of Cases 1) and 2)}

To prove 1) and 2), consider
$n_3$ and $n_6$ have the same parity. As no integer $\omega_5$ satisfies the condition in \eqref{eq:cli6}, we have $\mc Y_5(d,0) =\emptyset$ for all possible $d$. Hence by Corollary \ref{cor:1}--3), $\lambda_{C'} \geq \lambda_C$ where the equality holds if and only if $\mc Y_3=\emptyset$.

Further, if $n_5$ is of the different parity as $n_3$ and $n_6$, we see that $\mc Y_3(d,1)=\emptyset$ for all possible $d$ by checking \eqref{eq:cli7} and hence 1) is proved.

To prove 2), consider $n_5$ is of the same parity as $n_3$ and $n_6$. By \eqref{eq:cli7}, for $\vv y \in \mc Y_3(d,1)$, $w_3\leq \frac{n_3}{2}$. So when $n_3=1$, $w_3=0$ which is a contradiction to $y_1=1$. Hence, when $n_3=1$, $\mc Y_3(d,1) = \emptyset$ for all possible $d$. When $n_3\geq 2$, consider two cases: 
\begin{itemize}
\item $n_3>1,n_5,n_6$ are all odd, and hence $n=n_3+n_5+n_6$ is odd. We see that $\mc Y_3(\frac{n-1}{2},1)$ includes all $\vv y$ with $w_3=\frac{n_3-1}{2}$, $w_5=\frac{n_5-1}{2}$ and $w_6=\frac{n_6+1}{2}$. 
\item $n_3,n_5,n_6$ are all even, and hence $n=n_3+n_5+n_6$ is even. We see that $\mc Y_3(\frac{n-1}{2},1)$ includes all $\vv y$ with $w_3=\frac{n_3}{2}$, $w_5=\frac{n_5}{2}$ and $w_6=\frac{n_6}{2}$. 
\end{itemize}

\subsection{Proof of Case 3)}

Last, we prove case 3), where $n_5 \leq \min\{n_3,n_6\}$ and  $n_3-1, n_5,n_6$ have the same parity. In this case, both $\mc Y_5$ and $\mc Y_3$ are not empty and hence to apply Corollary \ref{cor:1}, we need to derive the formula of $\alpha^3(d)$ and $\alpha^5(d)$, where the subscript $C$ will be omitted henceforth to simplify the notations.

By \eqref{eq:cli7}, \eqref{eq:cli8}, we have
\begin{align}
  \alpha_{d+1}^3 
  &=2 \sum_{w_5 =d- \frac{n_3+n_6-3}{2}  }^{ n_5-d+ \frac{n_3+n_6-3}{2} } \binom{n_3-1}{d- \frac{n_5+n_6}{2} }  \binom{n_5}{w_5}
    \binom{n_6}{ \frac{n_5+n_6}{2}-w_5}. \label{eq:cli3}
\end{align}
By \eqref{eq:cli6} and \eqref{eq:cli9}, we have 
\begin{IEEEeqnarray*}{rCl}
  \alpha^5(d) 
  &= & 2 \sum_{w_3 = d- \frac{n_5+n_6-2}{2} }^{	n_3-d+ \frac{n_5+n_6-4}{2} } \binom{n_3-1}{w_3}  \binom{n_5}{ d- \frac{n_3+n_6-1}{2}} \binom{n_6}{\frac{n_3+n_6-1}{2}-w_3}  %
  \\
  & = & 2 \sum_{\tilde{w}_3 =d- \frac{n_3+n_6-3}{2}  }^{	 n_5-d+ \frac{n_3+n_6-3}{2} } \binom{n_3-1}{\tilde{ w}_3-\frac{n_5-n_3+1}{2}}   \binom{n_5}{ d- \frac{n_3+n_6-1}{2}}\binom{n_6}{\frac{n_5+n_6}{2}-\tilde{w}_3}, \IEEEyesnumber \label{eq:cli2}
\end{IEEEeqnarray*}
where the last equality is obtained by substituting $\tilde{ w}_3= w_3+\frac{n_5-n_3+1}{2}$.
We can verify that
\begin{itemize}
\item when $d > \frac{n-3}{2}$, $\alpha^3(d)=\alpha^5(d)=0$;
\item when $d < \max\left(\frac{n_5+n_6}{2},\frac{n_3+n_6-3}{2} \right)$, $\alpha^3(d)=0$; and 
\item when $d < \frac{n_3+n_6-1}{2}$, $\alpha^5(d)=0$. 
\end{itemize}

When $n_3=n_5+1$, by comparing \eqref{eq:cli3} and \eqref{eq:cli2}, we have 
$\alpha^3(d+1) = \alpha^5(d)$ for $d=0,\dots, n-1$. From Corollary \ref{cor:1}--1), $\lambda_{C'}= \lambda_C$.
	
When $n_3> n_5+1$, since $n_3$ and $n_5$ has diverse parity, we have $n_3\geq n_5+3$.
	When $d< \frac{n_3+n_6-3}{2}$ or $d> \frac{n-3}{2}$, 
	\begin{equation*}
		\alpha^3(d+1) = \alpha^5(d)=0.
	\end{equation*}
	When $d= \frac{n_3+n_6-3}{2}$, 
	\begin{align*}
		\alpha^3(d+1) &=2 \sum_{w_5 =0  }^{ n_5 } \binom{n_3-1}{ \frac{n_3-n_5-3}{2} }  \binom{n_5}{w_5} \binom{n_6}{ \frac{n_5+n_6}{2}-w_5} \\
		&>0= \alpha^5(d) .
	\end{align*}
	When $d\in \big[\frac{n_3+n_6-1}{2}, \frac{n-3}{2}\big]$, 
	\begin{equation*}
		\alpha^3(d+1) > \alpha^5(d) ,
	\end{equation*}
	since 	for $ \tilde{ w}_3, w_5$ that satisfy $d- \frac{n_3+n_6-3}{2} \leq \tilde{ w}_3 = w_5 \leq n_5-d+ \frac{n_3+n_6-3}{2} $,
	\begin{align}
		&\binom{n_3-1}{\tilde{ w}_3-\frac{n_5-n_3+1}{2}}  \binom{n_5}{ d- \frac{n_3+n_6-1}{2}}  \nonumber\\
		=&  	
		\binom{n_3-1}{ \frac{n_3-1}{2} +\tilde{ w}_3-\frac{n_5}{2}}  \binom{n_5}{ \frac{n_5}{2}+d- \frac{n-1}{2}} \nonumber \\
		<&\binom{n_3-1}{  \frac{n_3-1}{2} +d- \frac{n-1}{2}}  \binom{n_5}{ \frac{n_5}{2} +\tilde{ w}_3-\frac{n_5}{2}} \label{ieq:cli1}\\
		=&	 \binom{n_3-1}{d- \frac{n_5+n_6}{2} }  \binom{n_5}{w_5} ,\nonumber
	\end{align}
	where the inequality in \eqref{ieq:cli1} holds due to a refined version of \cite[Claim 47]{chen2013optimal}  with strict inequality and 
	\begin{align*}
		\Big| d- \frac{n-1}{2}\Big|  =  \frac{n-1}{2} - d> \Big|\tilde{ w}_3-\frac{n_5}{2}\Big|.
	\end{align*}
        For the completeness, we prove the refined version of \cite[Claim 47]{chen2013optimal} in Appendix (see Lemma \ref{lemma:cli1}).
	By Corollary \ref{cor:1}--2),  $\lambda_{C'}> \lambda_C$.

\section{Analysis of Class-I Codes}
\label{sec:classI}

In this section, we study the ML decoding performance of a Class-I code and a code obtained by changing one column using the approach introduced in \S\ref{sec:1col}. The proofs of  Theorem~\ref{thm:11} and Theorem~\ref{thm:301} are given.

\subsection{Characterizations of $\mc Y_3$ and $\mc Y_5$}
\label{sec:35}

Recall the definition of Class-I codes in Definition~\ref{def:class1}.
We consider a Class-I $(n,4)$ code $C$ with the first column $\bspan{1}$. Let $C'$ be the code obtained by replacing the first column of $C$ to $\bspan{3}$.
By the discussion in \S\ref{sec:1col},
$\mc O=\{1,2,4\}$ and $\mc P=\{3\}$.
See the formulae of $d_{\mc O}$, $d_{\mc O}'$, $d_{\mc P}$ and $d_{\mc P}'$ in Example~\ref{ex:s1s3} for this case.
Guided by Theorem~\ref{the:1} and Corollary~\ref{cor:1}, we first study $\mc Y_3$ and $\mc Y_5$ defined in \eqref{eq:y3} and \eqref{eq:y5}.

For $\vv y \in \mc Y_3$, $d_{\mc P}'=d_3 - y_1 + \overline{y_1}  < d_3 = d_{\mc P}$ implies $y_1 = 1$. Hence we rewrite $\mc Y_3$ as
\begin{IEEEeqnarray*}{rCl}
  \mc Y_3 & = &  \{d_{\mc P}' = d_{\mc O}' < d_{\mc P} = d_{\mc O}\} \\
  & = & \{y_1=1, d_3-1 = [(d_1\land d_2)-1]\land (d_4+1) < d_3= d_1\land d_2\land d_4  \} \\
  & = &  \{y_1 = 1, d_4 \geq d_1\land d_2 = d_3\}. \IEEEyesnumber \label{eq:c13}
\end{IEEEeqnarray*}
For $\vv y \in \mc Y_5$, $d_3 - y_1 + \overline{y_1}  < d_3$ implies $y_1 = 1$. Hence, we rewrite $\mc Y_5$ as
\begin{IEEEeqnarray*}{rCl}
  \mc Y_5 & = & \{ d_{\mc P}' = d_{\mc O} < d_{\mc O}' = d_{\mc P} \} \\
  & = & \{y_1=1, d_3-1 = d_1\land d_2\land d_4 < [(d_1\land d_2)-1]\land (d_4+1) = d_3\} \\
  & = & \{y_1 = 1, d_1\land d_2 \geq d_4+2 = d_3+1\}. \IEEEyesnumber \label{eq:c15}
\end{IEEEeqnarray*}

\subsubsection{Characterization of $\alpha_C^3(i)$}

For $\vv y \in \mc Y_3$, by Lemma~\ref{lemma:1}, $d_C(\vv y ) = d_3$. 
By \eqref{eq:d1} -- \eqref{eq:d4} and \eqref{eq:c13}, we have the following necessary and sufficient condition for $\vv y \in \mc Y_3$ with $d_C(\vv y) = i$:  $y_1  =  1$ and
\begin{IEEEeqnarray*}{rCl}
  w_1 + \overline{w_3} & = & i - w_5 - \overline{w_6},  \\
  w_1 - \overline{w_1} & \leq & \overline{w_5} + w_6 - w_5 - \overline{w_6}, \\
  w_3 - \overline{w_3} & = & w_5 + \overline{w_6} - (w_5+ w_6)\land (\overline{w_5} + \overline{w_6}).
\end{IEEEeqnarray*}
We discuss two cases according to $w_5+w_6<\overline{w_5} + \overline{w_6}$ or not. Define 
$\mc Y_{3}^{A}(i)$ the subset of $\mc Y_3$ with $d_C(\vv y)=i$ and $w_5+w_6<\overline{w_5} + \overline{w_6}$, and define $\mc Y_{3}^{B}(i)$ the subset of $\mc Y_3$ with $d_C(\vv y)=i$ and $w_5+w_6\geq\overline{w_5} + \overline{w_6}$. We see that
\begin{IEEEeqnarray*}{rCl}
  \alpha_C^3(i)
  & = & |\{\vv y \in \mc Y_3: d_C(\vv y) = i\}|\nonumber \\
  & = & |\mc Y_3^{A}(i)| + |\mc Y_3^{B}(i)|, \label{eq:alpha3i}
\end{IEEEeqnarray*}
where $|\mc Y_3^{A}(i)|$ and $|\mc Y_3^{B}(i)|$ are characterized as follows.

As $\mc Y_{3}^{A}(i)$ is the collection of $\vv y$ satisfying $y_1=1$ and
\begin{IEEEeqnarray}{rCl}
  w_5+ w_6 & < & (|5|+|6|)/2,  \label{eq:3yc1} \\
  w_1+w_5 & = & i - (|3|+|6|)/2, \label{eq:3yc2}\\
  w_1 + w_5 - w_6 & \leq & (|1|+|5|-|6|)/2, \label{eq:3yc3} \\
 w_3 +  w_6& = & {(|3|+|6|)}/{2}, \label{eq:3yc4}
\end{IEEEeqnarray}
we have
\begin{equation*}
  |\mc Y_3^{A}(i)| =  \sum_{\substack{w_1\geq 1, w_3,w_5,w_6:\\\eqref{eq:3yc1},\eqref{eq:3yc2},\eqref{eq:3yc3},\eqref{eq:3yc4}}} \binom{|1|-1}{w_1-1}\binom{|3|}{w_3} \binom{|5|}{w_5}\binom{|6|}{ w_6}.\label{eq:c1b}
\end{equation*}
As $\mc Y_{3}^{B}(i)$ is the collection of $\vv y$ satisfying $y_1=1$ and
\begin{IEEEeqnarray}{rCl}
  w_5+ w_6 & \geq & {(|5|+|6|)}/{2}, \label{eq:3yd1} \\
   w_1 + \overline{w_6}  & = & i - {(|3|+|5|)}/{2}, \label{eq:3yd2} \\
  w_1 + w_5 - w_6 & \leq & {(|1|+|5|-|6|)}/{2}, \label{eq:3yd3} \\
  w_3 - w_5 & = & {(|3|-|5|)}/{2}, \label{eq:3yd4}
\end{IEEEeqnarray}
we have
\begin{equation}\label{eq:y3b}
  |\mc Y_3^{B}(i)| =  \sum_{\substack{w_1\geq 1, w_3,w_5,w_6:\\\eqref{eq:3yd1},\eqref{eq:3yd2},\eqref{eq:3yd3},\eqref{eq:3yd4}}} \binom{|1|-1}{w_1-1}\binom{|3|}{w_3} \binom{|5|}{w_5}\binom{|6|}{ w_6}.
\end{equation}

\subsubsection{Characterization of $\alpha_C^5(i)$}
For $\vv y \in \mc Y_5$, by Lemma~\ref{lemma:1}, $d_C(\vv y ) = d_3-1$. 
By \eqref{eq:d1} -- \eqref{eq:d4} and \eqref{eq:c15}, we have the following necessary and sufficient condition for $\vv y \in \mc Y_5$ with $d_C(\vv y) = i$:   $y_1  =  1$ and
\begin{IEEEeqnarray*}{rCl}
  w_1 + \overline{w_3} & = & i+1 - w_5 - \overline{w_6},\\
  w_1 - \overline{w_1} & = & \overline{w_5} + w_6 - w_5 - \overline{w_6} +1,\\
  w_3 - \overline{w_3} & \geq & w_5 + \overline{w_6} - (w_5+ w_6)\land (\overline{w_5} + \overline{w_6}) + 1,
\end{IEEEeqnarray*}
which can be further simplified as $y_1 = 1$ and 
\begin{IEEEeqnarray}{rCl}
  w_3 & = & {(n+|3|-1)}/{2} - i, \label{eq:5yc1} \\
  w_1 + w_5 - w_6  & = & {(|1|+|5|-|6|+1)}/{2}, \label{eq:5yc2} \\
  w_3 + w_6 & \geq & {(|3|+|6|)}/{2}+1,  \label{eq:5yc3}\\
  w_3 - w_5 & \geq & {(|3|-|5|)}/{2}+1.  \label{eq:5yc4}
\end{IEEEeqnarray}
Hence
\begin{equation}\label{eq:alpha5i}
  \alpha_C^5(i)  = \sum_{\substack{w_1\geq 1, w_3, w_5, w_6:\\\eqref{eq:5yc1},\eqref{eq:5yc2},\eqref{eq:5yc3},\eqref{eq:5yc4}}} \binom{|1|-1}{w_1-1}\binom{|3|}{w_3}\binom{|5|}{w_5}\binom{|6|}{w_6}.
\end{equation}

\subsection{Class-I Codes with $|1|=1$: Proof of Theorem~\ref{thm:11}}
\label{sec:11}

Following the discuss in the last subsection, we consider the special case with $|1|=1$, and  $|3| = \min\{|3|, |5|, |6|\}$. Theorem~\ref{thm:11} states that $\lambda_{C'} > \lambda_C$ when $n\neq 3$ and $\lambda_{C'} = \lambda_C$ when $n= 3$. We prove Theorem~\ref{thm:11} for $|3| = \min\{|3|, |5|, |6|\}$, and for codes with other values of $\min\{|3|, |5|, |6|\}$, we can transform them to an equivalent code with $|3| = \min\{|3|, |5|, |6|\}$.

When $|1|=1$, $w_1 = 1$ we can simplify $\alpha_C^3(i)$ and $\alpha_C^5(i)$ as follows.
By \eqref{eq:3yd1} and \eqref{eq:3yd4},
\begin{equation}\label{eq:3yd5}
  w_3+w_6\geq \frac{|3|+|6|}{2},
\end{equation}
and by \eqref{eq:3yd3} and \eqref{eq:3yd4},
\begin{equation}\label{eq:3yd6}
  w_3-w_6\leq \frac{|3|-|6|-1}{2}.
\end{equation}
By \eqref{eq:y3b}, 
\begin{IEEEeqnarray*}{rCl}
  \IEEEeqnarraymulticol{3}{l}{\sum_{i=1}^d\alpha_C^3(i) \geq  \sum_{i=1}^d |\mc Y_3^B(i)|}   \IEEEyesnumber \IEEEeqnarraynumspace\label{eq:110}\\
  & = & \sum_{(w_3,w_6)\in \mathcal{W}_3} \binom{|3|}{w_3}\binom{|5|}{\frac{|5|-|3|}{2}+w_3}\binom{|6|}{w_6}\\
  & = & \sum_{(w_3',w_6')\in \mathcal{W}_3'}  \binom{|3|}{\frac{|3|-|6|}{2}+w_6'}\binom{|5|}{\frac{|5|-|6|}{2}+w_6'}\binom{|6|}{\frac{|6|-|3|}{2}+w_3'}
  \IEEEyesnumber \IEEEeqnarraynumspace\label{eq:113}
\end{IEEEeqnarray*}
where
\begin{IEEEeqnarray*}{rCl}
  \mathcal{W}_3 & = &  \left\{\substack{\eqref{eq:3yd5}, \eqref{eq:3yd6}, \\ w_6\geq \frac{n+|6|+1}{2}-d ,\\ 0\leq w_3 \leq |3|, 0\leq w_6\leq |6| }\right\} = \left\{\substack{w_3+w_6\geq \frac{|3|+|6|}{2} \\ w_3-w_6\leq \frac{|3|-|6|-1}{2} \\ w_6\geq \frac{n+|6|+1}{2}-d ,\\ 0\leq w_3 \leq |3|, 0\leq w_6\leq |6| }\right\}, \\
  \mathcal{W}_3' & = & \left\{ \substack{w_3'+w_6'\geq \frac{|3|+|6|}{2}, \quad w_3'-w_6' \geq \frac{|3|-|6|+1}{2} \\ w_3'\geq \frac{n+|3|+1}{2}-d,\\ \frac{|3|-|6|}{2}\leq w_3'\leq \frac{|3|+|6|}{2}, \frac{|6|-|3|}{2}\leq w_6'\leq \frac{|3|+|6|}{2}   }\right\},
\end{IEEEeqnarray*}
and  \eqref{eq:113} is
obtained by change of variables $w_3'-\frac{|3|}{2} = w_6 - \frac{|6|}{2}$ and $w_6'-\frac{|6|}{2} = w_3 - \frac{|3|}{2}$.

By \eqref{eq:5yc2}, $w_5 = \frac{|5|-|6|}{2} +w_6$. Substituting $w_5$ into \eqref{eq:5yc4}, we have
\begin{equation}
  w_3 - w_6 \geq {(|3|-|6|)}/{2}+1.  \label{eq:5yc5}
\end{equation}
By \eqref{eq:alpha5i},
\begin{IEEEeqnarray*}{rCl} %
  \sum_{i=0}^{d-1}\alpha_C^5(i) & = & \sum_{(w_3,w_6)\in \mathcal{W}_5} \binom{|3|}{w_3}\binom{|5|}{\frac{|5|-|6|}{2}+w_6}\binom{|6|}{w_6} %
  \IEEEyesnumber\IEEEeqnarraynumspace \label{eq:115}
\end{IEEEeqnarray*}
where
\begin{equation}\label{eq:w5}
  \mathcal{W}_5 = \left\{\substack{\eqref{eq:5yc3}, \eqref{eq:5yc5}, \\ w_3\geq \frac{n+|3|+1}{2}-d,\\ 0\leq w_3 \leq |3|, 0\leq w_6\leq |6|  }   \right\} = \left\{ \substack{w_3 + w_6 \geq \frac{|3|+|6|}{2}+1,  \\ w_3-w_6 \geq \frac{|3|-|6|}{2}+1,\\ w_3\geq \frac{n+|3|+1}{2}-d,\\ 0\leq w_3 \leq |3|, 0\leq w_6\leq |6|  } \right\}.
\end{equation}

When $2\leq n\leq 4$ and $n\neq 3$, we have $|3|=|5|=|6|=1$. It can be verified that $\mc W_5 = \emptyset$ for all $d$, but $\mc W_3'\neq \emptyset$ when $d\geq 2$. Therefore, by Corollary~\ref{cor:1}, $\ls_{\mc C'} > \ls_{\mc C}$. When $n=3,$ we have $|3|=0 $ and $\{|5|, |6|\}=\{0,2\}$. Then $\mc W_3=\mc W_5=\emptyset$. Therefore, by Corollary~\ref{cor:1}, $\ls_{\mc C'} = \ls_{\mc C}$.

When $n>4$, we show that  $\mathcal{W}_5 \subset \mathcal{W}_3'$.
Due to $|3|\leq |6|$, we have $\frac{|3|-|6|}{2}\leq 0$ and $\frac{|3|+|6|}{2}\geq |3|$. For $(w_3,w_6)\in \mathcal{W}_5 $, we have \eqref{eq:5yc3}, \eqref{eq:5yc5} and $0\leq w_3\leq |3|$, which implies $\frac{|6|-|3|}{2}+1\leq w_6\leq \frac{|3|+|6|}{2}-1$. Thus 
\begin{equation*}
\frac{|3|-|6|}{2}\leq w_3\leq \frac{|3|+|6|}{2}, \quad\frac{|6|-|3|}{2}\leq w_6\leq \frac{|3|+|6|}{2} ,
\end{equation*}
showing
$(w_3,w_6)\in \mathcal{W}_3'$.

For $(w_3,w_6)\in \mathcal{W}_5$, we have
\begin{IEEEeqnarray*}{rCl}
w_3-\frac{|3|}{2} &\geq & \max \left(w_6-\frac{|6|}{2}, \frac{|6|}{2}-w_6\right)+1 = \left|w_6-\frac{|6|}{2}\right|+1.
\end{IEEEeqnarray*}
Since $|3|\leq |6|$, we have $w_3-\frac{|3|}{2} \leq \frac{|6|}{2}$.
Based on these two inequalities and by Lemma~\ref{lemma:cli1} in Appendix, 
for $(w_3,w_6)\in \mathcal{W}_5$, we have
\begin{equation*}
  \binom{|3|}{w_3}\binom{|6|}{w_6} =  \binom{|3|}{\frac{|3|}{2} +(w_3-\frac{|3|}{2} )}\binom{|6|}{\frac{|6|}{2} +(w_6-\frac{|6|}{2} )}<  \binom{|3|}{\frac{|3|}{2}+(w_6-\frac{|6|}{2})} \binom{|6|}{\frac{|6|}{2}+(w_3-\frac{|3|}{2})}.
\end{equation*}
Comparing \eqref{eq:115} and \eqref{eq:113},
we obtain $\sum_{i=1}^d\alpha_C^3(i) \geq \sum_{i=0}^{d-1}\alpha_C^5(i)$ for any $d=1,\ldots,n$. 
 When $|6|\geq 1,$ there exists $w_3'=\frac{|3|+|6|}{2}$, $w_6' = \frac{|6|}{2}$ so that $(w_3',w_6')\in \mc W_3'$ when $d\geq \frac{n+1-|6|}{2}$ and thus $\mc W_3'\neq 
\emptyset$ when $d\geq \frac{n+1-|6|}{2}$. Comparing \eqref{eq:115} and \eqref{eq:113},
we obtain $\sum_{i=1}^d\alpha_C^3(i) > \sum_{i=0}^{d-1}\alpha_C^5(i)$  when $d\geq \frac{n+1-|6|}{2}$. When $|6|=0,$ we have $|5|\geq 1$ and by similar verification while exchanging $ \sum_{i=1}^d |\mc Y_3^B(i)| $ with $ \sum_{i=1}^d |\mc Y_3^A(i)|$ in~\eqref{eq:110}, we obtain $\sum_{i=1}^d\alpha_C^3(i) > \sum_{i=0}^{d-1}\alpha_C^5(i)$  when $d\geq \frac{n+1-|5|}{2}$. 
By Corollary~\ref{cor:1}, $\ls_{\mc C'} > \ls_{\mc C}$, proving Theorem~\ref{thm:11}.

\subsection{Class-I Codes with $|3|=0$ or $1$: Proof of Theorem~\ref{thm:301}}
\label{sec:301}

We consider the special case with $|3| = \min\{|3|, |5|, |6|\} =0$ or $1$. 
Theorem~\ref{thm:301} states that $\lambda_{C'} \geq \lambda_C$. We prove Theorem~\ref{thm:301} for this case. For codes with $\min\{|3|, |5|, |6|\}=|5|$ or $|6|$, we can transform them to an equivalent code with $|3| = \min\{|3|, |5|, |6|\}$.

\subsubsection{$C$ is Class-1-a} In this case, $n$ is odd, $|1|$ is odd, $|5|$ and $|6|$ are even, and $|3|=0$, which means $w_3=0$. By~\eqref{eq:5yc1}, $\alpha_C^5(i)=0$ if $i\neq \frac{n-1}{2}$.
So when $d < \frac{n+1}{2}$, 
\begin{equation}
\sum_{i=0}^{d-1} \alpha_C^5(i)=0 , \label{eq:c1a}
\end{equation}
and hence $\sum_{i=1}^d \alpha_C^3(i)\geq \sum_{i=0}^{d-1} \alpha_C^5(i)$.
By \eqref{eq:5yc2} and \eqref{eq:5yc3},
\begin{equation}
  \label{eq:7}
  w_6-w_1\leq \frac{|6|-|1|-3}{2}.
\end{equation}
When $d\geq \frac{n+1}{2}$, 
\begin{IEEEeqnarray}{rCl}
  \sum_{i=0}^{d-1} \alpha_C^5(i) & = & \alpha_{C}^5\left(\frac{n-1}{2}\right) \nonumber \\
  &=&\sum_{\substack{w_1\geq 1, \eqref{eq:5yc3},\eqref{eq:7},\\  \eqref{eq:5yc2}: w_5=w_6-w_1+\frac{|1|+|5|-|6|+1}{2},\\
      \eqref{eq:5yc4}: w_5\leq \frac{|5|}{2}-1}} \binom{|1|-1}{w_1-1}\binom{|5|}{w_5}\binom{|6|}{w_6}\nonumber\\
  &\leq & \sum_{\substack{w_1\geq 1, \eqref{eq:5yc3}, \eqref{eq:7}}} \binom{|1|-1}{w_1-1}\binom{|5|}{\frac{|5|}{2}}\binom{|6|}{w_6}\nonumber\\
  & = & \sum_{\substack{w_1\geq 1,w_6\geq \frac{|6|}{2}+1,\\w_6-w_1\leq \frac{|6|-|1|-3}{2}}} \binom{|1|-1}{w_1-1}\binom{|5|}{\frac{|5|}{2}}\binom{|6|}{w_6}, \IEEEeqnarraynumspace \label{eq:sa_2}
\end{IEEEeqnarray}
Substituting $w_1'=|1|-w_1 + 1$ into \eqref{eq:sa_2}, we obtain
\begin{equation}
  \label{eq:5}
  \alpha_{C}^5\left(\frac{n-1}{2}\right) \leq \sum_{\substack{1\leq w_1'\leq |1|,w_6\geq \frac{|6|}{2}+1,\\ w_1'\leq \frac{|1|+|6|-1}{2}-w_6}} \binom{|1|-1}{w_1'-1}\binom{|5|}{\frac{|5|}{2}}\binom{|6|}{w_6}
\end{equation}

By \eqref{eq:3yd1} and \eqref{eq:3yd4},
\begin{equation}
  \label{eq:8}
  w_6\geq \frac{|6|}{2},
\end{equation}
and by \eqref{eq:3yd3} and \eqref{eq:3yd4}
\begin{equation}
  \label{eq:9}
  w_1-w_6\leq \frac{|1|-|6|}{2},
\end{equation}
which is equivalent to $w_1-w_6\leq \frac{|1|-|6|-1}{2}$ as $|1|-|6|$ is odd.
When $d\geq \frac{n+1}{2}$, we further have
\begin{IEEEeqnarray}{rCl}
  \sum_{i=1}^d \alpha_C^3(i)&\geq& \sum_{i=1}^{\frac{n+1}{2}} \alpha_C^3(i)\nonumber \\
  &\geq & \sum_{i=1}^{\frac{n+1}{2}}|\mc Y_3^B(i)|\nonumber\\
  &=&   \sum_{\substack{w_1\geq 1, \eqref{eq:8}, \eqref{eq:9},\\
      \eqref{eq:3yd4}: w_1-w_6\leq \frac{|1|-|6|+1}{2}}} \binom{|1|-1}{w_1-1} \binom{|5|}{\frac{|5|}{2}}\binom{|6|}{ w_6}\nonumber\\
  &=&  \sum_{\substack{w_1\geq 1, w_6\geq \frac{|6|}{2},\\
      w_1\leq \frac{|1|-|6|-1}{2}+w_6}} \binom{|1|-1}{w_1-1} \binom{|5|}{\frac{|5|}{2}}\binom{|6|}{ w_6}, \label{eq:a3xs}
\end{IEEEeqnarray}

As $\frac{|1|-|6|-1}{2}+w_6 \geq \frac{|1|+|6|-1}{2}-w_6$ when $w_6\geq \frac{|6|}{2}$, comparing the RHS' of \eqref{eq:5} and \eqref{eq:a3xs}, 
we have $\sum_{i=1}^d \alpha_C^3(i)\geq \sum_{i=0}^{d-1} \alpha_C^5(i)$ for $d\geq \frac{n+1}{2}$. By Corollary~\ref{cor:1}, $\ls_{\mc C'} \geq \ls_{\mc C}$. %

\subsubsection{$C$ is Class-1-b} In this case, $n$ is even, $|1|$ is odd, $|5|$ and $|6|$ are odd, and $|3|=1$, which means $w_3=0\text{ or }1$.
When $i=\frac{n}{2}-1$, by \eqref{eq:5yc1}, $w_3=1$, and by \eqref{eq:5yc2} and \eqref{eq:5yc4}
\begin{equation}
  \label{eq:12}
  w_6-w_1\leq \frac{|6|-|1|-2}{2}.
\end{equation}
Hence
\begin{IEEEeqnarray}{rCl}
  \alpha_{C}^5\left(\frac{n}{2}-1\right)
  &=&\sum_{\substack{w_1\geq 1, w_6\geq \frac{|6|+1}{2}, \eqref{eq:12},\\  w_5=w_6-w_1+\frac{|1|+|5|-|6|+1}{2},\\
			w_5\leq \frac{|5|-1}{2}}} \binom{|1|-1}{w_1-1}\binom{|5|}{w_5}\binom{|6|}{w_6}\nonumber\\
	&\leq & \sum_{\substack{w_1\geq 1, w_6\geq \frac{|6|+1}{2}, \eqref{eq:12}}} \binom{|1|-1}{w_1-1}\binom{|5|}{\frac{|5|+1}{2}}\binom{|6|}{w_6}\nonumber\\
	&=& \sum_{\substack{w_1\geq 1, w_6\geq \frac{|6|+1}{2},\\w_6-w_1\leq \frac{|6|-|1|-2}{2}}} \binom{|1|-1}{w_1-1}\binom{|5|}{\frac{|5|+1}{2}}\binom{|6|}{w_6},\label{eq:of1} \\
        & = & \sum_{\substack{|1|\geq w_1'\geq 1, w_6\geq \frac{|6|+1}{2},\\w_1'\leq \frac{|6|+|1|}{2}-w_6}} \binom{|1|-1}{w_1'-1}\binom{|5|}{\frac{|5|+1}{2}}\binom{|6|}{w_6}.\IEEEeqnarraynumspace \label{eq:of2}
\end{IEEEeqnarray}
where \eqref{eq:of2} is obtained by
substituting $w_1'=|1|-w_1+1$ into \eqref{eq:of1}.

When $i=\frac{n}{2}$, by \eqref{eq:5yc1}, $w_3=0$, and by 
\eqref{eq:5yc2} and \eqref{eq:5yc4}
\begin{equation}
  \label{eq:11}
  w_6-w_1\leq \frac{|6|-|1|}{2}-2.
\end{equation}
Hence
\begin{IEEEeqnarray}{rCl}
  \alpha_{C}^5\left(\frac{n}{2}\right)&=&\sum_{\substack{w_1\geq 1, w_6\geq \frac{|6|+3}{2}, \eqref{eq:11},\\  w_5=w_6-w_1+\frac{|1|+|5|-|6|+1}{2},\\
      w_5\leq \frac{|5|-3}{2}}} \binom{|1|-1}{w_1-1}\binom{|5|}{w_5}\binom{|6|}{w_6}\nonumber\\
  &\leq & \sum_{\substack{w_1\geq 1, w_6\geq \frac{|6|+3}{2}, \eqref{eq:11}}} \binom{|1|-1}{w_1-1}\binom{|5|}{\frac{|5|-1}{2}}\binom{|6|}{w_6}\nonumber\\
  &=&  \sum_{\substack{w_1\geq 1, w_6\geq \frac{|6|+3}{2},\\w_6-w_1\leq \frac{|6|-|1|}{2}-2}} \binom{|1|-1}{w_1-1}\binom{|5|}{\frac{|5|-1}{2}}\binom{|6|}{w_6}  \label{eq:of3} \\
  & = & \sum_{\substack{|1|\geq w_1'\geq 1, w_6\geq \frac{|6|+3}{2},\\w_1'\leq -w_6+ \frac{|6|+|1|-2}{2}}} \binom{|1|-1}{w_1'-1}\binom{|5|}{\frac{|5|-1}{2}}\binom{|6|}{w_6} . \label{eq:of4}
\end{IEEEeqnarray}
where \eqref{eq:of4} is obtained by substituting $w_1'=|1|-w_1+1$ into \eqref{eq:of3}.

When $w_3=1$, by \eqref{eq:3yd1} and \eqref{eq:3yd4},
\begin{equation}
  \label{eq:8sdss}
  w_6\geq \frac{|6|-1}{2},
\end{equation}
and by \eqref{eq:3yd3} and \eqref{eq:3yd4},
\begin{equation}
  \label{eq:120sl}
  w_1-w_6\leq \frac{|1|-|6|-1}{2}.
\end{equation}
Similar as~\eqref{eq:y3b}, we have
\begin{IEEEeqnarray}{rCl}
&&	\left|\{w_3=1\}\cap \left(\mathop\cup_{i\leq \frac{n}{2}} \mc Y_3^B(i)\right)\right|\nonumber\\
&=&  \sum_{\substack{w_1\geq 1, \eqref{eq:8sdss}, \eqref{eq:120sl},\\ w_1-w_6\leq \frac{n}{2}-\frac{1+|5|}{2}-|6|}} \binom{|1|-1}{w_1-1} \binom{|5|}{\frac{|5|+1}{2}}\binom{|6|}{ w_6}  \nonumber\\
&=&  \sum_{\substack{w_1\geq 1,w_6\geq \frac{|6|-1}{2},\\ w_1\leq w_6+\frac{|1|-|6|-1}{2}}} \binom{|1|-1}{w_1-1} \binom{|5|}{\frac{|5|+1}{2}}\binom{|6|}{ w_6}  \label{eq:sa_5}
\end{IEEEeqnarray}
where \eqref{eq:sa_5} follows that $\frac{|1|-|6|-1}{2}\leq \frac{n}{2}-\frac{1+|5|}{2}-|6|$. %
Since $ w_6+\frac{|1|-|6|-1}{2}\geq -w_6 + \frac{|6|+|1|}{2}$ when $w_6\geq \frac{|6|+1}{2}$, comparing the RHS' of \eqref{eq:of2} and \eqref{eq:sa_5}, we get
\begin{equation}
	\alpha_{C}^5\left(\frac{n}{2}-1\right) \leq \left|\{w_3=1\}\cap \left(\mathop\cup_{i\leq \frac{n}{2}} \mc Y_3^B(i)\right)\right|. \label{eq:of5}
\end{equation}

When $w_3=0$, by \eqref{eq:3yd1} and \eqref{eq:3yd4},
\begin{equation}\label{eq:ssisi}
  w_6\geq \frac{|6|+1}{2},
\end{equation}
and by \eqref{eq:3yd3} and \eqref{eq:3yd4},
\begin{equation}
  \label{eq:8s8fs}
  w_1-w_6\leq \frac{|1|-|6|+1}{2}.
\end{equation}
Similar~\eqref{eq:y3b}, we have
\begin{IEEEeqnarray}{rCl}
  &&	\left|\{w_3=0\}\cap \left(\mathop\cup_{i\leq \frac{n}{2}+1} \mc Y_3^B(i)\right)\right|\nonumber\\
  &=&  \sum_{\substack{w_1\geq 1,\eqref{eq:ssisi}, \eqref{eq:8s8fs},\\ w_1-w_6\leq \frac{n}{2}+1-\frac{1+|5|}{2}-|6|}} \binom{|1|-1}{w_1-1} \binom{|5|}{\frac{|5|-1}{2}}\binom{|6|}{ w_6}  \nonumber\\
  &=&  \sum_{\substack{w_1\geq 1,w_6:w_6\geq \frac{|6|+1}{2},\\
      w_1\leq w_6+\frac{|1|-|6|+1}{2}}} \binom{|1|-1}{w_1-1} \binom{|5|}{\frac{|5|-1}{2}}\binom{|6|}{ w_6} , \label{eq:sa_6}
\end{IEEEeqnarray}
where \eqref{eq:sa_6} follows that $\frac{|1|-|6|+1}{2} \leq \frac{n}{2}+1-\frac{1+|5|}{2}-|6|$. %
Since $w_6+\frac{|1|-|6|+1}{2}\geq -w_6+ \frac{|6|+|1|-2}{2}$ when $w_6\geq \frac{|6|+1}{2}$,  comparing the RHS' of \eqref{eq:of4} and \eqref{eq:sa_6}, we get 
\begin{equation}
\alpha_{C}^5\left(\frac{n}{2}\right)\leq \left|\{w_3=0\}\cap \left(\mathop\cup_{i\leq \frac{n}{2}+1} \mc Y_3^B(i)\right)\right|. \label{eq:of6}
\end{equation}

When $d < \frac{n}{2}$, $\sum_{i=0}^{d-1} \alpha_C^5(i)=0 \leq \sum_{i=1}^d \alpha_C^3(i)$.
When $d=\frac{n}{2}$, by \eqref{eq:of5},
\begin{IEEEeqnarray*}{rCl}
	\sum_{i=0}^{d-1}\alpha_C^5(i)&=&\alpha_{C}^5\left(\frac{n}{2}-1\right) \leq \left|\{w_3=1\}\cap \left(\mathop\cup_{i\leq \frac{n}{2}} \mc Y_3^B(i)\right)\right| \leq \sum_{i=1}^{d}\alpha^3_i.
\end{IEEEeqnarray*}
When $d\geq \frac{n}{2}+1$, by \eqref{eq:of5} and \eqref{eq:of6},
\begin{IEEEeqnarray}{rCl}
	\sum_{i=0}^{d-1}\alpha_C^5(i)&=&\alpha_{C}^5\left(\frac{n}{2}-1\right)+\alpha_{C}^5\left(\frac{n}{2}\right)\nonumber\\
	&\leq& \left|\{w_3=1\}\cap \left(\mathop\cup_{i\leq \frac{n}{2}} \mc Y_3^B(i)\right)\right| +  %
        \left|\{w_3=0\}\cap \left(\mathop\cup_{i\leq \frac{n}{2}+1} \mc Y_3^B(i)\right)\right| \nonumber\\
	& \leq& \sum_{i=1}^{d}\alpha^3_i.\nonumber
\end{IEEEeqnarray}
Thus we have $\sum_{i=1}^d \alpha_C^3(i)\geq \sum_{i=0}^{d-1} \alpha_C^5(i)$ for $1\leq d \leq n$. By Corollary~\ref{cor:1}, $\ls_{\mc C'} \geq \ls_{\mc C}$. %

\subsection{Algorithm for Verifying Optimal Codes: Proof of Theorem~\ref{thm:18}}
  \label{sec:comp}
  \begin{algorithm}[t]
  	\SetKwInOut{Input}{input}\SetKwInOut{Output}{output}
  	\SetAlgoLined
  	\Input{$n$}
  	\Output{$a$ (If $a=-1$, Theorem~\ref{thm:18} holds with blocklength $n$.) }
  	Initialize $ a=-1 $\;        
  		\For{$ n_1=3,5,7,\ldots$ and $n_1\leq n$}{
  			\For{$ n_3=2,3,\ldots,\left \lfloor \frac{n-n_1}{3} \right \rfloor $}{
  				\For{$ n_5=n_3,n_3+2,\ldots$ and $n_5 \leq \left \lfloor \frac{n-n_1-n_3}{2} \right \rfloor $}{
  					$n_6 =n-n_1-n_3-n_5$;
  					
  					\If{ $n_5+n_6$ is even }{
  						Compute $\alpha_C^3(i)$ and $\alpha_C^5(i)$, $i=0,\ldots,n$ for code $C$ with $ |1|_C=n_1 $, $|3|_C=n_3$, $|5|_C=n_5$, $|6|_C=n_6 $\;
  						\If{$ \sum_{i=1}^{d} \alpha_{i}^3< \sum_{i=0}^{d-1} \alpha_C^5(i)$ for some $ d\in\{1,...,n\} $}{$ a=1 $;
  							break;}
  					}	
  				}
  		}}
  	\caption{Verify Theorem~\ref{thm:18} with blocklength $n$}\label{alg:1}
  \end{algorithm}
  To prove Theorem~\ref{thm:18}, we give an algorithm (see the pseudo-code in Algorithm~\ref{alg:1}) which checks whether all Class-I $ (n,4) $ code $ C $  with $|1|\geq 3$ are not optimal.
  In the algorithm, it only verify the Class-I $ (n,4) $ code $ C $  with $|1|\geq 3$ 
  and $ |3|_C\leq |5|_C\leq |6|_C$, and compares $C$ with $C'$ obtained by replacing one column $\bspan{1}$ of $C$ to $\bspan{3}$. 
  Other $ (n,4) $ Class-I codes with $|1|\geq 3$ can be converted to ones of the above type by flipping columns and interchanging rows, and hence do  not need to be checked again.
  The algorithm calculates $\alpha_{C}^3(i)$ and $\alpha_{C}^5(i)$ exactly using the formulae in \S\ref{sec:35}.
  For a give blocklength $n$, if for each code $C$ checked by the algorithm we have $ \sum_{i=1}^{d} \alpha_{C}^3(i) \geq \sum_{i=0}^{d-1} \alpha_{C}^5(i)$ for $d=1,\ldots,n$, then by Corollary~\ref{cor:1}, $\lambda_C \leq \lambda_{C'}$. %
 To prove Theorem~\ref{thm:18}, we evaluate Algorithm~\ref{alg:1} and get output $-1$ for $n$ up to $300$.
  The total number of types of Class-I codes to evaluate is $O(n^3)$. For each type, there are less than $2n$ values $\alpha_C^3(i)$ or $\alpha_C^5(i)$ to evaluate, each of which has complexity $O(n^2)$. Therefore,
  the complexity of the algorithm is  $O(n^6)$ .

\section{Concluding Remarks}
\label{sec:remark}

In this paper, we derived a technique for comparing the ML decoding performance of two $(n,4)$ codes. We use this technique for two cases: i) codes with differences in one column and ii) codes with differences in two columns.  The code comparison results obtained from these two cases can help us to derive many results about the optimal $(n,4)$ codes. We characterized all the optimal codes for $n$ from $2$ to $300$. The optimal codes obtained are all equivalent to  linear codes except for $n=3$,  where there exist nonlinear optimal codes that are not cosets of the linear codes. 

Technically, we could study larger values of $n$ if more computation costs are paid. But it is more attractive to show analytically that whether all the optimal codes are equivalent to linear codes when $n>300$.
Our technique has induced various ways towards solving the problem that may deserve further study. One is to verify whether $ \sum_{i=1}^{d} \alpha_{C}^3(i) \geq \sum_{i=0}^{d-1} \alpha_{C}^5(i)$ for $d=1,\ldots,n$, where the formulae of $\alpha_{C}^3(i)$ and $\alpha_{C}^5(i)$ are given in \S\ref{sec:35}.
If the inequalities can be confirmed in general for $n>300$, all the optimal codes are equivalent to linear codes universally when $n>300$.
If some of the inequalities do not hold, Theorem~\ref{the:1} can be further applied. 
Towards solving the optimal code problem completely, another way is to
apply our code comparison technique to other differences of two
$(n,4)$ codes.

The general principle of our code comparison  technique can be applied to codes with more than $4$ codewords, where the crucial part is to find a desired partition of $\{0,1\}^n$.

\appendix[A Binomial Inequality]\label{sec:proof}

We present a refined version of an inequality in \cite[Claim 47]{chen2013optimal}.
\begin{lemma}\label{lemma:cli1}
	Let $n_1>n_2>0$ be two integers with the same parity and let $v_1$ and $v_2$ be integers or half integers such that $\frac{n_2}{2}\geq |v_1|>|v_2|\geq 0$ and $n_1/2 + v_1$ and $n_1/2+v_2$ are integers. Then
	\begin{align*}
		\binom{n_1}{\frac{n_1}{2}+v_1}\binom{n_2}{\frac{n_2}{2}+v_2} > 	\binom{n_1}{\frac{n_1}{2}+v_2}\binom{n_2}{\frac{n_2}{2}+v_1} .
	\end{align*}
\end{lemma}
\begin{IEEEproof}
  As $\binom{a}{b}=\binom{a}{a-b}$, we only need to prove the case with $v_1>v_2>0$.
  The lemma can be proved by expanding the binomial terms:
  \begin{IEEEeqnarray*}{rCl}
    \frac{\binom{n_1}{\frac{n_1}{2}+v_1}\binom{n_2}{\frac{n_2}{2}+v_2}}{\binom{n_1}{\frac{n_1}{2}+v_2}\binom{n_2}{\frac{n_2}{2}+v_1}}
    &=&\frac{\frac{n_1!}{ (\frac{n_1}{2}+v_1)! (\frac{n_1}{2}- v_1)! } \cdot \frac{n_2!}{ (\frac{n_2}{2}+v_2)! (\frac{n_2}{2}- v_2)! }}{\frac{n_1!}{ (\frac{n_1}{2}+v_2)! (\frac{n_1}{2}- v_2)! } \cdot \frac{n_2!}{ (\frac{n_2}{2}+v_1)! (\frac{n_2}{2}- v_1)! }}\\
    &=& \frac{  (\frac{n_1}{2}+v_2)! (\frac{n_1}{2}- v_2)!  (\frac{n_2}{2}+v_1)! (\frac{n_2}{2}- v_1)!   }
    { (\frac{n_1}{2}+v_1)! (\frac{n_1}{2}- v_1)!  (\frac{n_2}{2}+v_2)! (\frac{n_2}{2}- v_2)!}\\
    &=& \frac{(\frac{n_1}{2}- v_1+1)\cdots (\frac{n_1}{2}- v_2) \cdot  (\frac{n_2}{2}+v_2+1) \cdots  (\frac{n_2}{2}+v_1)    }{ (\frac{n_1}{2}+v_2+1) \cdots  (\frac{n_1}{2}+v_1) \cdot  (\frac{n_2}{2}- v_1+1)\cdots (\frac{n_2}{2}- v_2) }\\
    &=& \frac{\frac{\frac{n_1}{2}- v_1+1}{\frac{n_1}{2}+v_2+1}\cdots \frac{\frac{n_1}{2}- v_2}{\frac{n_1}{2}+v_1}      }{\frac{\frac{n_2}{2}- v_1+1}{\frac{n_2}{2}+v_2+1}\cdots \frac{\frac{n_2}{2}- v_2}{\frac{n_2}{2}+v_1}   } \\
    &>&1,
  \end{IEEEeqnarray*}
  where the last inequality holds since for $0<a<b$ and $\epsilon > 0, $
  $\frac{a}{b} > \frac{a-\epsilon}{b-\epsilon}$.
\end{IEEEproof}


\begin{thebibliography}{10}
\providecommand{\url}[1]{#1}
\csname url@samestyle\endcsname
\providecommand{\newblock}{\relax}
\providecommand{\bibinfo}[2]{#2}
\providecommand{\BIBentrySTDinterwordspacing}{\spaceskip=0pt\relax}
\providecommand{\BIBentryALTinterwordstretchfactor}{4}
\providecommand{\BIBentryALTinterwordspacing}{\spaceskip=\fontdimen2\font plus
\BIBentryALTinterwordstretchfactor\fontdimen3\font minus
  \fontdimen4\font\relax}
\providecommand{\BIBforeignlanguage}[2]{{%
\expandafter\ifx\csname l@#1\endcsname\relax
\typeout{** WARNING: IEEEtran.bst: No hyphenation pattern has been}%
\typeout{** loaded for the language `#1'. Using the pattern for}%
\typeout{** the default language instead.}%
\else
\language=\csname l@#1\endcsname
\fi
#2}}
\providecommand{\BIBdecl}{\relax}
\BIBdecl

\bibitem{shannon1948mathematical}
C.~E. Shannon, ``A mathematical theory of communication,'' \emph{Bell system
  technical journal}, vol.~27, no.~3, pp. 379--423, 1948.

\bibitem{Arikan09}
E.~{Arikan}, ``Channel polarization: A method for constructing
  capacity-achieving codes for symmetric binary-input memoryless channels,''
  \emph{IEEE Transactions on Information Theory}, vol.~55, no.~7, pp.
  3051--3073, 2009.

\bibitem{gallager1962low}
R.~Gallager, ``Low-density parity-check codes,'' \emph{IRE Transactions on
  information theory}, vol.~8, no.~1, pp. 21--28, 1962.

\bibitem{mitchell2015spatially}
D.~G. Mitchell, M.~Lentmaier, and D.~J. Costello, ``Spatially coupled ldpc
  codes constructed from protographs,'' \emph{IEEE Transactions on Information
  Theory}, vol.~61, no.~9, pp. 4866--4889, 2015.

\bibitem{di2002finite}
C.~Di, D.~Proietti, I.~E. Telatar, T.~J. Richardson, and R.~L. Urbanke,
  ``Finite-length analysis of low-density parity-check codes on the binary
  erasure channel,'' \emph{IEEE Transactions on Information theory}, vol.~48,
  no.~6, pp. 1570--1579, 2002.

\bibitem{eslami2013finite}
A.~Eslami and H.~Pishro-Nik, ``On finite-length performance of polar codes:
  stopping sets, error floor, and concatenated design,'' \emph{IEEE
  Transactions on communications}, vol.~61, no.~3, pp. 919--929, 2013.

\bibitem{mondelli2014polar}
M.~Mondelli, S.~H. Hassani, and R.~L. Urbanke, ``From polar to {R}eed-{M}uller
  codes: A technique to improve the finite-length performance,'' \emph{IEEE
  Transactions on Communications}, vol.~62, no.~9, pp. 3084--3091, 2014.

\bibitem{olmos2015scaling}
P.~M. Olmos and R.~L. Urbanke, ``A scaling law to predict the finite-length
  performance of spatially-coupled {LDPC} codes,'' \emph{IEEE Transactions on
  Information Theory}, vol.~61, no.~6, pp. 3164--3184, 2015.

\bibitem{gaudio2017performance}
L.~Gaudio, T.~Ninacs, T.~Jerkovits, and G.~Liva, ``On the performance of short
  tail-biting convolutional codes for ultra-reliable communications,'' in
  \emph{SCC 2017; 11th International ITG Conference on Systems, Communications
  and Coding}.\hskip 1em plus 0.5em minus 0.4em\relax VDE, 2017, pp. 1--6.

\bibitem{cheng2021bch}
J.~Cheng and L.~Chen, ``{BCH} based {U-UV} codes and its decoding,'' in
  \emph{2021 IEEE International Symposium on Information Theory (ISIT)}.\hskip
  1em plus 0.5em minus 0.4em\relax IEEE, 2021, pp. 1433--1438.

\bibitem{valembois2004sphere}
A.~Valembois and M.~P. Fossorier, ``Sphere-packing bounds revisited for
  moderate block lengths,'' \emph{IEEE Transactions on Information Theory},
  vol.~50, no.~12, pp. 2998--3014, 2004.

\bibitem{wiechman2008improved}
G.~Wiechman and I.~Sason, ``An improved sphere-packing bound for finite-length
  codes over symmetric memoryless channels,'' \emph{IEEE Transactions on
  Information Theory}, vol.~54, no.~5, pp. 1962--1990, 2008.

\bibitem{polyanskiy2010channel}
Y.~Polyanskiy, H.~V. Poor, and S.~Verd{\'u}, ``Channel coding rate in the
  finite blocklength regime,'' \emph{IEEE Transactions on Information Theory},
  vol.~56, no.~5, pp. 2307--2359, 2010.

\bibitem{Slepian1956A}
D.~Slepian, ``A class of binary signaling alphabets,'' \emph{Bell System
  Technical Journal}, vol.~35, no.~1, pp. 203--234, 1956.

\bibitem{fontaine1959group}
A.~Fontaine and W.~Peterson, ``Group code equivalence and optimum codes,''
  \emph{IRE Transactions on Information Theory}, vol.~5, no.~5, pp. 60--70,
  1959.

\bibitem{wagner1966search}
T.~J. Wagner, ``A search technique for quasi-perfect codes,'' \emph{Information
  and Control}, vol.~9, no.~1, pp. 94--99, 1966.

\bibitem{tokura1967search}
N.~Tokura, K.~Taniguchi, and T.~Kasami, ``A search procedure for finding
  optimum group codes for the binary symmetric channel,'' \emph{IEEE
  Transactions on Information Theory}, vol.~13, no.~4, pp. 587--594, 1967.

\bibitem{cordaro1967optimum}
J.~Cordaro and T.~Wagner, ``Optimum (n, 2) codes for small values of channel
  error probability (corresp.),'' \emph{IEEE Transactions on Information
  Theory}, vol.~13, no.~2, pp. 349--350, 1967.

\bibitem{peterson72}
W.~W. Peterson and E.~J. Weldon~Jr., \emph{Error-Correcting Codes}.\hskip 1em
  plus 0.5em minus 0.4em\relax MIT Press, 1972.

\bibitem{chen2013equidistant}
P.-N. Chen, H.-Y. Lin, and S.~M. Moser, ``Equidistant codes meeting the
  {P}lotkin bound are not optimal on the binary symmetric channel,'' in
  \emph{2013 IEEE International Symposium on Information Theory (ISIT)}.\hskip
  1em plus 0.5em minus 0.4em\relax IEEE, 2013, pp. 3015--3019.

\bibitem{chen2013optimal}
------, ``Optimal ultrasmall block-codes for binary discrete memoryless
  channels,'' \emph{IEEE Transactions on Information Theory}, vol.~59, no.~11,
  pp. 7346--7378, 2013.

\bibitem{klove2006binary}
T.~Kl{\o}ve, ``Binary linear codes that are optimal for error correction,'' in
  \emph{General Theory of Information Transfer and Combinatorics}.\hskip 1em
  plus 0.5em minus 0.4em\relax Springer, 2006, pp. 1081--1083.

\bibitem{vazquez2016bayesian}
G.~Vazquez-Vilar, A.~T. Campo, A.~G. i~F{\`a}bregas, and A.~Martinez,
  ``Bayesian $ m $-ary hypothesis testing: The meta-converse and
  {V}erd{\'u}-{H}an bounds are tight,'' \emph{IEEE Transactions on Information
  Theory}, vol.~62, no.~5, pp. 2324--2333, 2016.

\bibitem{lin2018weak}
H.-Y. Lin, S.~M. Moser, and P.-N. Chen, ``Weak flip codes and their optimality
  on the binary erasure channel,'' \emph{IEEE Transactions on Information
  Theory}, vol.~64, no.~7, pp. 5191--5218, 2018.

\end{thebibliography}
\end{document}